\documentclass[a4paper]{article}

\usepackage{environ}
\newcommand{\acksection}{\section*{Acknowledgments}}
\NewEnviron{ack}{%
  \acksection
  \BODY
}
\usepackage[margin=1.3in]{geometry}

\usepackage[round]{natbib}
\usepackage[utf8]{inputenc} %
\usepackage[T1]{fontenc}    %
\usepackage{hyperref}       %
\usepackage{url}            %
\usepackage{booktabs}       %
\usepackage{amsfonts}       %
\usepackage{nicefrac}       %
\usepackage{microtype}      %
\usepackage{xcolor}         %

\usepackage{amsmath,amsthm,amssymb}
\usepackage{dsfont}

\usepackage[noabbrev,capitalize]{cleveref}

\usepackage{bbm}

\usepackage{bm}
\usepackage{paralist}
\usepackage{enumerate}
\usepackage[inline]{enumitem}
\setlist[enumerate,1]{label={\roman*)}}

\usepackage{mathtools}

\usepackage{graphicx}

\usepackage{tikz}
\usetikzlibrary{arrows}%
\tikzstyle{graphnodesmall} = [circle,draw=black,minimum size=20pt,text centered,inner sep=0.2mm]
\tikzstyle{observedsmall} = [graphnodesmall,fill=white,text=black]
\tikzstyle{unobservedsmall} = [graphnodesmall,fill=white,text=black,style=dashed]

\usetikzlibrary{shapes.geometric}
\usetikzlibrary{arrows.meta}

\usepackage{ifthen} %

\usepackage{caption}
\usepackage{subcaption}

\usepackage{algorithm, algorithmicx}
\usepackage[noend]{algpseudocode}

\usepackage{setspace}
\let\Algorithm\algorithm
\renewcommand\algorithm[1][]{\Algorithm[#1]\setstretch{1.0}}

\usepackage{centernot}

\usepackage{booktabs}
\newcommand{\R}{\mathbb{R}}
\newcommand{\N}{\mathbb{N}}
\newcommand{\Z}{\mathbb{Z}}

\newcommand{\convD}{\stackrel{\mathcal{D}}{\longrightarrow}}
\newcommand{\convP}{\stackrel{\mathcal{P}}{\longrightarrow}}

\newcommand{\PA}{\operatorname{PA}}
\newcommand{\scm}{\mathcal{S}}

\newcommand{\E}{\mathbb{E}}
\newcommand{\VAR}{\mathbb{V}}
\renewcommand{\P}{\mathbb{P}}
\newcommand{\DO}{\operatorname{do}}

\newcommand{\mat}[1]{\begin{pmatrix} #1 \end{pmatrix}}
\newcommand{\ind}[1]{\mathbbm{1}_{#1}} %
\newcommand{\indep}{\mbox{${}\perp\mkern-11mu\perp{}$}}
\renewcommand{\epsilon}{\varepsilon}
\DeclareMathOperator*{\argmin}{{\arg\min}}
\newcommand{\cP}{\mathcal{P}} %
\newcommand{\cQ}{\mathcal{Q}} %
\newcommand{\cX}{\mathcal{X}} %
\newcommand{\cZ}{\mathcal{Z}} %
\newcommand{\bX}{\mathbf{X}} %
\newcommand{\bZ}{\mathbf{Z}} %
\DeclareMathAlphabet\mathbfcal{OMS}{cmsy}{b}{n} %

\newcommand{\widebar}[1]{\bar{#1}} %

\newcommand\james[1]{{\color{red}James: #1}}
\newcommand\note[1]{{\color{red}note: #1}}
\newcommand\niklas[1]{{\color{blue}Niklas: #1}}
\newcommand\rmnik[1]{{\color{gray}Remove(Niklas): #1}}
\newcommand\jonas[1]{{\color{orange}Jonas: #1}}
\newcommand\rmj[1]{{\color{gray}Remove(Jonas): #1}}
\newcommand\Rmjs[1]{{\color{gray}Remove(James): #1}}
\newcommand\nikolaj[1]{{\color[HTML]{538f00}{\tiny NT:}#1}}
\newcommand\Rmn[1]{{\color{gray}{\small RM[NT]: #1}}}
\newcommand\rmn[1]{\Rmn{#1}}
\newcommand\new[1]{{#1}}
\newcommand\old[1]{} %
\newtheorem{theorem}{Theorem}

\theoremstyle{plain}
\newtheorem{lemma}{Lemma}
\newtheorem{corollary}{Corollary}
\newtheorem{proposition}{Proposition}

\newtheorem*{remark}{Remark}

\crefname{assumption}{}{}

\crefname{equation}{}{}

\newlist{assumpenum}{enumerate}{1} %
\setlist[assumpenum]{label={(A\arabic*)}}
\crefalias{assumpenumi}{assumption}

\newlist{alt-assumpenum}{enumerate}{1} %
\setlist[alt-assumpenum]{label={(A\arabic*')}}
\crefalias{alt-assumpenumi}{assumption}

\title{Statistical Testing under Distributional Shifts
}

\author{Nikolaj Thams\footnote{\{thams, ss, np, jonas.peters\}@math.ku.dk},\, Sorawit Saengkyongam\footnotemark[\value{footnote}],\, Niklas Pfister\footnotemark[\value{footnote}],\, and Jonas Peters\footnotemark[\value{footnote}]}
\date{University of Copenhagen, Denmark}

\begin{document}

\maketitle

\begin{abstract}
    In this work, we introduce statistical testing under distributional shifts. We are interested in the hypothesis $P^* \in H_0$ for a target distribution $P^*$, but observe data from a different distribution $Q^*$. We assume that $P^*$ is related to $Q^*$ through a known shift $\tau$ and formally introduce hypothesis testing in this setting. We propose a general testing procedure that first resamples from the observed data to construct an auxiliary data set and then applies an existing test in the target domain. We prove that if the size of the resample is at most $o(\sqrt{n})$ and the resampling weights are well-behaved, this procedure inherits the pointwise asymptotic level and power from the target test. If the map $\tau$ is estimated from data, we can maintain the above guarantees under mild conditions if the estimation works sufficiently well. We further extend our results to \new{finite sample level,} uniform asymptotic level and a different resampling scheme. Testing under distributional shifts allows us to tackle a diverse set of problems. We argue that it may prove useful in reinforcement learning and covariate shift, we show how it reduces conditional to unconditional independence testing and we provide example applications in causal inference. 
\end{abstract}

\section{Introduction} \label{sec:intro}
Testing scientific hypotheses about an observed data generating mechanism is an important part of many areas of empirical research and is relevant for almost all types of data. In statistics, the data generating mechanism is described by a distribution $P^*$ and the process of testing a hypothesis corresponds to testing whether $P^*$ belongs to a subclass of distributions $H_0$. 
In practice, observations from $P^*$, for which we want to test the hypothesis $P^*\in H_0$, may not always be available. For instance, sampling from $P^*$ may be unethical if it corresponds to assigning patients to a certain treatment. 
$P^*$ may also represent 
the response to a policy that a government is considering to introduce. Yet, in many cases, one may still have data from a different, but related, distribution $Q^*$.
In the examples above, this could be data from an observational study or under the policies currently deployed by the government. 

Although specialized solutions exist for many such problems, there is no general method for tackling them.
In this paper, we 
aim to analyze the above testing task from a general point of view. We 
assume that a distributional shift $\tau: Q \mapsto P$ is known
and, using data from $Q^*$, aim to test the hypothesis $P^* = \tau(Q^*) \in H_0$.
We propose the following  general framework.
We resample from $Q^*$ to construct an auxiliary data set mimicking a sample from $P^*$ and then exploit the existence of a test in the target domain. 
Our method does not assume full knowledge of $Q^*$ or $P^*$, but only knowledge of the (potentially unnormalized) ratio $p^*/q^*$, where $q^*$ and $p^*$ are densities of $Q^*$ and $P^*$ respectively. 
If, for example, the shift corresponds to a change in the conditional distribution of a few of the observed variables, one only needs to know these changing conditionals.

Our framework assumes the existence of a test $\varphi$ in the target domain, i.e., a test that could be applied if data from $P^*$ were available. This test $\varphi$ is then applied to a resampled version of the observed data set. Here, we propose a sampling scheme that 
is similar to sampling importance resampling (SIR), proposed by \citet{rubin1987calculation} and \citet{smith1992bayesian}
but generates a distinct sample of size $m$, using weights $r(X_i) \propto q^*(X_i)/p^*(X_i)$.
We prove that this  procedure inherits the pointwise asymptotic properties of the test $\varphi$ if the weights $r$ have finite second moment in $Q^*$, and $m = o(\sqrt{n})$. 
In particular, the procedure holds pointwise asymptotic level if the test $\varphi$ does.
We show that the same can be obtained if $r$ is not known, but can be estimated from data sufficiently well.
The proposed method is easy-to-use and can be applied to any hypothesis test, even if the test is based on a nonlinear test statistic.

Several problems can be cast as hypothesis tests under distributional shifts. This includes hypothesis tests in off-policy evaluation, tests of conditional independence, testing the absence of causal edges through dormant independences \citep{Verma1991,shpitser2008dormant}, that is, testing certain equality constraints in an observed distribution, and problems of covariate shift.  
Our proposed method can be applied to all of these problems. 
For some of them, we are not aware of other methods with theoretical guarantees -- this includes dormant independence testing with continuous variables, off-policy testing with complex hypotheses and model selection under covariate shift with complex scoring functions. The framework also inspired a novel method for causal discovery that exploits knowledge of a single causal conditional.
For some of the above problems, however, more specialized solutions exist, and as such, the proposed testing procedure relates to a line of related work. 

Ratios of densities have been applied in the reinforcement learning literature \citep[e.g.][]{Sutton1998}, where inference in $P^*$ using data from $Q^*$ is known as off-policy prediction \citep{precup2001off}.
One can estimate the expectation of $X$ under $P^*$ using importance sampling, (IS),
that is, as weighted averages using weights $r(X)$, possibly truncated to decrease the variance of the estimation
\citep{precup2001off,mahmood2014weighted}. 
An approach based on resampling was proposed by \citet{schlegel2019importance} to predict expectations in $P^*$.  However, they consider the size of the resample fixed and do not consider statistical testing. \cite{thomas2015high} propose bootstrap confidence intervals for off-policy prediction based on important weighted returns. \cite{hao2021bootstrapping} present a bootstrapping approach with Fitted Q-Evaluation (FQE) for off-policy statistical inference and demonstrate its distributional consistency guarantee. 

In the causal inference literature, inverse probability weighting (IPW) 
can be used to adjust for confounding or selection bias in data \citep[e.g.][]{Horvitz1952,robins2000}.
To estimate the effect of a treatment $X$ on a response $Y$, 
one can weight each observed response $Y_i$ with $1/q^*(X_i|Z_i)$, where $Z$ is an observed confounder. 
For continuous treatments, it has been proposed to
change the numerator to a marginal distribution $p^*(x)$ to stabilize the weights \citep{hernan2006estimating,naimi2014constructing}. 
Both choices of weights appear in our framework, too (e.g., the first one corresponds to a target distribution with $p^*(x|z) \propto 1$).
In general, IS and IPW can only be applied if the population version of the test statistic can be written as a mean of a function of a single observation, such as $\E[Y_i]$ or $\E[f(X_i, Z_i)]$, whereas our approach also applies to test statistics that are functions of the entire sample, which is the case for many tests that go beyond testing moments, such as several independence tests, for example.

SIR sampling schemes were first studied by \citet{rubin1987calculation} and are often used in the context of Bayesian inference \citep{smith1992bayesian}.  \citet{skare2003improved} show that when using weighted resampling with or without replacement, for $n\to\infty$ and fixed $m$, the sample converges towards $m$ i.i.d.\ draws from the target distribution, and provide rates for the convergence. Our work is inspired by these types of results, even though our proofs require different techniques.

\new{Our paper adds to the literature on distributional shifts by considering hypothesis tests in shifted distributions. 
In the context of prediction, distributional shifts, or dataset shifts, have been studied in the machine learning literature both to handle the situations where a marginal covariate distribution changes and when the conditional distribution of label given covariate changes \citep{quinonero2009dataset}. 
If the shift represents a changing marginal distribution and unlabelled samples are available from both training and test environments, \citet{huang2006correcting} propose kernel mean matching, which non-parameterically reweights the training loss to resemble the loss on a target sample.
In settings where a generative model and causal graph is known, \citet{pearl2011transportability,subbaswamy2019preventing} provide graphical criteria under which causal estimates can be `transported' from one distribution $Q$ to a shifted distribution $P$, assuming knowledge of both joint distributions $Q$ and $P$. 
In contrast, we consider statistical testing, and neither assume knowledge of the full causal graph nor availability of samples from the target distribution, but instead knowledge of how the target data differs from the observed data.
}

\new{
This paper contains four main contributions: 
First, we formally define testing under distributional shifts, and define notions such as pointwise and asymptotic level when using observed data to test the hypothesis in the target domain. 
Second, we outline a number of statistical problems, that can be solved by testing under a shift, including conditional independence testing and testing dormant independences. 
Third, we propose methods that enable testing under distributional shifts: both a simple method based on rejection sampling and a resampling scheme that requires fewer assumptions than the rejection sampler. 
Fourth, we provide finite sample and asymptotic guarantees for our proposed resampling scheme; contrary to the existing literature, 
where typically $m$ fixed and $n\to\infty$ has been studied
\citep[e.g.,][]{skare2003improved}, we study the asymptotic behaviour of our resampling test when both $m$ and $n$ approach infinity, and show that under any resampling scheme, the requirement $m=o(\sqrt{n})$ is necessary.
}

The remainder of this work is structured as follows. We formally introduce the framework of testing under distributional shifts in Section~\ref{sec:framewo}. 
Section~\ref{sec:applications} showcases how several problems from different statistical fields fit into this framework. We introduce a general procedure for testing under distributional shifts and provide theoretical results in Section~\ref{sec:main-results}. In Section~\ref{sec:expe} we conduct simulation experiments and \cref{sec:summary_future_work} concludes and discusses future work.

\section{Statistical testing under distributional shifts}\label{sec:framewo}
\begin{figure}
    \centering
    \begin{tikzpicture}[yscale = 0.8]
        \node[ellipse, draw, minimum width=5cm,
        minimum height=3cm, line width=1pt]
        at (0,0) (E1){};
        \node at (0,2.5) (E1text){ $\mathcal{Q}$};
        \node at (0,-3) (E1text){Distributions on observed domain $\mathcal{X}$;};
        \node at (0,-3.6) (E1text){
            $X_1,\ldots,X_n\overset{\text{i.i.d.}}{\sim} Q^*$.};
        \node[circle, draw, fill, minimum width=0.1cm, inner sep=0, color=red] at (0.5, 0.8) (Qstarpt){};
        \node at (0.5, 1.2) (Qstar){ $Q^*$};    

        \node[ellipse, draw, minimum width=5cm,
        minimum height=3cm, line width=1pt]
        at (6,0) (E2){};
        \node at (6,2.5) (E2text){ $\mathcal{P}$};
        \node at (6,-3) (E2text){Distributions on target domain $\mathcal{Z}$;};
        \node[circle, draw, fill, minimum width=0.1cm, inner sep=0, color=red] at (5.5, -1.5) (Pstarpt){};
        \node at (5.9, -1.1) (Pstar){ $P^*=\tau(Q^*)$};
        \node[ellipse, draw, minimum width=1.5cm,
        minimum height=1.5cm, line width=1pt]
        at (5.5,0.8) (C2){};
        \node at (5.5, 0.8) (H0target){ $H_0$};
        \node at (6,-3.6) (E2text){unobserved.};

        \node[ellipse, draw, minimum width=2cm,
        minimum height=1.5cm, line width=1pt]
        at (0.2,-0.8) (C1){};
        \node at (0.2, -0.8) (H0target){ $\tau^{-1}(H_0)$};

        \draw[-Latex, line width=1pt]
        ([shift={(0.25,0.5)}]E1.north east) to [out=30, in=150]
        ([shift={(-0.25,0.5)}]E2.north west);
        \node at (3, 1.5) (tau){\Large $\tau$};
     \end{tikzpicture}
     \caption{Illustration of observed and target domains, $\cQ$ and $\cP$, target hypothesis $H_0 \subseteq \cP$ and pullback hypothesis $\tau^{-1}(H_0)$.}
     \label{fig:framework}
\end{figure}

\subsection{Testing hypotheses in a target distribution}\label{subsec:framework}
Consider a set of distributions $\mathcal{P}$ on a target domain $\mathcal{Z} \subseteq \R^d$ and a null hypothesis $H_0 \subseteq \mathcal{P}$. In hypothesis testing, we are usually given data from a distribution $P^*\in\mathcal{P}$ and want to test whether $P^*\in H_0$. 
In this paper, we consider the problem of testing the same hypothesis but instead of observing data from $P^*$ directly, we assume the data are generated by a different, but related, distribution $Q^*$ from a set of distributions $\mathcal{Q}$ over a (potentially) different observational domain $\mathcal{X}\subseteq\R^e$.

More formally, we assume that we have observed data $\mathbf{X}_n \coloneqq  (X_1,\ldots,X_n) \in\mathcal{X}^n$ consisting of $n$ i.i.d.\ random variables $X_i$ with distribution $Q^* \in \mathcal{Q}$. We assume that $Q^*$ and $P^*$ are related through a map $\tau: \cQ \rightarrow \cP$, called a (distributional) shift, which satisfies $P^*=\tau(Q^*)$. We aim to construct a randomized hypothesis test $\psi_n: \mathcal{X}^{n} \times \R \rightarrow \{0,1\}$ that we 
can
apply to the observed data $\mathbf{X}_{n}$ to test the null hypothesis
\begin{align}\label{eq:H_do}
    \tau(Q^*) \in H_0.
\end{align}
We reject this null hypothesis if $\psi_n = 1$ and do not reject the null if $\psi_n = 0$. 
To allow for random components, 
we let $\psi_n$ take as input a uniformly distributed random variable $U$ (assumed to be independent of the other variables) that generates the randomness of $\psi_n$.
Whenever there is no ambiguity about the randomization, we omit $U$ and write $\psi_n(\mathbf{X}_{n})$; unless stated otherwise, any expectation or probability includes the randomness of $U$.
For $\alpha \in (0, 1)$, we say that $\psi_n$ holds level $\alpha$ at sample size $n$ if it holds that
\begin{align} \label{eq:sc2}
    \sup_{Q \in \tau^{-1}(H_0)} \P_{Q}(\psi_n(\mathbf{X}_{n}, U) = 1) \leq \alpha.
\end{align}
In practice, requiring level at sample size $n$ is often too restrictive. We say that the test has pointwise asymptotic level $\alpha$ if  
\begin{align}\label{eq:pointwise-level}
    \sup_{Q \in \tau^{-1}(H_0)} \limsup_{n \rightarrow \infty} \P_Q(\psi_n(\mathbf{X}_n, U) = 1) \leq \alpha.
\end{align}
We illustrate the setup in \cref{fig:framework}.
\begin{remark}
The map $\tau:\cQ\rightarrow\cP$ above represents a view that starts with the distribution  $Q^*$ of the observed data and considers the distribution $P^*$ of interest as the image under $\tau$. Alternatively, one may also start with a map $\eta: \cP \rightarrow \cQ$ and say that the test holds level $\alpha$ at sample size $n$ if 
\begin{align}\label{eq:uniform-level}
    \sup_{P \in H_0} \P_{\eta(P)}(\psi_n(\mathbf{X}_n, U) = 1) \leq \alpha.
\end{align}
This corresponds to a level guarantee for a test of the hypothesis $\eta^{-1}(Q^*) \cap H_0 \neq \emptyset$. If $\tau$ is invertible, the two views trivially coincide with $\eta \coloneqq \tau^{-1}$, but in general there are subtle differences, see \cref{sec:forward-backward-map} for details.
In this paper, we use the formulation based on $\tau: \mathcal{Q} \rightarrow \mathcal{P}$, that is, Equations~\eqref{eq:H_do},~\eqref{eq:sc2} and~\eqref{eq:pointwise-level}.
\end{remark}

\subsection{Distributional shifts}\label{sec:dist_shift}
We consider two types of maps $\tau: \mathcal{Q} \rightarrow \mathcal{P}$, both of which can be written in product form.
First, assume that there is a subset $A \subseteq \{1, \ldots, d\}$ together with a known map $r: x^A \mapsto r(x^A) \in [0,\infty)$ such that for all $q \in \mathcal{Q}$ the 
target density\footnote{In the remainder of this work, we assume that $\mathcal{X}$ and $\mathcal{Z}$ are both subsets of $\R^d$, that is $e=d$, and that all distributions in $\cP$ and $\cQ$ have densities with respect to the same dominating product measure $\mu$.
We refer to a distribution $Q$ and its density $q$ interchangeably.} $\tau(q)$ satisfies that \begin{align}\label{eq:tauform}
    \tau(q)(x^1, \ldots, x^d) \propto r(x^A) \cdot q(x^1, \ldots, x^d) \qquad \text{for all } (x^1, \ldots, x^d) \in \mathcal{Z}.
\end{align}
Here, we assume that the factor $r$ is known in the sense that it can be evaluated for any given $x^A$ (or at least on all points in the observed sample $\bX^A_n$). 
This type of map naturally arises in many examples, such as in off-policy evaluations with a known training policy or when performing a conditional independence test with known conditional, see \cref{sec:offpolicytesting}. 

Second, assume that there is a subset $A \subseteq \{1, \ldots, d\}$ together with a known map $r_{(\cdot)}: (q,x^A) \mapsto r_q(x^A) \in [0, \infty)$ such that for all $q \in \mathcal{Q}$, the target density $\tau(q)$ satisfies that
\begin{align} \label{eq:tauform-q}
    \tau(q)(x^1, \ldots, x^d) \propto r_q(x^A) \cdot q(x^1, \ldots, x^d) \qquad \text{for all } (x^1, \ldots, x^d) \in \mathcal{Z}.
\end{align}
Here, we assume that the factor $r_{(\cdot)}$ can be evaluated for any given $(q,x^A)$.
This case arises, for example, when the training policy or the conditional is unknown and needs to be estimated from data.
Section~\ref{sec:applications} contains examples for both types of shifts.
If, in any of the above two cases, the set $A$ is not mentioned explicitly, we implicitly assume $A = \{1, \ldots, d\}$.
In many applications $\tau$ represents a local change in the system, so even though $d$ may be large, $|A|$ will be much smaller than $d$. In particular we do not need to know the entire distribution
to evaluate $r(x^A)$. 

\new{In principle, this approach applies to any full-support distribution $Q$, since for a given target distribution $P \in \mathcal{P}$, $P = \tau(Q)$ 
is satisfied as long as we define $r(x) = p(x)/q(x)$, and in the case that we consider a change in a single conditional, this simplifies to $r(x^{\{a_1, a_2\}}) = p(x^{a_1}|x^{a_2})/q(x^{a_1}|x^{a_2})$. In practice, there may be regions of the support where $q(x)$ is much smaller than $p(x)$, in which case the weights will be ill-behaved. We address this issue in \cref{assump:finite-second-moment} and analyze its impact in \cref{thm:finite-level-SIR}. 
For some shifts and hypotheses, direct solutions are available that do not use the importance weights \cref{eq:tauform}: 
For example, when testing a hypothesis about $X^1$ under a mean shift in the marginal distribution of $X^1$, one could directly add the anticipated shift in mean to every observation before testing.
However, in most cases involving shifts in conditional distributions or in variables different from those entering the test, such approaches fail. 
If $\tau$ is misspecified, in the sense that $\tau(Q^*) \neq P^*$, then the guarantees for the methodology below still hold, but for testing the distribution $\tau(Q^*)\in H_0$ instead of $P^*\in H_0$.}

\subsection{Exploiting a test in the target domain}\label{sec:test-target-to-observable-domain}
In this work, we assume that there is a test $\varphi$ for the hypothesis $H_0$ that can be applied to data from the target domain $\cZ$.
Formally, we consider a sequence $\varphi_k: \cZ^k \times \R \rightarrow \R$ of (potentially randomized) hypothesis tests for $H_0$ that can be applied to $k$ observations $\bZ_k$ from the target domain $\cZ$ and a uniformly distributed random variable $V$, generating the randomness of $\varphi_k$. 
For simplicity, we omit $V$ from the notation and write $\varphi_k(\bZ_k)$.
We say that  $\varphi := (\varphi_k)_k$  has pointwise asymptotic level $\alpha$ for $H_0$ in the target domain if
\begin{align}\label{eq:pointwise-level-target-test}
    \sup_{P \in H_0} \limsup_{k\rightarrow\infty} \P_P(\varphi_k(\bZ_k) = 1) \leq \alpha.
\end{align}
To address the problem of testing under  distributional shifts, we propose in \cref{sec:main-results} to resample a data set of size $m $ from
the observed data $\mathbf{X}_n$ (using resampling weights that depend on the shift)
and apply the test $\varphi_m$ to the resampled data.
We show that this yields a randomized test $\psi$ , which inherits the pointwise asymptotic properties of $\varphi_m$ and in particular satisfies the level requirement~\eqref{eq:pointwise-level} if $\varphi$ has pointwise asymptotic level.
This procedure is easy-to-use and can be combined with any testing procedure $\varphi$ from the target domain. 

\subsection{Testing hypotheses in the observed domain} \label{sec:testingobserved}

The framework of testing hypotheses in the target distribution can be helpful even if we are interested in testing a hypothesis about the observed distribution  $Q^*$, that is, testing $Q^* \in H_0^\mathcal{Q}$ for some $H_0^\mathcal{Q} \subseteq \mathcal{Q}$.
If $\tau(H_0^\mathcal{Q}) \subseteq H_0^{\mathcal{P}} :=  H_0 $, any test $\psi_n$ satisfying pointwise asymptotic level~\cref{eq:pointwise-level} for 
$
H_0^{\mathcal{P}} 
\subseteq \mathcal{P}$ can be used as a test for $Q^* \in H_0^\mathcal{Q}$, and will still satisfy asymptotic level, see~
\cref{sec:testing_in_obs_res}.

Such an approach can be particularly interesting when it is more difficult to test $Q^* \in H_0^\mathcal{Q}$ in the observed domain than it is to test $\tau(Q^*) \in H_0^{\mathcal{P}}$ 
in the target domain. 
For example, testing conditional independence in the observed domain can be reduced to (unconditional) independence testing in the target domain.
Here, we may benefit from transferring the test into the target domain if one of the conditionals is known or can be estimated from data. Also testing a Verma equality \citep{Verma1991} in the observed distribution can be turned into an independence test in the target distribution, too; but here, testing directly in the observed domain may not even be possible.
Often there is a computational advantage of our approach: 
In many situations, the resampled data set, where the hypothesis is easier to test, is much smaller than the original data set, see for instance the experiment in \cref{exp:cond-ind-test}.
When the hypothesis of interest is in the observed domain, usually different choices for the target distribution are possible. 
In practice, it is helpful to choose a target distribution that yields well-behaved resampling weights (see~\cref{eq:SIR-weights}), 
which can often be achieved by matching certain marginals, see, e.g., Section~\ref{sec:expverma} \citep[see also][]{robins2000, hernan2006estimating}.

The following \cref{sec:applications} 
discusses the above and other applications of 
testing under distributional shifts in more detail.
Corresponding simulation  experiments are presented in Section~\ref{sec:expe}. \cref{sec:main-results}
provides details of our method and its theoretical guarantees. 

\section{Example applications of testing under distributional shifts} \label{sec:applications}

\subsection{Conditional independence testing} 
\label{sec:conditionaltesting}
Let us first consider 
a random vector $(X, Y, Z)$
with joint probability density function $q^*$ and assume that the conditional $q^*(z|x)$ is known.
We can then apply our framework to test 
$$
H_0^{\mathcal{Q}}=\{Q: X\indep Y\mid Z \text{ and } q(z|x) = q^*(z|x)\} \quad \text{(cond.\ ind.\ in observed domain)}
$$
by reducing the problem to an unconditional independence test.
The key idea is to factor a density $q \in H_0^{\mathcal{Q}}$ as $q(x,y,z) = q(y|x,z)q^*(z|x)q(x)$, 
replace\footnote{If the factorization happens to correspond to the factorization using a causal graph, this is similar to performing an intervention on $Z$, see \cref{sec:scm}. However, the proposed factorization is always valid, so this procedure does not make any assumptions about causal structures.} the conditional $q^*(z|x)$ by, e.g., a standard normal density $\phi(z)$ to obtain the target density $p$, and then test for unconditional independence of $X$ and $Y$. 
When $X$ is a randomized treatment, $Y$ the outcome, and $Z$ is a mediator, this corresponds to testing (non-parametrically) the existence of a direct causal effect \citep[e.g.,][]{pearl2009causality, imbens2015causal, hernan2020causal}.

Formally, we define a corresponding hypothesis in the target domain: 
$$H_0^{\mathcal{P}}:=\{P: X\indep Y \text{ and } p(z|x) = \phi(z)\} \quad \text{(ind.\ in target domain)}$$  
with $\phi$ being the standard normal density.
We can then define a map $\tau$ by
$$\tau(q)(x, y, z) := \frac{\phi(z)}{q^*(z|x)} \cdot q(x,y,z) \qquad \text{for all } (x,y,z) \in \mathcal{Z}.$$
Considering any $q \in H_0^{\mathcal{Q}}$ and writing $p := \tau(q)$, we have
$$
p(x,y,z) = \frac{\phi(z)}{q^*(z|x)}q(y|x,z)q^*(z|x)q(x) = q(y|x,z)q(x)\phi(z).
$$
This shows\footnote{The following statement holds because, clearly, $p(z|x) = \phi(z)$ and if $X \indep Y \,|\,Z$ in $q$, that is, $q(y|x,z) = q(y|z)$ for all $x,y,z$ yielding this expression well-defined, it follows $p(x,y) = p(x)p(y)$.
} that $X \indep Y \,|\,Z$ in $q$ implies $X \indep Y$ in $p$ and therefore
$\tau(H_0^\mathcal{Q}) \subseteq H_0^\mathcal{P}$.
Starting with an independence test $\varphi_m$ for $H_0^{\mathcal{P}}$, we can thus test
$\tau(Q^*) \in H_0^{\mathcal{P}}$, with level guarantee in \eqref{eq:pointwise-level}.
As we have argued in Section~\ref{sec:testingobserved}, this corresponds to testing
$Q^* \in H_0^{\mathcal{Q}}$, and thereby reduces the question of conditional independence to independence.

If, instead of $q^*(z|x)$, we know the reverse conditional
$q^*(x|z)$, we can use 
the same reasoning as above
using the factorization 
$q(x,y,z) = q(z)q^*(x|z)q(y|x,z)$ and a marginal target density $\phi(x)$
to again test $X \indep Y \,|\,Z$. 
When $X$ is a treatment, 
$Y$ the outcome, $Z$ is the full set of covariates, and 
$q^*(x|z)$ represents the 
randomization scheme, this corresponds to testing (non-parametrically) the existence of a total causal effect 
  \citep[e.g.,][]{peters2017elements}
between $X$ and $Y$.

If neither of the conditionals is known, we can 
still fit the test into our framework. 
To do so, define the hypotheses
$H_0^{\mathcal{Q}}:=\{Q: X\indep Y| Z\}$, 
$H_0^{\mathcal{P}}:=\{P: X\indep Y \text{ and } p(z|x) = \phi(z)\}$, 
and the map $\tau$ via 
$\tau(q)(x, y, z) := \frac{\phi(z)}{q(z|x)} \cdot q(x, y, z)$, for all $(x^1, \ldots, x^d) \in \mathcal{Z}$;
cf.~\eqref{eq:tauform-q}. 
Section~\ref{subsec:theoretical-guarantees} shows that one can
estimate the conditional $q(z|x)$ from data and may still maintain the level guarantee of the overall procedure.
There are other, more specialized conditional independence tests but this viewpoint may be  an interesting alternative if we can estimate one of the conditionals well, e.g., because there are many more observations of $(X,Z)$ than there are of $(X,Z,Y)$.

The assumption of knowing one conditional $q(x|z)$ is also exploited by the conditional randomization (CRT) and the conditional permutation test (CPT) by \citet{candes2018panning} and \citet{berret2020the}, respectively. They simulate
(in case of CRT) or permute (in case of CPT)
$X$ while keeping $Z$ and $Y$ fixed and construct $p$-values for the hypothesis of conditional independence.
The approaches are 
similar in that they use the known conditional to create weights. Our method, however, explicitly 
constructs a target distribution 
and, as argued above, cannot only exploit knowledge of $q^*(x|z)$ but also knowledge of $q^*(z|x)$. 

\subsection{Off-policy testing}\label{sec:offpolicytesting}
Consider a contextual bandit setup  \citep[e.g.][]{langford2007epoch,agarwal2014taming}. In each round, an agent observes a context $Z \coloneqq (Z^1, \ldots, Z^d)$ and selects an action $A \in \{a_1, \ldots, a_L\}$, based on a known policy $q^*(a|z)$. 
The agent then receives a reward $R$ depending on the chosen action $A$ and the observed context $Z$.
Suppose 
we have access to a data set $\bX_n$ of $n$ rounds containing observations $X_i := (Z_i, A_i, R_i)$,
$i =1, \ldots, n$. 
We can then 
test 
statements 
about the distribution under another policy  $p^*(a|z)$.
For example, we can test
whether the expected reward is smaller than zero. 
To do so, we define
$$
H_0 \coloneqq
\{P: 
\E_{P}[R] \leq 0 \quad \text{ and } \quad
p(a|z) = p^*(a|z)
\}
$$
and
$\tau(q)(x) \coloneqq r(x)q(x)$ with
the shift factor $r(z,a) \coloneqq p^*(a|z)/q^*(a|z)$. 
Here, the function of interest 
can be written as an expectation of a single observation,
so other, simpler approaches such as IS or 
IPW can be used, too (see~\cref{sec:intro}).

But it is also possible to test more involved hypotheses. 
This includes
testing (conditional) independence under a new policy, for example. 
Suppose that one of the covariates $Z^j$ is used for selecting actions by an observed policy $q^*(a|z)$. This creates a dependence between $Z^j$ and $R$, but it is unclear whether 
this dependence is only
due to the action $A$ being based on $Z^j$, 
 or whether $Z^j$ also depends on $R$ in other ways, for instance in that $Z^j$ has a direct effect on $R$.
 To test the latter statement, we can create a new policy $p^*(a|z)$ that does not use $Z^j$ for selecting actions. Then, we can test whether, under $p^*(a|z)$,  $R$ is independent of $Z^j$, given the other variables that the action is based on.
If not, we know that there must be a dependence between $R$ and $Z^j$ under $q^*(a|z)$ beyond the action $A$ being based on $Z^j$.
This may be relevant for learning sets of features that are invariant across different environments, that is, features $Z^J$ such that $R\mid Z^J$ is stable across environments. A policy that depends on such invariant features is guaranteed to generalize to unseen environments \citep{saengkyongam2021invariant}.
Another, more involved hypothesis for off-policy evaluation compares the reward distributions under two different policies. This can be written as a two-sample test, which we discuss next.

\subsection{Two-sample testing with one transformed sample} \label{sec:twosampletesting}
We can use the framework to perform a two-sample test, after transforming one of the two samples. 
Consider
the observed distribution $q^*$ over $X=(X^1,\ldots,X^d)\in\R^d$
and $K \in \{1,2\}$, 
where the latter indicates which of the two samples a data point belongs to. 
We now keep the first sample as it is and change the second sample, i.e., 
\begin{align*}
q^* \mapsto \tau(q^*) \quad \text{ with } \quad \tau(q^*)({x} | k=1) = q^*({x} | k=1).
\end{align*} 
We can then
test whether, after the transformation, 
the two samples come from the same distribution, i.e., whether
$$
q^*({x} | k=1) = \tau(q^*)({x} | k=2)
$$
for all ${x}$. 
For example, 
let us assume that in the second sample, we know the conditional $q^*(x^2|x^1, k=2)$
and change it to $p^*(x^2|x^1, k=2)$, which we assume to be known, too.
To formally apply 
our framework, we then define
\begin{align*}
H_0 := \{P\,:\, &(X^1, \ldots, X^d)_{|K=1} \, \overset{
{\mathcal{L}}}{=} \, (X^1, \ldots, X^d)_{|K=2}, \\
& \quad 
p(x^1|x^2, k=1) = q^*(x^1|x^2, k=1)
\text{ and } 
p(x^1|x^2, k=2) = p^*(x^1|x^2, k=2)\}
\end{align*}
and the shift 
$
\tau(q)(x^1, \ldots, x^d, k) := 
r(x^1,x^2,k) \cdot q(x^1, \ldots, x^d, k),
$
where
$$
r(x^1,x^2,k) := \left\{\begin{array}{cl}
1 & \text{if } k = 1\\ 
\frac{p^*(x^1|x^2,k=2)}{q^*(x^1|x^2,k=2)} & \text{if } k = 2. 
\end{array}
\right.
$$
Two-sample testing under distributional shifts can be used for off-policy evaluation (the setting is described in the previous section).
We first split the training 
sample 
into two subsamples ($K=1$ and $K=2$) and then test
whether the distribution of the reward is different under the two policies,
$$
H_0 := \{P\,:\, R_{|K=1} \, \overset{{\mathcal{L}}}{=} \, R_{|K=2},\; 
p(a|z, k=1) = q^*(a|z)
\text{ and } 
p(a|z, k=2) = p^*(a|z)\}.
$$
With 
a similar reasoning we can also 
test, non-parametrically, whether the expected reward  under the new policy
$p^*(a|z,k=2)$
is larger than under the current policy
$q^*(a|z,k=2)$.
To do so,
we define
$$
H_0 := \{P\,:\, \E_P [R_{|K=2}] \leq \E_P[R_{|K=1}], \;
p(a|z, k=1) = q^*(a|z)
\text{ and } 
p(a|z, k=2) = p^*(a|z)\},
$$
for example. 
Section~\ref{sec:exp-off-policy}
shows some empirical evaluations of such tests. 

\subsection{Dormant independences}\label{sec:dormant-indep-theory}
Let us consider a random vector $(X^1, \ldots, X^d)$ with a distribution $Q$ that is Markovian with respect to a directed acyclic graph and that has a density w.r.t.\ a product measure. By the global Markov condition  \citep[e.g.][]{Lauritzen1996}, we then have for all disjoint subsets $A, B, C \subset \{1, \ldots, d\}$ that $X^A \indep X^B \,|\,X^C$ if $A$ $d$-separates\footnote{Whether a $d$-separation statement holds is entirely determined from the graph; the precise definition of $d$-separation can be found in \citep[e.g.,][]{Spirtes2000} but is not important here.} $B$ given $C$.
If some of the components of the random vector are unobserved, the Markov assumption still implies conditional independence statements in the observational distribution. In addition, however, it may impose constraints on the observational distribution that are different from conditional independence constraints. 
Figure~\ref{fig:vermaL} shows a famous example, due to \citet{Verma1991}, that
gives rise to the Verma-constraint: If the random vector $(X^1, X^2, X^3, X^4, H)$ has a distribution $Q$ that is Markovian w.r.t.\ the graph $\mathcal{G}$ shown in Figure~\ref{fig:vermaL} (left), there exists a function $f$ such that, for all $x^1, x^3, x^4$,
\begin{equation} \label{eq:verma}
    \int_{-\infty}^{\infty} q(x^2| x^1) q(x^4 | x^1, x^2, x^3)\, d{x^2} = f(x^3,x^4)
\end{equation}
(in particular, $f$ does not depend on $x^1$). This constraint cannot be written as a conditional independence constraint in the observational distribution $Q$. In general, the constraint~\eqref{eq:verma} does not hold if $Q$ is Markovian w.r.t.\ $\mathcal{H}$ (see Figure~\ref{fig:vermaL}, right).
Assume now that the conditional $q(x^3|x^2) = q^*(x^3|x^2)$ is known (e.g., through a randomization experiment).
We can then hope to test for this constraint by considering the null hypothesis
$$
H_0^{\mathcal{Q}}:= \{Q: Q \text{ satisfies } \eqref{eq:verma} \text{ and } q(x^3|x^2) = q^*(x^3|x^2)\}
$$
and hence distinguish between $\mathcal{G}$ and $\mathcal{H}$.
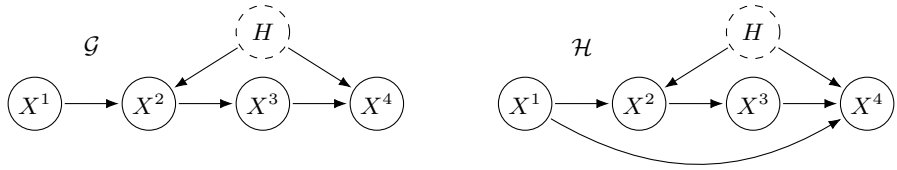
\begin{figure}
\centerline{
\begin{tikzpicture}[xscale=1.5, yscale=1.9, shorten >=1pt, shorten <=1pt]
\small
  \draw (0,3) node(a) [observedsmall] {$X^1$};
  \draw (1,3) node(b) [observedsmall] {$X^2$};
  \draw (2,3) node(c) [observedsmall] {$X^3$};
  \draw (3,3) node(d) [observedsmall] {$X^4$};
  \draw (2,3.5) node(h) [unobservedsmall] {$H$};
  \draw[-Latex, white] (a) to [bend right = 25] (d);
  \draw[-Latex] (a) -- (b);
  \draw[-Latex] (b) -- (c);
  \draw[-Latex] (c) -- (d);
  \draw[-Latex] (h) -- (b);
  \draw[-Latex] (h) -- (d);
  \node[] at (0.5,3.4) {$\mathcal{G}$};
\end{tikzpicture}
\hspace{1cm}
\begin{tikzpicture}[xscale=1.5, yscale=1.9, shorten >=1pt, shorten <=1pt]
\small
  \draw (0,3) node(a) [observedsmall] {$X^1$};
  \draw (1,3) node(b) [observedsmall] {$X^2$};
  \draw (2,3) node(c) [observedsmall] {$X^3$};
  \draw (3,3) node(d) [observedsmall] {$X^4$};
  \draw (2,3.5) node(h) [unobservedsmall] {$H$};
  \draw[-Latex] (a) to [bend right = 25] (d);
  \draw[-Latex] (a) -- (b);
  \draw[-Latex] (b) -- (c);
  \draw[-Latex] (c) -- (d);
  \draw[-Latex] (h) -- (b);
  \draw[-Latex] (h) -- (d);
  \node[] at (0.5,3.4) {$\mathcal{H}$};
\end{tikzpicture}
}
\caption{\label{fig:vermaL}
If $Q$ is Markovian w.r.t.\ graph $\mathcal{G}$ (left), then $Q$ satisfies the Verma constraint~\eqref{eq:verma}. In general, this constraint 
does not hold if $Q$ is Markovian w.r.t.\ $\mathcal{H}$ (right).
Such constraints can be tested for using the framework of statistical testing under distributional shifts, see Section~\ref{sec:dormant-indep-theory}.}
\end{figure}%
Constraints of the form~\eqref{eq:verma} have been studied recently, and in a few special cases, such as binary or Gaussian data, the constraints can be exploited to construct score-based structure learning methodology \citep{shpitser2012parameter,nowzohour2017distributional}. 
\citet{shpitser2008dormant}
show that some of such constraints, 
called dormant independence constraints,
can be written 
as a conditional independence constraint in an interventional distribution \citep[see also][]{ShpitserIntroToNested2014, Richardson2017}, and \citet{shpitser2009testing} propose an algorithm that detects constraints that arise due to dormant independences using oracle knowledge.
The Verma constraint~\eqref{eq:verma}, too, is a dormant independence, that is, we have
\begin{equation}\label{eq:verma2}
X^1 \indep X^4 \qquad \text{ in } Q^{do(X^3 := N)},
\end{equation}
where $N \sim \mathcal{N}(0,1)$, for example.
Here, $Q^{do(X^3 := N)}$,
denotes the distribution in which 
$q^*(x^3|x^2)$ is replaced by 
$\phi(x^3)$
see~\cref{sec:scm} for details.
Using the described framework, we can test \cref{eq:verma2} to distinguish between $\mathcal{G}$ and $\mathcal{H}$.

In practice, we may need to estimate the corresponding conditional, such as $q(x^3|x^2)$ 
in the example above, from data; as before, this still fits into the framework using~\cref{eq:tauform-q},
see~\cref{sec:expverma} for a simulation study.
In special cases, such as binary, applying resampling methodology to this type of problem has been considered before \citep{rohitverma}, but we are not aware of any work proposing a general testing procedure with theoretical guarantees.

\subsection{Uncovering heterogeneity for causal discovery} \label{icpplus}

For a response variable $Y$, consider the problem of finding 
the causal predictors 
$X^{\PA_Y}$, with $\PA_Y\subseteq\{1,\ldots,d\}$, 
among a set of potential predictors 
$X^1, \ldots, X^d$. 
The method of invariant causal prediction (ICP) \citep{Peters2016jrssb, HeinzeDeml2017, Pfister2018jasa}, for example, 
assumes that data are observed in different environments and that the causal mechanism for $Y$, given its causal predictors $\PA_Y$
is 
invariant over the observed environments  \citep[see also][]{Haavelmo1944, Aldrich1989, pearl2009causality}.
This allows for the following procedure: For all subsets $S \subseteq \{1, \ldots, d\}$ one tests whether the conditional $Y|X^S$ is invariant. The hypothesis is true for the set of causal parents, so taking the intersection over all such invariant sets yields, with large probability, a subset of %
$\PA_Y$
  \citep{Peters2016jrssb}.
  Environments can, for example, correspond to different interventions on 
  a node $X^j$.
Using the concept of testing under distributional shifts, we can apply a similar reasoning even if no environments are available and one causal conditional is known instead.

Assume a causal model (e.g., a structural causal model, SCM, see \cref{sec:scm}) over the variables $Y, X^1,\ldots,X^d$ and denote the causal predictors of $X^j$ by $\PA_j$. 
Assume further that there is a $j$ for which the conditional $q^*(x^j|x^{\PA_{j}})$ is known. To infer the causal parents of $Y$, we now construct a new distribution, in which the conditional
$q^*(x^j|x^{\PA_{j}})$
has been changed to another conditional 
$p^*(x^j|x^{\PA_{j}})$
 -- this corresponds to a distribution generated by an intervention on $X^j$. 
We then take the original and 
the resampled data as two `environments' and  apply the ICP methodology by testing whether the conditional $Y\mid X^S$ is invariant w.r.t.\ these two environments. That is, in the absence of `true heterogeneity', we use the known conditional to artificially sample heterogeneity.
Formally, for a candidate set $S \subseteq \{1, \ldots, d\}$
and an indicator variable $K$ indexing the two environments, we define the hypothesis
\begin{align*}
    H_{0,S} := \{P\,:\, Y\mid {X^S}_{|K=1} \, \overset{{\mathcal{L}}}{=} \, Y\mid {X^S}_{|K=2},\quad 
&p(x^j|x^{\PA_{j}}, k=1) = q^*(x^{\PA_{j}}|x^{\PA_{j}})
\text{ and } \\
&p(x^j|x^{\PA_{j}}, k=2) = p^*(x^{\PA_{j}}|x^{\PA_{j}})
\}
\end{align*}
and the shift factor $r(x^j, x^{\PA_j}, k)$ similar to the one in Section~\ref{sec:twosampletesting}. Naturally, the procedure extends to $K > 2$.
The distributional shift corresponds to an intervention on $X^j$ and it follows by modularity\footnote{Formally, given an SCM, the statement follows from the global Markov condition \citep{Lauritzen1996} in the augmented graph, including an intervention node with no parents that points into $X^j$.} that $H_{0, \PA_Y}$ is true.
Therefore, the intersection over 
all sets for which $H_{0,S}$ holds trivially satisfies 
$$
\bigcap_{S: H_{0,S} \text{ holds}} S \subseteq \PA_Y, 
$$
where we define the intersection over an empty index set as the empty set. 
Our framework allows for testing such hypotheses from finitely many data (that were generated only using the conditional $q^*(x^j|x^{\PA_{j}})$) and prove theoretical results that imply level statements for testing $H_{0,S}$. Such guarantees carry over to coverage statements for $\hat {S} := \cap_{S: H_{0,S} \text{ not rej.}} S$, that is, $\hat S \subseteq \PA_Y$ with large probability.

\subsection{Model selection under covariate shift} \label{sec:modselundercovsh}
Consider the problem of comparing models in a supervised learning task when the covariate distribution 
changes compared to the distribution that generated the training data. 
Formally, let us consider an i.i.d.\ sample $D \coloneqq \{(X_i, Y_i)\}_{i=1}^n$ from a distribution $q^*$, where $X_i \in \mathcal{X}$ are covariates with density
$q^*(x)$ and $Y_i \in \mathcal{Y}$ is a label with conditional density $q^*(y | x)$. First, we randomly split the sample into two distinct sets, which we call training set $D_{train}$ and test set $D_{test}$. Let $\hat{f}_1: \mathcal{X} \xrightarrow{} \mathcal{Y}$ and $\hat{f}_2: \mathcal{X} \xrightarrow{} \mathcal{Y}$ be outputs of two supervised learning algorithms trained on $D_{train}$. 
In model selection under covariate shift \citep[e.g.][]{quinonero2009dataset}, we are interested in comparing the performance of the predictors
$\hat{f}_1$ and $\hat{f}_2$ on a distribution $p^*$, where the covariate distribution is changed from $q^*(x)$ to $p^*(x)$, but the conditional $p^*(y | x) = q^*(y | x)$ remains the same. 
If we had an i.i.d.\ data set
$D^{sh}_{test}$
from the shifted distribution 
$p^*$,
we could compare the performances using a scoring function $\mathcal{S}(D_{test}^{sh},\hat{f})$ that for each of the predictors
outputs a real-valued evaluation score. 
However, we only have access to $D_{test}$, which comes from $q^*$. Let us for now assume that the shift from $q^*(x)$ to $p^*(x)$ is known.
Existing methods use IPW to correct for the distributional shift \citep{sugiyama2007covariate}, 
which requires that the scoring function can be expressed in terms of an expectation of a single observation, such as the mean squared error. However, such a decomposition is not immediate for many scoring functions  as for example the area under the curve (AUC).
The framework of testing under distributional shifts allows for an arbitrary 
scoring function (as long as a corresponding test exists) while maintaining statistical guarantees. To this end, we define the hypothesis 
\begin{align*}
    H_{0,\hat{f}_1,\hat{f}_2} := \{P\,:\, \E_{D_{test}^{sh} \sim P}\big[\mathcal{S}(D_{test}^{sh},\hat{f}_1) - \mathcal{S}(D_{test}^{sh},\hat{f}_2)\big] \leq 0,\quad 
&p(x) = p^*(x), p(y|x) = q^*(y|x)\},
\end{align*}
with the shift factor $r(x) \coloneqq p^*(x)/q^*(x)$. 
Using data $D_{test}$ from $q^*$, the methodology developed below allows us to test this hypothesis
$H_{0,\hat{f}_1,\hat{f}_2}$, that is, whether, in expectation, 
$\hat{f}_1$ outperforms $\hat{f}_2$ in the target distribution $p^*$, which includes the shifted covariate distribution.
In practice, the densities $p^*(x)$ or
$q^*(x)$ may not be given but one can still estimate these densities from data and apply our framework using \eqref{eq:tauform-q}.

\section{Testing by Resampling}\label{sec:main-results}
\new{
In \cref{sec:applications}, we listed various problems that can be solved by testing a hypothesis about a shifted distribution. 
In this section, 
we 
outline several approaches to test a target hypothesis $\tau(Q^*) \in H_0$, see \cref{eq:H_do}, using a sample $\bX_n$ from the observed distribution $Q^*$.
We initially consider the shift $\tau$ known, and later show that asymptotic level guarantees also apply if $\tau$ can be estimated sufficiently well from data.

Our approach relies on the existence of a hypothesis test $\varphi_m$ for the hypothesis $H_0$ in the target domain and applies this test to a resampled version of the observed data, which mimics a sample in the target domain.
We show that -- under suitable assumptions -- properties of the original test $\varphi_m$ carry over to the overall testing procedure $\psi^r_n$ 
(of combined resampling and testing, as defined in \cref{eq:resampled_test}).

This section is organised as follows. 
First, in \cref{subsec:method}, we propose a resampling scheme, which we show in \cref{subsec:theoretical-guarantees} has asymptotic guarantees. In \cref{subsec:resampling}, we discuss how to sample from the scheme in practice and we describe a number of extensions in \cref{subsec:extensions}.
In \cref{subsec:rejection-sampler} we show that a simpler rejection sampling scheme can be used if stricter assumptions are satisfied. 
}

\subsection{\texorpdfstring{\new{Distinct Replacement (DRPL) Sampling}\old{ \\ A consistent statistical test based on resampling}}{}}
\label{subsec:method}
\old{We propose a method for testing the target hypothesis $\tau(Q^*) \in H_0$, see \cref{eq:H_do}, using a sample $\bX_n$ from the 
observed
distribution $Q^*$.
The idea is to take an existing hypothesis test $\varphi_m$ for the hypothesis $H_0$ in the target domain and apply it to a resampled version of the observed data, which mimics a sample in the target domain.
We show that -- under suitable assumptions -- the pointwise asymptotic properties of the  original test $\varphi_m$ carry over to the overall testing procedure $\psi^r_n$ 
(defined in \cref{eq:resampled_test})
in the target domain.}

We consider the setting, where $\tau(q)(x) \propto r(x)q(x)$ for a known shift factor $r$; see \cref{eq:tauform}.
First, we draw a weighted resample of size $m$ from $\bX_n$ similar to the sampling importance resampling (SIR) scheme proposed by \citet{rubin1987calculation} but using a sampling scheme DRPL (`distinct replacment') that is different from sampling with or without replacement.
More precisely, we draw a resample $(X_{i_1}, \ldots, X_{i_m})$ from $\bX_n$, where $(i_1, \ldots, i_m) \in \{1, \ldots, n\}^m$ is a sequence of distinct\footnote{We use `distinct' and `non-distinct' only to refer to the potential repetitions that occur due to the resampling $(i_1, \ldots, i_m)$ and not due to potential repetitions in the values of the original sample $\bX_n$.} values;
the probability of drawing the sequence $(i_1,\ldots,i_m)$ is
\begin{align}\label{eq:SIR-weights}
    w_{(i_1, \ldots, i_m)} \propto 
    \left\{
    \begin{array}{cl}
    \prod_{\ell = 1}^m r(X_{i_\ell}) \propto \prod_{\ell=1}^m \frac{\tau(q)(X_{i_\ell})}{q(X_{i_\ell})} & \text{if } (i_1, \ldots, i_m) \text{ is distinct and}\\ 
    0 & \text{otherwise.}
    \end{array}\right.
\end{align}
We provide an efficient sampling algorithm and discuss different sampling schemes
in \cref{subsec:resampling}.
We refer to $(X_{i_1}, \ldots, X_{i_m})$ as the target sample and denote it by $\Psi_{\texttt{DRPL}}^{r, m}(\bX_n, U)$, where $U$ is a random variable representing the randomness of the resample. If the randomness is clear from context, we omit $U$ and write $\Psi_{\texttt{DRPL}}^{r, m}(\bX_n)$.
When $m$ is fixed and $n$ approaches infinity, the target sample $\Psi_{\texttt{DRPL}}^{r,m}(\bX_n)$ converges in distribution to $m$ i.i.d.\ draws from the target distribution $\tau(Q^*)$; see \citet{skare2003improved} for a proof for a slightly different sampling scheme.
Based on our proposed resampling scheme we construct a test $\psi^r_n$ for the target hypothesis \cref{eq:H_do} using only the observed data $\bX_n$ by defining
\begin{equation}
\label{eq:resampled_test}
    \psi^r_n(\bX_n) \coloneqq \varphi_m(\Psi^{r,m}_{\texttt{DRPL}}(\bX_n)),
\end{equation}
see also
\cref{alg:resampling-and-testing}.

\begin{algorithm}[t]
\caption{Testing a target hypothesis with known distributional shift and resampling}
\begin{algorithmic}[1]
\Statex \textbf{Input:} Data $\bX_{n}$, target sample size $m$, hypothesis test $\varphi_m$, shift factor $r(x^A)$.
\State $(i_1, \ldots, i_m) \gets$ sample from $\{1, \ldots, n\}^m$ with weights~\eqref{eq:SIR-weights} (see \cref{sec:sampling-DRPL})
\State $\Psi_{\texttt{DRPL}}^{r,m}(\bX_n) \gets (X_{i_1}, \ldots, X_{i_m})$
\Statex \Return $\psi_n^r(\bX_{n}) \coloneqq \varphi_m(\Psi_{\texttt{DRPL}}^{r,m}(\bX_{n}))$
\end{algorithmic}
\label{alg:resampling-and-testing}
\end{algorithm}

\new{We show in 
\cref{subsec:theoretical-guarantees} that
distinct sampling, $\Psi_{\texttt{DRPL}}$, allows us to show guarantees without introducing regularity assumptions on the test $\varphi_m$. The motivation for resampling without replacement comes from the fact that tests, as
opposed to estimation of means, may be sensitive to duplicates; an extreme but instructive
example is a test of the null hypothesis that no point mass is present in a distribution. A
resampling test with large replacement size and possible duplicates would not be able to obtain
level in such a hypothesis. Although we show in \cref{subsec:testing-with-REPL} that under stricter assumptions, sampling with replacement, that is using $\Psi_{\texttt{REPL}}$, becomes asymptotically equivalent to using $\Psi_{\texttt{DRPL}}$, this example highlights, that in the non-asymptotic regime, sampling duplicates may be harmful.}

\new{
The resampling scheme $\Psi_{\texttt{DRPL}}$ in \cref{eq:SIR-weights} is similar, but not identical to what would commonly be called `resampling without replacement' ($\Psi_{\texttt{NO-REPL}}$), where one draws a single observation $X_{i_1}$, removes $X_{i_1}$ from the list of candidates for further draws and normalizes the remaining weights to reflect the absence of $X_{i_1}$ (see also \cref{subsec:resampling}). 
$\Psi_{\texttt{DRPL}}$ and $\Psi_{\texttt{NO-REPL}}$ differ in the normalization constants, and the normalization constant in $\Psi_{\texttt{DRPL}}$ is easier to analyze theoretically. This enables \cref{lemma:sum-distinct-weights} in \cref{sec:proofs}, which describes the asymptotic behaviour of the mean and variance of \cref{eq:SIR-weights} as well as of the normalization constant of \cref{eq:SIR-weights}. 
We consider $\Psi_{\texttt{DRPL}}$ a tool that enables simpler theoretical analysis of SIR methods; in practice it is plausible that using $\Psi_{\texttt{NO-REPL}}$ instead of $\Psi_{\texttt{DRPL}}$ will yield similar results, though we are not aware of any theory justifying this.
}
\subsection{\texorpdfstring{\new{Pointwise Asymptotic Level and Power}\old{ \\ Theoretical guarantees}}{}}\label{subsec:theoretical-guarantees}

We now prove that the hypothesis test $\psi^r_n$ inherits the pointwise asymptotic properties of the test $\varphi$ in the target domain.
To do so, we require two assumptions: $m$ and $n$ have to approach 
infinity 
at a suitable rate, and we require the weights to be well-behaved. 
More precisely, we will make the following assumptions.
\begin{assumpenum}
    \item \label{assump:m-rate-n} 
    $m = m(n)$ satisfies $1 \leq m \leq n$, $m \rightarrow \infty$ and $m=o(\sqrt{n})$ for $n \rightarrow \infty$.
    \item \label{assump:finite-second-moment} 
    $\E_{Q}[r(X_i)^2] < \infty$.
\end{assumpenum}
\Cref{assump:m-rate-n} states that $m$ must approach infinity at a slower rate than $\sqrt{n}$.
\Cref{assump:finite-second-moment} is a condition to ensure the weights are sufficiently well-behaved, and is similar to conditions required for methods based on IPW, for example \citep[][]{robins2000}. 
If $r(x^A)$ only depends on a subset $A$ of variables, and $x^A$ takes finitely many values, \cref{assump:finite-second-moment} is trivially satisfied for all $Q$. In the case of an off-policy hypothesis test, such as the one described in \cref{sec:offpolicytesting}, a sufficient but not necessary condition for \Cref{assump:finite-second-moment} to hold for $Q^*$ is that the policy $q^*(a|z)$ is randomized, such that there is a lower bound on the probability of each action. 
For a Gaussian setting, where $r$ represents a change of a conditional $q(x^j|x^{j'})$ to a Gaussian marginal $p(x^j)$, we provide in \cref{sec:assumption-a3-gaussian} 
sufficient and necessary conditions under which  \cref{assump:finite-second-moment} is satisfied.
If the hypothesis of interest is in the observed domain (see \cref{sec:testingobserved}), we are usually free to choose any target density, so we can ensure that the tails decay sufficiently fast to satisfy \cref{assump:finite-second-moment}.
In \cref{sec:exp-assumptions-a2-a3} below, we 
analyze the influence of \cref{assump:m-rate-n,assump:finite-second-moment} on our test holding level in the context of synthetic data.
We now \old{state}\new{present} the first main result \new{which states that if $\alpha_\varphi \coloneqq \limsup_{k\to\infty} \P_{P^*}(\varphi_k(\bZ_k) = 1)$ is the asymptotic level of the test $\varphi$ when applied to a sample $\bZ_k$ from $P^*$, then this is also the asymptotic level of the resampling test in \cref{alg:resampling-and-testing} when applied to a sample $\bX_n$ from $Q^*$.} 
All proofs can be found in \cref{sec:proofs}.
\begin{theorem}[Pointwise asymptotics -- known weights]\label{thm:asymptotic-level-SIR}
    Consider a null hypothesis $H_0 \subseteq \mathcal{P}$ in the target domain. Let $\tau: \cQ \rightarrow \cP$ be a distributional shift for which a known map $r:\mathcal{X}\rightarrow[0,\infty)$ exists, satisfying $\tau(q)(x)=r(x)q(x)$, see \cref{eq:tauform}.
    Consider an arbitrary $Q \in \cQ$ and $P = \tau(Q)$. Let $\varphi_k$ be a sequence of tests for $H_0$ 
    and define     
    $\alpha_\varphi \coloneqq \limsup_{k\to\infty} \P_{P}(\varphi_k(\bZ_k) = 1)$.
    Let $m = m(n)$ be a resampling size and let $\psi^r_n$ be the DRPL-based resampling test defined by $\psi^r_n(\bX_n) \coloneqq \varphi_m(\Psi_{\texttt{DRPL}}^{r,m}(\bX_n))$, see \cref{alg:resampling-and-testing}.
    Then, if $m$ and $Q$  satisfy \cref{assump:m-rate-n,assump:finite-second-moment}, respectively,
    it holds that
    \begin{equation*}
        \limsup_{n\rightarrow\infty}\P_{Q} (\psi^r_{n}(\mathbf{X}_n)=1)=\alpha_\varphi.
    \end{equation*}
    The same statement holds when replacing both $\limsup$'s with $\liminf$'s.
\end{theorem}
\Cref{thm:asymptotic-level-SIR} shows that the rejection probabilities of $\psi^r$ and $\varphi$ converge towards the same limit.
In particular, the theorem states that if $\varphi$ satisfies pointwise asymptotic level in the sense of \cref{eq:pointwise-level-target-test}, and \cref{assump:finite-second-moment} holds for all $Q \in \tau^{-1}(H_0)$, then also $\psi^r$ satisfies pointwise asymptotic level~\cref{eq:pointwise-level}.
Similarly, because the statement holds for $P \notin H_0$, too, $\psi^r$ has the same asymptotic power properties as $\varphi$. 

\new{
We show in \cref{thm:asymptotic-level-SIR}, that \cref{assump:m-rate-n} is sufficient to obtain asymptotic level of the rejection procedure. In fact, as we show in the following theorem, it is also necessary: If $m,n$ approach infinity with $m \geq n^{q}$ for $q>\tfrac{1}{2}$, there exists a distribution $Q$, a shift $r$ and a sequence of tests $\varphi_k$ such that \cref{assump:finite-second-moment} is satisfied and $\alpha_\varphi \coloneqq \limsup_{k\to\infty} \P_{P}(\varphi_k(\bZ_k) = 1) < 1$ but the probability of rejecting the hypothesis on any resample of size $m$ converges to $1$. 
\begin{theorem}\label{thm:necessity-sqrt-n}
    Fix $\ell\in\{2,3,\ldots\}$ and let $\Psi^{m}$ be any resampling scheme that outputs a (not necessarily distinct) sample of size $m=n^q$ with $q > (\ell-1)/\ell$, and let $\alpha_\varphi \in (0,1)$. 
    Then there exist a distribution $Q\in\mathcal{Q}$, a distribution shift $\tau: \mathcal{Q}\rightarrow \mathcal{P}$ with a known map $r:\mathcal{X}\to [0,\infty)$, a null hypothesis $H_0 \subseteq \mathcal{P}$ and a sequence of hypothesis tests $\varphi_{k}$ such that $\tau(Q) \in H_0$, $\limsup_{k\to\infty} \mathbb{P}_{\tau(Q)}(\varphi_{k}(\mathbf{Z}_{k}) = 1) \leq \alpha_\varphi$, $\mathbb{E}_Q[r(X)^\ell] < \infty$ and
    \begin{equation*}
        \lim_{n\rightarrow\infty}\mathbb{P}_{Q}(\varphi_{m}(\Psi^m(\bX_n))=1)= 1.
    \end{equation*}
\end{theorem}
In particular, letting $\ell = 2$ in \cref{thm:necessity-sqrt-n} shows that \cref{thm:asymptotic-level-SIR} does not hold without \cref{assump:m-rate-n}. 
}

\old{\Cref{,thm:asymptotic-level-SIR} consider}So far, we have considered the case in which the known shift factor $r$ does not depend on $q$. 
Next, we consider the setting in which the shift factor is allowed to explicitly depend on $q$.
This is relevant, for example, if the shift $\tau$ represents a change of the conditional of a variable $X^j$ from $q^*(x^j|x^B)$ to $p^*(x^j|x^C)$, say, but the observational conditional $q^*(x|x^B)$ is unknown, corresponding to the setting in \cref{eq:tauform-q}.
If $q^*(x^j|x^B)$ is unknown, we are not able to compute the weights 
$p^*(X_i^j|X_i^C)/q^*(X_i^j|X_i^B)$.
However, we can still try to estimate  $q^*(x^j|x^B)$ (or even $r$) and obtain approximate weights $\hat{r} \propto p^*/\hat{q}^*$. \new{Assume we have two data sets $\bX_{n_1}$ and $\bX_{n_2}$ both containing samples from $Q^*$, with $n_1$ and $n_2$ observations respectively and the first one is used to estimate $r$ and the second one to perform the test, see \cref{alg:resampling-and-testing-unknown-shift}.}
Then, if we make the following modifications to
 \cref{assump:finite-second-moment,assump:m-rate-n},
\begin{alt-assumpenum}[start=1]
    \item
    \label{assump:m-rate-n-prime}
    \new{$m = m(n_2)$ satisfies $1\leq m \leq n_2, m \to \infty$ and
    $m = o(\min(n_1^a,n_2^{1/2}))$ for $n_1,n_2\rightarrow \infty$, }
    \item $\E_Q[r_q(X_i)^2] < \infty$, \label{assump:finite-second-moment-prime}
\end{alt-assumpenum}
the following theorem states that even when estimating the weights, it is possible to obtain pointwise asymptotic level for the target hypothesis \cref{eq:H_do} -- if the weight estimation works sufficiently well.
\begin{theorem}[Pointwise asymptotics -- estimated weights]\label{thm:asymptotic-level-unknown-weights}
    Consider a null hypothesis $H_0 \subseteq \cP$ in the target domain. Let $\tau: \cQ \rightarrow \cP$ be a distributional shift,
    satisfying $\tau(q)(x) \propto r_q(x)q(x)$, see \cref{eq:tauform-q}. Consider an arbitrary $Q \in \cQ$ and $P = \tau(Q)$.
    Let $\varphi_k$ be a sequence of tests for $H_0$ and let $\alpha_\varphi \coloneqq \limsup_{k\to\infty} \P_{P}(\varphi_k(\bZ_k) = 1)$. 
    Let $\hat{r}_n$ be an estimator for $r_q$ such that there exists $a\in (0,1)$
    satisfying 
    \begin{align*}
        \lim_{n\to\infty}\sup_{x \in \cX} \E_Q \left|\left(\frac{\hat{r}_n(x)}{r_q(x)}\right)^{n^a} - 1 \right| = 0,
    \end{align*}
    where the expectation is taken over the randomness of $\hat{r}_n$. \old{Partition $\bX_n$ into two data sets $\bX_{n_1}$ and $\bX_{n_2}$ such that $n_1^a = \sqrt{n_2}$.
    Let $m = m(n)$ be a resampling size and let $\psi_n^{\hat{r}}$ be the DRPL-based resampling test defined by $\psi_n^{\hat{r}} \coloneqq \varphi_m(\Psi^{\hat{r}_{n_1}, m}_{\texttt{DRPL}}(\bX_{n_2}))$, where $\hat{r}_{n_1}$ is estimated using $\bX_{n_1}$, see \cref{alg:resampling-and-testing-unknown-shift}.}\new{Let $m = m(n_2)$ be a resampling size and let $\psi_{n_1,n_2}^{\hat{r}}$ be the DRPL-based resampling test defined by $\psi_{n_1,n_2}^{\hat{r}} \coloneqq \varphi_m(\Psi^{\hat{r}_{n_1}, m}_{\texttt{DRPL}}(\bX_{n_2}))$ from \cref{alg:resampling-and-testing-unknown-shift} in \cref{sec:sir-details}.}
    Then, if $m$ and $Q$ satisfy \cref{assump:m-rate-n-prime,assump:finite-second-moment-prime}, respectively, it holds that
    \begin{align*}
        \limsup_{n\rightarrow\infty} \P_Q(\psi_n^{\hat{r}}(\bX_n) = 1) = \alpha_\varphi.
    \end{align*}
    The same statement holds when replacing both $\limsup$'s with $\liminf$'s.
\end{theorem}
\cref{thm:asymptotic-level-unknown-weights} shows that $\psi_n^{\hat{r}}$ converges to the same limit as $\varphi$. 
In particular, as for the case of known weights, if $\varphi$ satisfies pointwise asymptotic level and \cref{assump:finite-second-moment-prime} holds for all $Q \in \tau^{-1}(H_0)$, then also $\psi^{\hat{r}}$ satisfies pointwise asymptotic level for the hypothesis $\tau(Q^*) \in H_0$ and $\psi^{\hat{r}}$ inherits asymptotic power properties, too.

\subsection{\texorpdfstring{\new{Computationally Efficient Resampling with $\Psi_{\texttt{DRPL}}$}\old{ \\ Resampling Schemes}}{}} \label{subsec:resampling}
In \cref{subsec:method} we propose a sampling scheme $\Psi_{\texttt{DRPL}}$, defined by \cref{eq:SIR-weights}, and in \cref{subsec:theoretical-guarantees} we prove theoretical level guarantees when we resample the observed data using $\Psi_{\texttt{DRPL}}$. In this section, we display a number of ways to sample from $\Psi_{\texttt{DRPL}}$ in practice.

To do so, let $\Psi_{\texttt{REPL}}$ and $\Psi_{\texttt{NO-REPL}}$ 
denote weighted sampling with and without replacement, respectively, both of which are implemented in most standard statistical software packages. 
Though $\Psi_{\texttt{DRPL}}$ and $\Psi_{\texttt{NO-REPL}}$ both sample distinct sequences $(i_1, \ldots, i_m)$, they are not equal, i.e., they distribute the weights differently between the sequences (see \cref{sec:sampling-DRPL}).
We can sample from $\Psi_{\texttt{DRPL}}$ by sampling from $\Psi_{\texttt{REPL}}$ and rejecting the sample until the indices $(i_1, \ldots, i_m)$ are distinct, see \cref{subsec:acceptance-rejection-repl}. In \cref{prop:repl-becomes-dist} 
we prove that under suitable assumptions, such as $m=o(\sqrt{n})$, the probability of drawing a distinct sequence already in a single draw approaches $1$, when $n\to\infty$. 

In some cases (though these typically only occur when \cref{assump:m-rate-n} or \cref{assump:finite-second-moment} are violated, and our asymptotic guarantees do not apply), 
the above rejection sampling from $\Psi_{\texttt{REPL}}$ may take a long time to accept a sample. For these cases, we propose to use an (exact) rejection sampler based on $\Psi_{\texttt{NO-REPL}}$, which will typically be faster (since it has the same support as $\Psi_{\texttt{DRPL}}$).
We provide all details in \cref{sec:no-repl-sampling}.

If neither of the two exact sampling schemes for $\Psi_{\texttt{DRPL}}$ is
computationally feasible, we provide an approximate sampling method  that applies a Gibbs sampler to a sample from $\Psi_{\texttt{NO-REPL}}$; we refer to this scheme as $\Psi_{\texttt{DRPL-GIBBS}}$. Finally, one can simply approximate $\Psi_{\texttt{DRPL}}$ by a sample from $\Psi_{\texttt{NO-REPL}}$ -- this is computationally faster, and 
leads to similar results
in many cases. 
The details are provided in \cref{sec:no-repl-gibbs}. In practice, our implementation first attempts to sample from $\Psi_{\texttt{NO-REPL}}$ by (exact) rejection sampling, and if the number of rejections exceed some threshold, sampling without replacement is used instead.

\Cref{prop:repl-becomes-dist} (mentioned above)
has another implication. 
We prove that we can obtain the same level guarantee, when using $\Psi_{\texttt{REPL}}$ instead of $\Psi_{\texttt{DRPL}}$ (see \cref{cor:asymptotic-level-SIR-REPL} in \cref{sec:sampling-DRPL}). 
This result, however, requires an assumption that is stronger than \cref{assump:finite-second-moment}. Intuitively, stronger assumptions are required for $\Psi_{\texttt{REPL}}$ because sampling with replacement is much more prone to experience large variance due to 
observations with huge weights.

\subsection{Extensions}\label{subsec:extensions}
In this section, we discuss a number of extensions of the methodology and theory presented in the preceding sections. 

\subsubsection{\texorpdfstring{\new{Heuristic}}{} data driven choice of \texorpdfstring{$m$}{}}\label{subsec:test-resampling-worked}

\new{
Resampling distinct sequences
requires that we choose a resampling size $m$ that is
smaller than the original sample size $n$.}
If $m$ is too large when sampling distinct sequences, it can happen that 
eventually there are no more points left that are likely under the target distribution. 
Consequently, the resampling procedure disproportionally often has to sample points that are very unlikely in the target distribution. This leads to the target sample being a poor approximation of the target distribution. 
Our theoretical results show that choosing a resampling size of order $o(\sqrt{n})$ 
avoids this problem, see \cref{thm:asymptotic-level-SIR}.

However, this result is asymptotic and does not immediately translate into finite sample statements. Furthermore, in many cases also the requirement $m = o(\sqrt{n})$ is too strict, and asymptotic level can also be obtained by setting $m = o(n^a)$ for some $a \in (1/2, 1]$ (with the most extreme case being $P^*=Q^*$, where $a = 1$ can be applied). 
Since a larger $m$ typically results in increased power of the hypothesis test, we want to choose $m$ as large as possible while maintaining that the target sample still approximates the target distribution. 

Consider the case where $\tau$ corresponds to changing $q^*(x^j|x^C)$ to $p^*(x^j|x^B)$. We can then test the validity of the resampling by testing whether the target sample matches the theoretical conditional. 
Specifically, for a fixed $m$, we can verify whether the conditional $X^j\mid X^B$ in the resampled data 
$\Psi_{\texttt{DRPL}}^{r,m}(\bX_n)$
is close to the target conditional $p^*(x^j|x^B)$ by a goodness-of-fit test $\kappa(\Psi_{\texttt{DRPL}}^{r,m}(\bX_n)) \in \{0, 1\}$.
If $m$ is chosen too large, the resampling is likely to include many points with small weights, corresponding to small likelihoods 
$p^*(x^j|X^B)$, 
which will cause the goodness-of-fit test to reject the hypothesis that the target sample has the conditional 
$p^*(x^j|X^B)$.

We can use this to construct a data-driven approach to selecting $m$: For an increasing sequence of $m$'s, perform the goodness-of-fit test for several resamples $\kappa(\Psi_{\texttt{DRPL}}^{r,m}(\bX_n)_1)$, $\ldots$, $\kappa(\Psi_{\texttt{DRPL}}^{r,m}(\bX_n)_K)$. If $\tfrac{1}{K}\sum_{k=1}^K \kappa(\Psi_{\texttt{DRPL}}^{r,m}(\bX_n)_k)$ is smaller than some predefined cutoff \new{$\texttt{qt}$},
we accept $m$ as a valid target sample size\footnote{\new{Concretely, 
since under the null hypothesis, $\kappa(\Psi_{\texttt{DRPL}}^{r,m}(\bX_n)_k)$ is uniform,
we chose $\texttt{qt}$ to be the $\alpha_c$-quantile of the mean of $K$ uniform distributions for some $\alpha_c\in(0,1)$, see \cref{sec:targetheuris}. 
Doing so, ensures that for a fixed $m$, under the null hypothesis of the resample having the intended conditional, the test has level $\alpha_c$.}}. We then use the largest accepted $m$ as the resampling size in the actual test for the hypothesis of interest. We summarize the procedure for finding $m$ in \cref{alg:tuning-m} in~\cref{sec:targetheuris}.

To avoid potential dependencies between the tuning of $m$ and the hypothesis test, we could use sample splitting. In practice, however, we use the entire sample, since the dependence between the tuning step and the final test in our empirical analysis appears to be sufficiently low such that the level properties of the final tests were preserved, see e.g., the experiment in \cref{sec:exp-assumptions-a2-a3}.

If the target conditional $p^*(x^j|X^B)$ is a linear Gaussian conditional density (i.e., $X^j\mid X^B \sim \mathcal{N}(\beta^\top X^B, \sigma^2)$ for some parameters $\beta, \sigma$ under $P^*$) the goodness-of-fit test can be performed by using a linear regression and testing the hypothesis that the regression slope in the resample is $\beta$. 
For more complex conditional densities, \new{one should prefer a test that has (pointwise asymptotic) power against a wide range of alternatives. Here,}
we propose to use the kernel conditional-goodness-of-fit test by \citet{jitkrittum2020testing} to test that the resampled data $\Psi_{\texttt{DRPL}}^{r,m}(\bX_n)$ has the desired conditional. 

\new{
\subsubsection{Finite-sample level guarantees}
\label{subseq:finite-sample}
In addition to the asymptotic results presented in \cref{subsec:theoretical-guarantees}, 
we now prove that the hypothesis test $\psi_n^r$ inherits finite-sample level if the test $\varphi$ in the target domain satisfies finite-sample guarantees.

\begin{theorem}[Finite level -- known weights]\label{thm:finite-level-SIR}
    Consider a null hypothesis $H_0 \subseteq \mathcal{P}$ in the target domain. Let $\tau: \cQ \rightarrow \cP$ be a distributional shift for which a known map $r:\mathcal{X}\rightarrow[0,\infty)$ exists, satisfying $\tau(q)(x) = r(x)q(x)$, see \cref{eq:tauform}.
    Consider an arbitrary $Q \in \cQ$ and $P = \tau(Q)$. 
    Let $m$ be a resampling size and let $\varphi_m$ be a test for $H_0$ and define $\alpha_\varphi \coloneqq \P_{P}(\varphi_m(\bZ_m) = 1)$.
    Also let $\psi^r_n$ be the DRPL-based resampling test defined by $\psi^r_n(\bX_n) \coloneqq \varphi_m(\Psi_{\texttt{DRPL}}^{r,m}(\bX_n))$, see \cref{alg:resampling-and-testing}.
    Then, if $Q$ satisfies \cref{assump:finite-second-moment}, it holds that
    \begin{equation}\label{eq:finite-sample-level-bound}
        \P_{Q} (\psi^r_{n}(\mathbf{X}_n)=1) \leq \min_{\delta\in (0,1)} \left(\frac{\alpha_\varphi}{1-\delta} + \frac{V(n,m)}{V(n,m) + \delta^2}\right),
    \end{equation}
    where $V(n,m) = \binom{n}{m}^{-1} \sum_{\ell=1}^m \binom{m}{l}\binom{n-m}{m-\ell}(\E_Q[r(X_1)^2]^\ell - 1)$.
\end{theorem}
Thus if $\E_Q[r(X_1)^2]$ is known, one can evaluate the finite-sample level of the DRPL-based resampling test for any choice of $m$.
We show in \cref{app:evaluation-of-variance} that the term $V(n,m)$ can be computed efficiently and such that numerical under- or overflows is avoided, even if $n$ and $m$ are so large that evaluating the individual terms $\binom{n}{m}$ and $\binom{m}{l}\binom{n-m}{m-\ell}(\E_Q[r(X_1)^2]^\ell - 1)$ may cause under- or overflows.
Given $V(n,m)$, the minimization problem on the right hand side can easily be implemented in numerical optimizers or solved explicitly for the minimal $\delta$.\footnote{Taking the derivative with respect to $\delta$ and equating it to $0$, the resulting equation can be rewritten to a polynomial equation of degree $4$. One can then evaluate the right hand side of \cref{eq:finite-sample-level-bound} at the (at most $4$) roots and additionally the boundary point $\delta=0$,
and use the one that yields the smallest bound.}

If $\tau$ is the identity mapping, i.e. $Q^* = P^*$, then $\E_Q[r(X_1)^2]=1$, and for all $m$, $V(n,m)=0$.
As one would expect, in that case \cref{thm:finite-level-SIR} states that for any $m$, $\P_{Q^*}(\psi^r_{n}(\mathbf{X}_n)=1) \leq \alpha_\varphi$, that is, the probability
of rejecting 
when applying $\varphi$ to the resampled data is upper bounded by 
the probability of rejecting when applying $\varphi$ directly to target data. 

If $P^* \neq Q^*$ then $V(n,m)>0$, and for any $m$ the right hand side of \cref{eq:finite-sample-level-bound}
exceeds $\alpha_\varphi$.
To control the level of the resampling test, say at a rate $\alpha_\psi$, one can set the resample size $m$ small enough
such that the right hand side of \cref{eq:finite-sample-level-bound} is smaller than $\alpha_\psi$. 
We propose to use the largest $m$ such that the right hand side of \cref{eq:finite-sample-level-bound} is bounded by $\alpha_\psi$. 

In practice, we find that in many settings, the inequality \cref{eq:finite-sample-level-bound} is not strict:  the largest $m$ such that the right hand side \cref{eq:finite-sample-level-bound} is bounded by $\alpha_\psi$ is not close to being the largest $m$ such that the left hand side is bounded by $\alpha_\psi$. Hence, for practical purposes, the scheme for choosing $m$ proposed in \cref{subsec:test-resampling-worked} often returns larger values $m$ while retaining level $\alpha_\psi$ under the null hypothesis; we explore this further in \cref{sec:expe}.
}

\subsubsection{Uniform level}\label{subsec:uniform-level}
The asymptotic level guarantees implied by \cref{thm:asymptotic-level-SIR} are pointwise, meaning we are not guaranteed the same convergence rate for all distributions $Q \in \tau^{-1}(H_0)$. However, as the following theorem shows, if a uniform bound on the weights exists, i.e., $\sup_{Q\in\tau^{-1}(H_0)} \E_Q[r(X_i)^2] < \infty$, and the test $\varphi$ has uniform asymptotic level, the overall procedure can be 
shown to hold uniform asymptotic level.
\begin{theorem}[Uniform asymptotic level]\label{thm:uniform-level}
    Assume the same setup and assumptions as in \cref{thm:asymptotic-level-SIR}. If additionally $\sup_{Q\in\tau^{-1}(H_0)} \E_Q[r(X_i)^2] < \infty$ and $\limsup_{k\rightarrow\infty}\sup_{P\in H_0}\P_P(\varphi_k(\bZ_k) = 1) \leq \alpha_\varphi$, then 
    \begin{align*}
        \limsup_{n\rightarrow\infty}\sup_{Q\in\tau^{-1}(H_0)} \P_Q(\psi_n^r(\bX_n) = 1) \leq \alpha_\varphi,
    \end{align*}
    i.e., $\psi_n^r$ satisfies uniform asymptotic level $\alpha_\varphi$ for the hypothesis $\tau(Q^*)\in H_0$.
\end{theorem}

\subsubsection{Hypothesis testing in the observed domain}\label{sec:testing_in_obs_res}
In Section~\ref{sec:testingobserved}, we argue that one can use our  framework for testing under distributional shifts to test a hypothesis in the observed domain, too. Indeed, the results in \cref{subsec:theoretical-guarantees} directly imply the following corollary (see Corollary~\ref{cor:hypothesis-obs-space} in~\cref{sec:hypothesis-obs-space} for a more detailed version).

\begin{corollary}[Pointwise level in the observed domain] \label{cor:hypothesis-obs-space-short}
Consider hypotheses $H_0^{\mathcal{Q}} \subseteq \mathcal{Q}$ and $H_0^{\mathcal{P}} \subseteq \mathcal{P}$ and let $\tau:\cQ \rightarrow \cP$ be a distributional shift such that $\tau(H_0^{\mathcal{Q}}) \subseteq H_0^{\mathcal{P}}$. Under the same assumptions as in Theorem~\ref{thm:asymptotic-level-SIR}, if $\varphi$ is a test that satisfies pointwise asymptotic level in the target domain, $\psi^r_n$ satisfies pointwise asymptotic level for the hypothesis $Q^* \in H_0^{\mathcal{Q}}$.
\end{corollary}

\new{
\subsection{An Alternative for Uniformly Bounded Weights}
\label{subsec:rejection-sampler}
In \cref{subsec:method}, we propose the `distinct replacement' resampling scheme and show in \cref{thm:asymptotic-level-SIR,thm:finite-level-SIR} that this has finite and asymptotic level. 
The procedure requires \cref{assump:finite-second-moment}, that is that the weights have finite second moment.

We now consider the stricter assumption that the weights are globally bounded. Although this assumption is not met for most distributions that are not compactly supported, this is satisfied for example by distributions on finite state spaces. We show that, under this assumption, one can use a rejection sampler with finite sample guarantees. 

Suppose that $\tau(q)(x) \propto r(x)q(x)$ and there exists a known $M\in(0,\infty)$ such that $\sup_x r(x) \leq M$.
Given a sample $\bX_n = (X_1, \ldots, X_n)$ of size $n$ from $Q^*$, we can use a rejection sampler that retains observations $X_i$ with probability $r(X_i)/M$ (and otherwise discards them) to
obtain a sample from $P^* = \tau(Q^*)$, and apply a hypothesis test $\varphi_m$ to the rejection sampled data; see \cref{alg:rejection-sampler}.
\begin{algorithm}[t]
    \caption{Testing a target hypothesis with known distributional shift and rejection sampling}
    \begin{algorithmic}[1]
    \Statex \textbf{Input}: Data $\bX_n$, hypothesis test $\varphi_m$, shift factor $r(x^A)$ and bound $M$.
    \For{$i = 1, 2, \ldots, n$}
        \State Sample $U_i $ uniform on $(0, 1)$
        \If{$U_i > \tfrac{r(X_i)}{M}$}
            \State Discard $X_i$
        \EndIf
    \EndFor
    \Return $\psi_n^r(\bX_n) \coloneqq \varphi_m(X_{i_1}, \ldots, X_{i_m})$
    \end{algorithmic}
    \label{alg:rejection-sampler}
\end{algorithm}

If $\varphi_m$ has level guarantees when applied to data $\bZ_k$ from $P^*$, we can test the hypothesis $\tau(Q^*) \in H_0$ with the same level guarantee, since the rejection sampled data $(X_{i_1}, \ldots, X_{i_m})$ are i.i.d.\ distributed with distribution $P^*$. We state this as a proposition.
\begin{proposition}[Finite level -- bounded weights]\label{prop:rejection-sampler}
    Consider a null hypothesis $H_0 \subseteq \mathcal{P}$ in the target domain. Let $\tau: \cQ \rightarrow \cP$ be a distributional shift for which a known map $r:\mathcal{X}\rightarrow[0,\infty)$ exists, satisfying for all $x$: $\tau(q)(x) \propto r(x)q(x)$ and $r(x) \leq M$.
    Consider an arbitrary $Q \in \cQ$ and $P = \tau(Q)$. 
    Let $\varphi_k$ be a sequence of tests for $H_0$ and assume there exist $\alpha_\varphi\in (0,1)$ such that for each $k\in\N$: $\alpha_\varphi=\sup_{k}\P_{P}(\varphi_k(\bZ_k) = 1)$.
    Let 
    $\psi_n^r(\bX_n)$ be the rejection-sampling test defined in \cref{alg:rejection-sampler}.
    Then it holds that
    \begin{equation*}
        \P_{Q} (\psi^r_{n}(\mathbf{X}_n)=1)=\alpha_\varphi.
    \end{equation*}
\end{proposition}
}

\section{Experiments} \label{sec:expe}
We present a series of simulation experiments that support the theoretical results developed in Section~\ref{subsec:theoretical-guarantees} 
and analyze the underlying assumptions. We also apply the proposed methodology to the problems described in Section~\ref{sec:applications}. 
A simulation  experiment for model selection under covariate shift (see Section~\ref{sec:modselundercovsh}) can be found in Appendix~\ref{sec:model_selection_experiment}. Unless noted otherwise, the experiments use the $\Psi_{\texttt{DRPL}}$ resampling scheme.
Code that reproduces all the experiments is available at \url{https://github.com/nikolajthams/testing-under-shifts}.

\subsection{Exploring assumptions \texorpdfstring{\cref{assump:m-rate-n,assump:finite-second-moment}}{}}
\label{sec:exp-assumptions-a2-a3}
We explore the impact of violating either \cref{assump:m-rate-n}, stating that $m = o(\sqrt{n})$, or \cref{assump:finite-second-moment}, stating that the weights must have finite second moment in the observational distribution. 
To do so, we apply the procedure discussed in \cref{sec:conditionaltesting} that reduces a conditional independence test 
$X \indep Y\,|\, Z$
in the observational domain to an unconditional independence test in the target domain. Specifically, we simulate 
$n = 10'000 $ i.i.d.\ observations from the linear Gaussian model with
\begin{align*}
    X \coloneqq \epsilon_X \qquad 
    Z \coloneqq X + 2\epsilon_Z \qquad 
    Y \coloneqq \theta X + Z + \epsilon_Y 
\end{align*}
for some $\theta\in\R$ and  $\epsilon_X, \epsilon_Z, \epsilon_Y\sim\mathcal{N}(0,1)$ inducing a distribution $Q^*$ over $(X, Y, Z)$. We assume that the conditional distribution $q^*(z | x)$ is known and replace it with an independent Gaussian distribution $\phi_{\sigma}(z)$ with mean zero and variance $\sigma^2$,
breaking the dependence between $X$ and $Z$ in the target distribution.

We then perform a test for independence of $X$ and $Y$ in the target distribution using a Pearson correlation test. 
We do this both for $\theta = 0.4$ (where $X \centernot\indep Y \,|\, Z$ and ideally we reject the hypothesis) and for $\theta = 0$ (where $X \indep Y \,|\, Z$ and ideally we accept the hypothesis).
\cref{fig:a2-a3} shows
the resulting rejection rates of the test, 
where we have repeated the procedure of simulating, resampling and testing (at level $\alpha=0.05$) $500$ (left) or $10'000$ ({right}) times.
In this experiment, the $\Psi_{\texttt{NO-REPL}}$ sampler is used, since the rejection samplers break as $m$ gets very large.
\begin{figure}[t]
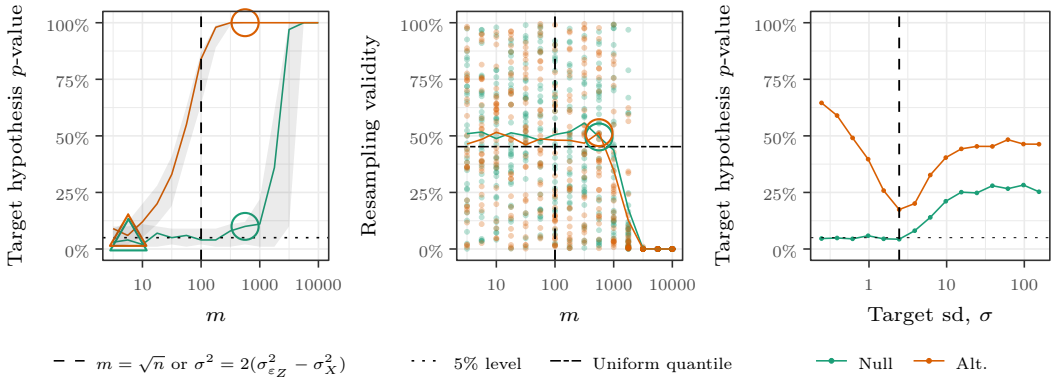

    \centering
    \scriptsize

     \caption{
    {(All)} Rejection rates from the experiments in \cref{sec:exp-assumptions-a2-a3}. We replace the conditional distribution $q^*(z | x)$ with a marginal distribution $\phi_{\sigma}(z)$. 
    We perform a test for independence $X \indep Y$ in the target distribution, and plot the rejection rates. 
    ({Left}) Validation of \cref{assump:m-rate-n}. We run the resampling with $\sigma=1$ and different sample sizes $m$. Shaded regions show $95\%$ quantiles of the rejection rates of the hypothesis test. 
    The level seems to hold for $m \leq \sqrt{n}$, the latter  corresponding  to the asymptotic rate 
    \cref{assump:m-rate-n} (left of the dashed vertical line). Circles indicate the $m$ suggested by the middle plot \new{and triangles the $m$ suggested by the finite-sample method described in \cref{subseq:finite-sample} -- as expected, this is a conservative choice}. The target heuristic suggests an $m$ that indeed yields larger power. 
    {(Middle)} $P$-values (dots) and average $p$-values (lines) when applying \cref{alg:tuning-m} to choose $m$. We select $m$ (circled) from the first time, the goodness-of-fit test (horizontal dashed line) is rejected.
    {(Right)} Validation of \cref{assump:finite-second-moment}. Our procedure is run with different standard deviations~$\sigma$ in the Gaussian target distribution $p_{\sigma}(z)$ $\phi_{\sigma}(z)$. The dashed vertical line indicates the theoretical threshold of $\sqrt{6}$, see Section~\ref{sec:exp-assumptions-a2-a3}.}
    \label{fig:a2-a3}
\end{figure}

First, we test the impact of changing the resampling size $m$. 
For each simulated data set $\bX_n$ and each $m$, we resample $100$ target data sets $\Psi^{r,m}(\bX_n)$, and compute the rates of rejecting the hypothesis. The shaded areas in \cref{fig:a2-a3} (left) indicate $95\%$ of the resulting trajectories, and the solid lines show one example simulation.
Our theoretical results assume $m=o(\sqrt{n})$ and, indeed,
the hypothesis rejects around $5\%$ of simulations when $X\not\rightarrow Y$ ($\theta = 0$) for small $m$. As discussed in \cref{subsec:test-resampling-worked}, \cref{assump:m-rate-n} may in some cases be too strict, and we observe that the $5\%$ level is retained when $m$ moderately exceeds $\sqrt{n}$; but as $m$ grows larger, the level is eventually lost. 

For the same example simulation $\bX_n$ as in the left plot, we also apply the target heuristic for choosing $m$ as described in \cref{alg:tuning-m}, and plot the resulting $p$-values in \cref{fig:a2-a3} (middle).
Since the data are Gaussian, we can perform the goodness-of-fit test by a simple linear regression analysis.
For each $m$, we compute the average of the $p$-values (solid lines), and increase $m$ until the average $p$-value drops below the $5\%$ quantile of the distribution of $\texttt{mean}(U_1, \ldots, U_{\ell})$ where $U_1, \ldots, U_{\ell}$ are i.i.d.\ uniform random variables. The circles in the left and middle plot indicate the $m$ that is chosen by \cref{alg:resampling-and-testing} for this simulation. We observe in the left plot that the power of the test can be increased using the $m$ suggested by the middle plot, while the level approximately holds at 5\%. 

Second, we test the importance of \cref{assump:finite-second-moment}. 
For different $\sigma$ (and fixed $m = \sqrt{n})$, we compute the weights 
$r = \phi_\sigma(Z_i)/q(Z_i|X_i)$,
and in \cref{fig:a2-a3} (right) we plot the rejection rates of the test statistic when $X \rightarrow Y$ ($\theta =0.4$) and $X\not\rightarrow Y$ ($\theta=0$).
We show in \cref{sec:assumption-a3-gaussian} that \cref{assump:finite-second-moment} is satisfied if and only if $\sigma^2 < 2(\sigma_{\epsilon_{2\epsilon_Z}}^2 - \sigma_{X}^2)$, where $\sigma_{X}^2$ is the variance of $X$ and $\sigma_{\epsilon_{2\epsilon_Z}}^2$ is the variance of the noise term in the structural assignment for $Z$.
In this experiment, it follows that \cref{assump:finite-second-moment} holds if and only if $\sigma < \sqrt{6}$. 
We observe that when $\sigma$ exceeds the threshold of $\sqrt{6}$ (vertical dashed line), the level eventually deviates from the $5\%$ level. Furthermore, the power drops when $\sigma$ approaches the threshold.

\subsection{Off-policy testing} \label{sec:exp-off-policy}
We apply our method to perform statistical testing in an off-policy contextual bandit setting as discussed in Section~\ref{sec:offpolicytesting}.
We generate a data set $\bX_n$, $(n=30'000)$, consisting of observations $X_i = (Z_i, A_i, R_i)$ \new{with dimensions $d_Z = 3, d_A = d_R = 1$,} drawn according to the following data generating process:
\begin{align*}
    Z \coloneqq \epsilon_Z \qquad
    A\mid Z \sim q^*(A | Z) \qquad
    R \coloneqq \beta_A^{\top}Z + \epsilon_R,
\end{align*}
where $\epsilon_Z \sim\mathcal{N}(0, I_3)$ and $\epsilon_R \sim \mathcal{N}(0,1)$, $A$ takes values in the action space $\{a_1, \ldots, a_L\}$, \new{where $L=4$,} $q^*(a|z)$ denotes an initial policy that was used to generate the data $\bX_n$ and $\beta_{a_1}, \ldots, \beta_{a_L}$ are parameters of the reward function corresponding to each action. A uniform random policy was used as the initial 
policy, i.e., for all $a \in \{a_1, \dots a_L\}$ and $z \in \mathbb{R}^3$, 
$q^*(a | z) = 1/L$.

The goal is to test hypotheses about the reward $R$ if we were to deploy a target policy $p^*(a|z)$ instead of the policy $q^*(a|z)$.
Here, we consider three hypotheses, namely one-sample test of means, two-sample test of difference in means and two-sample test of difference in distributions. We set the false positive rate to 5\% and use $m=\sqrt{n}$ without the target heuristic in all three experiments. Rejection rates are computed from $500$ repeated simulations.

In the first experiment, we construct different target policies $p_{\delta}^*(a|z)$.
For $\delta = 0$, the target policy reduces to a uniform random policy and with
increasing $\delta$, the policy puts more mass on the optimal action (and thereby increasing the deviation from the initial policy). 
As $\delta \rightarrow \infty$, the target policy converges to an optimal policy.
More precisely, $p_{\delta}^*(a|z)$ is a linear softmax policy, i.e., $p_{\delta}^*(a|z) \propto \exp(\delta \beta_a^\top z)$. 
We then apply our method to non-parametrically
test whether $\E_{P^*_{\delta}}(R) \leq 0$ on the target distribution in which the policy $p_{\delta}^*(a|z)$ is used. For $\delta = 0$, the expected reward is zero (here, the null hypothesis is true) and for increasing $\delta$ the expected reward increases. 
To apply our methodology, we employ the Wilcoxon signed-rank test \citep{wilcoxon1992individual} in the target domain. 
Figure~\ref{fig:off-policy} (left) shows that for $\delta=0$, our method indeed holds the correct level and eventually starts to correctly reject for increasing $\delta$. For comparison, we include an estimate of the expected reward based on IPW.

In the second experiment, we use the same setup as in the first experiment, but now apply the two-sample testing method discussed in Section~\ref{sec:twosampletesting} to test whether $R_{|K=1} \, \overset{\mathcal{L}}{=} \, R_{|K=2}$, where $K=1$ indicates a sample under the initial 
policy and $K=2$ indicates a sample under a target policy. We consider two non-parametric tests, namely a kernel two-sample test based on the maximum mean discrepancy (MMD) \citep{gretton2012kernel} 
\new{(using the Gaussian kernel with the bandwidth chosen by the median heuristic \citep{Sriperumbudur2009})} and the Mann-Whitney (M-W) U test \citep{mann1947test}.
Here, for $\delta = 0$, the two policies coincide and for $\delta > 0$, there is a difference in the expected reward.
As shown in Figure~\ref{fig:off-policy} (middle), both tests are able to detect the difference. The M-W U test has more power than the MMD test.

In a third experiment, we construct different target policies $p_{\delta'}^*(a|z)$ by varying their effect on the variance of the reward distribution, while keeping the mean unchanged. More specifically, $p_{\delta'}^*(a|z)$ is a weighted random policy, i.e., $p_{\delta'}^*(a|z) \propto
\delta' \mbox{ if } a = a_1$ and $\propto 1 \mbox{ otherwise}$. This target policy yields the same expected reward as the initial policy (a uniform random policy), but yields a different variance of the reward. When $\delta'=1$, the target policy is the same as the initial policy, whereas the variance of the reward becomes smaller when $\delta'$ increases (in Figure~\ref{fig:off-policy} (right), $\delta'$ is rescaled to 0–1 range). We then apply the same two-sample testing methods used in the second experiment to test whether $R_{|K=1} \, \overset{\mathcal{L}}{=} \, R_{|K=2}$. 
This difference is not picked up by the M-W U test and this time, the MMD test has more power, see Figure~\ref{fig:off-policy} (right). 
\begin{figure}[t]
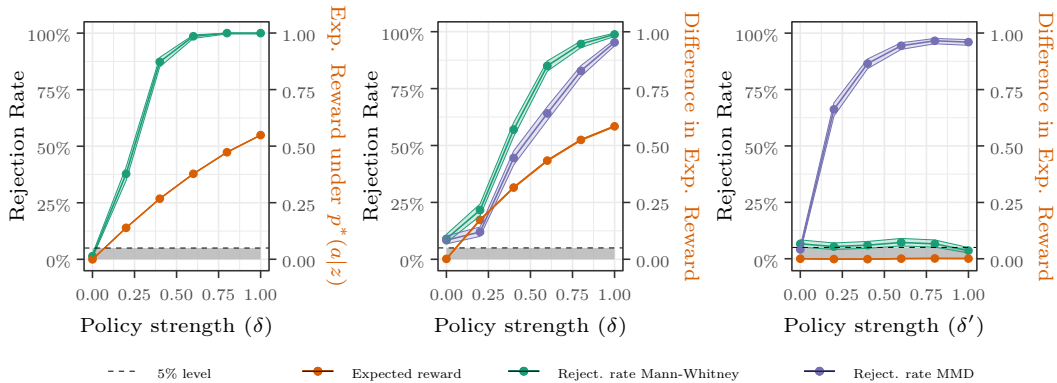

    \centering
    \scriptsize

          \label{fig:off-policy1}
    \caption{
Off-policy statistical testing as described in \cref{sec:exp-off-policy}. ({Left}) One-sample test for testing whether the mean under $p_{\delta}^*(a|z)$ is less than or equal to $0$, see Section~\ref{sec:offpolicytesting}. {(Middle, Right)} Two-sample tests for testing whether the reward under $p^*_{\delta}(a|z)$ and $p^*_{\delta'}(a|z)$, respectively, has a different distribution than the reward under the initial policy, see  Section~\ref{sec:twosampletesting}.
In all cases, the null hypothesis is true for $\delta = 0$.
The target policies affect the mean of the reward in the left and middle plot, whereas they affect its variance in the right plot.
The framework can be combined with non-parametric tests and thereby allows for detecting complex differences in the reward distribution when comparing two policies.}
\label{fig:off-policy}
\end{figure}

\subsection{Testing a conditional independence with a complex conditional}\label{exp:cond-ind-test}
In the setting of conditional independence testing, we now compare our method -- when turning the problem into a test for unconditional independence as discussed in \cref{sec:conditionaltesting} -- to existing conditional independence tests.
We sample $n = 150$ observations from the following structural causal model
\begin{align*}
    X \coloneqq \texttt{GaussianMixture}(-2, 2) \quad
    Z \coloneqq -X^2 + \epsilon_Z\quad
    Y \coloneqq \sin(Z) + \theta X^\tau + \epsilon_Y, 
\end{align*}
inducing a distribution $Q^*$,
where \texttt{GaussianMixture}(-2, 2) is an even mixture (i.e., $p = 0.5$) of two Gaussian distributions with means $\mu_1=-2$, $\mu_2=2$ and unit variances,
$\epsilon_Z, \epsilon_Y$ are independent $\mathcal{N}(0, 4)$-variables and $\theta \in [0, 3/2]$, $\tau \in \{1, 2\}$.
\begin{figure}[t]
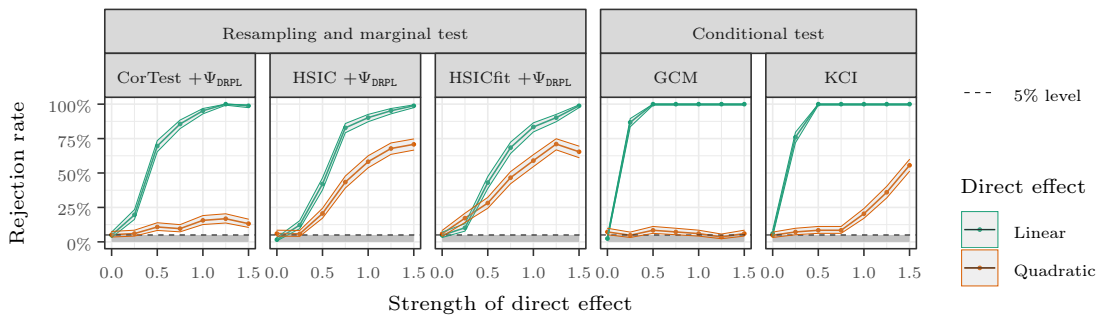

    \centering
    \scriptsize

     \caption{Rejection rates for conditional independence tests $X \indep Y \mid Z$ in the setting from \cref{exp:cond-ind-test}. The direct effect of $X$ on $Y$ is $\theta X^\tau$, where the exponent $\tau$ is either $1$ (green) or $2$ (orange), and the $\theta$ is shown on the $x$-axis. 
    The three left panels show our method of resampling  and testing marginal independence. 
    Here, we  combine our approach with a Pearson correlation test or HSIC. HSICfit indicates that we estimate the conditional $q^*(z|x)$ from data. 
    We compare to two conditional independence tests, GCM  and KCI. As expected, GCM cannot detect the nonlinear dependence, whereas the resampling approach can be combined with any independence test. Despite its simplicity, it is able to quickly reject the alternatives.  
    }
    \label{fig:experiment-cit}
\end{figure}
Considering the conditional $q^*(z|x)$
to be known, we apply our methodology 
for testing conditional independence $X \indep Y \,|\, Z$ with a $5\%$ level and using the target heuristic in \cref{alg:tuning-m}.
To do so, we replace $q^*(z|x)$ by a marginal density $\phi(z)$, which is Gaussian with mean and variance set to the empirical versions under $Q^*$. 
In the target distribution, we test for independence of $X$ and $Y$ using either a simple correlation test (CorTest) 
or a kernel independence test (HSIC) \citep{gretton2007kernel}.
For comparison, we also conduct conditional independence tests in the observable distribution, using the generalized covariance measure (GCM) by \citet{shah2020hardness} and a kernel conditional independence (KCI) by \citet{zhang2012kernel} (both using standard versions, without hyperparameter tuning). 
Our resampling methods use knowledge of the conditional $q^*(z|x)$,
which may be seen as an unfair advantage over the conditional independence tests. 
Therefore, we also apply our method with estimated weights, called HSICfit, where the conditional $q^*(z|x)$ is estimated using a generalized additive model. 

We repeat the experiment $500$ times and plot the rejection rates in \cref{fig:experiment-cit} at various strengths $\theta$ of the edge $X \rightarrow Y$. All instances of our method have the correct level, see rejection rates for $\theta=0$. When $\tau = 1$, i.e., the direct effect $X \rightarrow Y$ is linear, the power of our method approaches $100\%$ as the causal effect increases, albeit the conditional independence tests obtain power more quickly.
When the direct effect is quadratic, CorTest and GCM have little or no power, as expected since they are based on correlations 
(we believe that the slight deviation from $5\%$ level in the left plot is due to very small sample sizes and the heuristic choice of $m$).
KCI and HSIC have comparable power in the quadratic case, with our approach even obtaining slightly more power than KCI.
Our approach has the additional benefit of low computational costs: 
Conditional independence testing is usually a more complicated procedure than marginal independence testing and, furthermore, the marginal test is applied to a data set of size $m$, which by \cref{assump:m-rate-n} is chosen much smaller than $n$.

\subsection{Testing dormant independences} \label{sec:expverma}
We now employ our method to test a dormant independence from observational data, as described in \cref{sec:dormant-indep-theory}.
We simulate data from a distribution $Q^*$ that factorizes according to the graph $\mathcal{H}$ in \cref{fig:vermaL} and test the existence of the edge $X^1 \rightarrow X^4$. As discussed by \citet{shpitser2008dormant}, the presence of this edge cannot be tested by a conditional independence test, and instead we test marginal independence between $X^1$ and $X^4$ in the target distribution $Q^{\DO(X^3 \coloneqq N)}$, which can be obtained by applying our method using 
$r_q(x^3, x^2) \coloneqq q^{\DO}(x^3)/q(x^3|x^2)$.

More precisely, we conduct three experiments. In the first experiment, we consider binary random variables for the observables $X^1, X^2, X^3$ and $X^4$, while the hidden variable $H$ is a discrete random variable with 4 possible values. We estimate $q(x^3 | x^2)$ by the empirical probabilities and use the empirical marginal distribution of $X^3$ as a target distribution, i.e., $p(x^3) = \hat{q}(x^3)$. In this setting, using the marginal distribution of $X^3$ as a target distribution corresponds to minimizing the empirical variance of the weights $r_{q}(x^3, x^2)$ (not shown), see~\cref{assump:finite-second-moment}.
We employ Fisher's exact test to determine whether $X^1$ and $X^4$ in the interventional distribution are independent of each other. We compare our method to a 
more specialized method based on binary nested Markov models \citep{shpitser2012parameter} that is based on a likelihood ratio test. \Cref{sec:sim-details} contains simulation parameters from all three experiments.

In the second experiment, we consider a linear Gaussian SCM. We estimate the conditional $q(x^3 | x^2)$ by a linear regression and, as before, use the empirical marginal distribution of $X^3$ as a target distribution. We then test for independence between $X^1$ and $X^4$ in the interventional distribution using a simple correlation test. 
Since the distribution is jointly Gaussian (with linear functions) and satisfies the `bow-free' condition \citep{brito2002new}, there is a specialized, non-trivial likelihood procedure for model selection that we can compare with: 
We perform maximum likelihood estimation as suggested by~\citet{drton2009computing} and use the penalty from \citet{nowzohour2017distributional} to score the graphs $\mathcal{G}$ and $\mathcal{H}$. 

In the third experiment, we consider a nonlinear SCM with non-Gaussian errors. We estimate the conditional $q(x^3|x^2)$ using generalized additive models. For simplicity, we consider a distribution where $X^3 - \E[X^3\mid X^2]$ is Gaussian. Our procedure also applies to
more general settings by applying conditional density estimation to learn $q(x^3|x^2)$, for example.
To the best of our knowledge, there exist no other methods for testing the dormant independence in any of such cases.

In all experiments, we consider two strategies for choosing the resampling size $m$: (1) $m = \sqrt{n}$ and (2) the target heuristic (see Algorithm~\ref{alg:tuning-m}). The resulting rejection rates over 500 repeated experiments, for several sample sizes, are shown in \cref{fig:dormant-indep-binary}.
Our method identifies both the absence and presence of the causal edge $X^1 \rightarrow X^4$ in both the binary and the Gaussian setting.
In both the binary and Gaussian settings, the tailored score-based approaches have more power to detect the absence of the edge (though in the binary case, the level of the test does not seem to hold exactly when the sample sizes are small). 
For the general case (nonlinear and non-Gaussian), our method has the correct level and 
increasing power as sample size increases. We are not aware of any other existing test that can achieve this in general.
Compared to $m = \sqrt{n}$, the choice of $m$ with the target heuristic yields larger test power without sacrificing too much the level of the test (although the level is violated for small sample sizes in the Gaussian setting).

\begin{figure}[t]
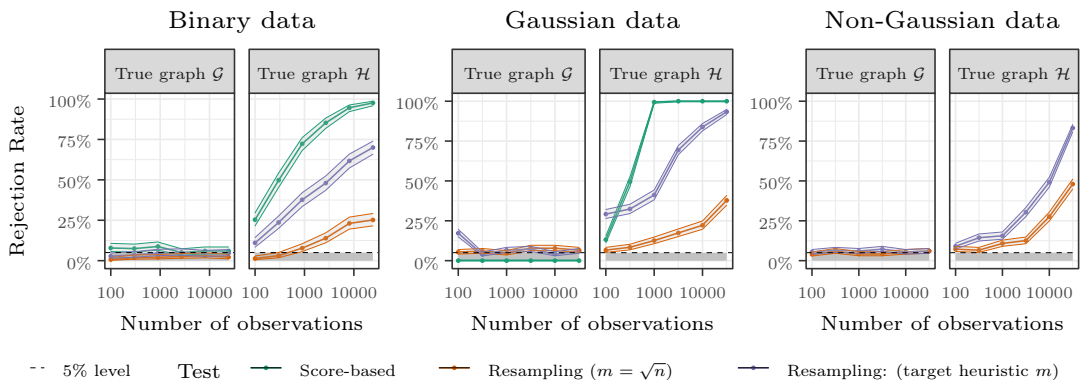

    \scriptsize
    \centering

         \caption{
        We test the dormant independence discussed in \cref{sec:dormant-indep-theory}.
        The existence of the edge $X^1 \rightarrow X^4$ 
        can be inferred by testing a Verma constraint in the observational distribution, which translates to an independence statement in the target distribution. 
        The plots show rejection rates for the underlying graphs $\mathcal{G}$ (first, third and fifth plot) and $\mathcal{H}$ (second, fourth and sixth plot), see Figure~\ref{fig:vermaL}, as the number of observations increases.
        Tailored  score-based approaches exist in the binary and Gaussian cases (left and middle) but not in the case of  more complex distributions (right). We plot the rates both when $m=\sqrt{n}$ and when $m$ is chosen according to the target heuristic in \cref{alg:tuning-m}. The target heuristic chooses an $m$ that yields a good balance between level (graph $\mathcal{G}$) and power (graph $\mathcal{H}$).
        } 
    \label{fig:dormant-indep-binary}
\end{figure}

\subsection{Comparison to IPW}\label{sec:ipw}
Inverse probability weighting (IPW) allows us to test simple hypotheses such as 
$\E_P[f(X)] = c$ for some constant $c \in \mathbb{R}$ and a given function $f$.
If data $\bZ_m$ sampled from the target distribution $P^*$ are available, we could test the hypothesis using the test statistic 
$\sum_{i=1}^m f(Z_i)$.
If, instead, data are available from an observable distribution $Q^*$, we can estimate the corresponding test statistic in the target domain using the test statistic
\begin{align*}
    T(\bX_n) \coloneqq \sum_{i=1}^n \bar{r}(X_i) f(X_i),
\end{align*}
where $\bar{r}(X_i)$ is the normalized versions of the shift factor $r$ (elsewhere we do not require $r$ to be normalized). 
Under \cref{assump:finite-second-moment}, 
$T(\bX_n)$ is asymptotically normal with mean $\E_P[f(X_1)]$ and variance $\sigma^2 \coloneqq {\VAR(r(X_1)f(X_1))}/{\sqrt{n}}$, and one can construct a $(1-\alpha)$ confidence interval as 
$$[T(\bX_n) - z_{\alpha/2} \sigma/\sqrt{n}, T(\bX_n) + z_{\alpha/2} \sigma/\sqrt{n}],$$ where $z_{\alpha/2}$ is the $\alpha/2$ quantile from the standard normal distribution. 

To compare our approach to the IPW approach, we simulate data ($n=100$) from the following structural equation model
\begin{align*}
    X^1 \coloneqq 1 + \epsilon_{X^1} \quad X^2 \coloneqq X^1 + \epsilon_{X^2} \quad X^3 \coloneqq X^2 - X^1 + \epsilon_{X^3},
\end{align*}
with $\epsilon_{X^1}\sim\mathcal{N}(0, 3)$, $\epsilon_{X^2}\sim\mathcal{N}(0, 4)$ and $\epsilon_{X^3}\sim\mathcal{N}(0, 1)$.
In this model, the mean of $X^3$ is $\E_Q[X^3] = 0$. 
We consider the distributional shift corresponding to the intervention
$\DO(X^2 \coloneqq \mu + \widebar{\epsilon}_{X^2})$ with $\widebar{\epsilon}_{X^2}\sim\mathcal{N}(0,1)$
, where $X^3$ has mean $\E_P[X^3] = \mu - 1$, and test the hypothesis $\E_P[X^3] = 0$ for various $\mu$ using both our resampling approach (with $m$ chosen according to the target heuristic in \cref{alg:tuning-m}) and the IPW based confidence intervals. Since IPW is sensitive to degenerate weights, we also use a `clipped IPW', where we truncate the $10$ largest weights at the $10$th largest value (see e.g. \citet{cole2008constructing}).

Ideally, we accept the hypothesis for $\mu=1$ and reject the hypothesis for all other $\mu$. The larger $\mu$ becomes, the easier it should be to reject the hypothesis $\mu = 1$, if target data are available.
At the same time, since the target distribution is a Gaussian distribution centered at $\mu-1$, as $\mu$ increases, the weights get increasingly degenerate, because the weights of the  data points with the largest numerical values $X^2$ dominate the weights of all other data points. 

We observe in \cref{fig:ipw-clt} that all methods have the correct level at 5\% (when $\mu = 1$) and approximately the same power for small $\mu$. As $\mu$ grows, the plain IPW loses its power, due to weight degeneracy.
Both the clipped IPW and our resampling approach do not suffer from this issue, with power approaching $1$, even as weights get increasingly degenerate. 
This experiment indicates that our method may share some of the robustness to degenerate weights that is known from clipped IPW, and at the same time is able to estimate more complex test statistics, that cannot be estimated using IPW.

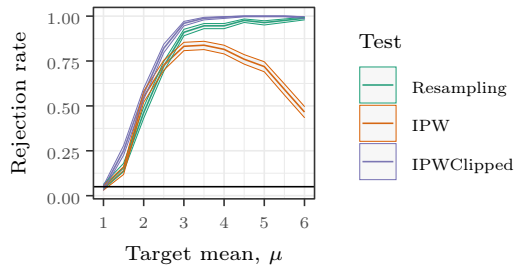
\begin{figure}[t]
    \centering
    \scriptsize
\begin{tikzpicture}[x=1pt,y=1pt]
\definecolor{fillColor}{RGB}{255,255,255}
\begin{scope}
\definecolor{drawColor}{RGB}{255,255,255}
\definecolor{fillColor}{RGB}{255,255,255}

\path[draw=drawColor,line width= 0.6pt,line join=round,line cap=round,fill=fillColor] (  0.00,  0.00) rectangle (216.81,108.41);
\end{scope}
\begin{scope}
\definecolor{fillColor}{RGB}{255,255,255}

\path[fill=fillColor] ( 38.56, 30.69) rectangle (121.33,102.91);
\definecolor{drawColor}{gray}{0.92}

\path[draw=drawColor,line width= 0.3pt,line join=round] ( 38.56, 40.35) --
	(121.33, 40.35);

\path[draw=drawColor,line width= 0.3pt,line join=round] ( 38.56, 57.28) --
	(121.33, 57.28);

\path[draw=drawColor,line width= 0.3pt,line join=round] ( 38.56, 74.22) --
	(121.33, 74.22);

\path[draw=drawColor,line width= 0.3pt,line join=round] ( 38.56, 91.15) --
	(121.33, 91.15);

\path[draw=drawColor,line width= 0.3pt,line join=round] ( 49.84, 30.69) --
	( 49.84,102.90);

\path[draw=drawColor,line width= 0.3pt,line join=round] ( 64.89, 30.69) --
	( 64.89,102.90);

\path[draw=drawColor,line width= 0.3pt,line join=round] ( 79.94, 30.69) --
	( 79.94,102.90);

\path[draw=drawColor,line width= 0.3pt,line join=round] ( 95.00, 30.69) --
	( 95.00,102.90);

\path[draw=drawColor,line width= 0.3pt,line join=round] (110.05, 30.69) --
	(110.05,102.90);

\path[draw=drawColor,line width= 0.6pt,line join=round] ( 38.56, 31.88) --
	(121.33, 31.88);

\path[draw=drawColor,line width= 0.6pt,line join=round] ( 38.56, 48.82) --
	(121.33, 48.82);

\path[draw=drawColor,line width= 0.6pt,line join=round] ( 38.56, 65.75) --
	(121.33, 65.75);

\path[draw=drawColor,line width= 0.6pt,line join=round] ( 38.56, 82.69) --
	(121.33, 82.69);

\path[draw=drawColor,line width= 0.6pt,line join=round] ( 38.56, 99.62) --
	(121.33, 99.62);

\path[draw=drawColor,line width= 0.6pt,line join=round] ( 42.32, 30.69) --
	( 42.32,102.90);

\path[draw=drawColor,line width= 0.6pt,line join=round] ( 57.37, 30.69) --
	( 57.37,102.90);

\path[draw=drawColor,line width= 0.6pt,line join=round] ( 72.42, 30.69) --
	( 72.42,102.90);

\path[draw=drawColor,line width= 0.6pt,line join=round] ( 87.47, 30.69) --
	( 87.47,102.90);

\path[draw=drawColor,line width= 0.6pt,line join=round] (102.52, 30.69) --
	(102.52,102.90);

\path[draw=drawColor,line width= 0.6pt,line join=round] (117.57, 30.69) --
	(117.57,102.90);
\definecolor{drawColor}{RGB}{27,158,119}

\path[draw=drawColor,line width= 0.6pt,line join=round] ( 42.32, 35.13) --
	( 49.84, 42.45) --
	( 57.37, 63.18) --
	( 64.89, 80.72) --
	( 72.42, 93.53) --
	( 79.94, 95.96) --
	( 87.47, 95.96) --
	( 95.00, 98.00) --
	(102.52, 97.18) --
	(110.05, 98.00) --
	(117.57, 98.81);
\definecolor{drawColor}{RGB}{217,95,2}

\path[draw=drawColor,line width= 0.6pt,line join=round] ( 42.32, 34.72) --
	( 49.84, 41.23) --
	( 57.37, 69.21) --
	( 64.89, 81.13) --
	( 72.42, 88.24) --
	( 79.94, 88.65) --
	( 87.47, 87.09) --
	( 95.00, 83.36) --
	(102.52, 80.59) --
	(110.05, 71.98) --
	(117.57, 63.45);
\definecolor{drawColor}{RGB}{117,112,179}

\path[draw=drawColor,line width= 0.6pt,line join=round] ( 42.32, 34.72) --
	( 49.84, 48.88) --
	( 57.37, 70.29) --
	( 64.89, 87.50) --
	( 72.42, 96.91) --
	( 79.94, 98.81) --
	( 87.47, 99.22) --
	( 95.00, 99.62) --
	(102.52, 99.62) --
	(110.05, 99.62) --
	(117.57, 99.22);
\definecolor{fillColor}{RGB}{179,179,179}

\path[fill=fillColor,fill opacity=0.10] ( 42.32, 36.17) --
	( 49.84, 44.07) --
	( 57.37, 65.28) --
	( 64.89, 82.56) --
	( 72.42, 94.65) --
	( 79.94, 96.82) --
	( 87.47, 96.82) --
	( 95.00, 98.55) --
	(102.52, 97.87) --
	(110.05, 98.55) --
	(117.57, 99.17) --
	(117.57, 98.21) --
	(110.05, 97.23) --
	(102.52, 96.28) --
	( 95.00, 97.23) --
	( 87.47, 94.89) --
	( 79.94, 94.89) --
	( 72.42, 92.21) --
	( 64.89, 78.76) --
	( 57.37, 61.08) --
	( 49.84, 40.98) --
	( 42.32, 34.32) --
	cycle;
\definecolor{drawColor}{RGB}{27,158,119}

\path[draw=drawColor,line width= 0.3pt,line join=round,line cap=round] ( 42.32, 36.17) --
	( 49.84, 44.07) --
	( 57.37, 65.28) --
	( 64.89, 82.56) --
	( 72.42, 94.65) --
	( 79.94, 96.82) --
	( 87.47, 96.82) --
	( 95.00, 98.55) --
	(102.52, 97.87) --
	(110.05, 98.55) --
	(117.57, 99.17);

\path[draw=drawColor,line width= 0.3pt,line join=round,line cap=round] (117.57, 98.21) --
	(110.05, 97.23) --
	(102.52, 96.28) --
	( 95.00, 97.23) --
	( 87.47, 94.89) --
	( 79.94, 94.89) --
	( 72.42, 92.21) --
	( 64.89, 78.76) --
	( 57.37, 61.08) --
	( 49.84, 40.98) --
	( 42.32, 34.32);

\path[fill=fillColor,fill opacity=0.10] ( 42.32, 35.70) --
	( 49.84, 42.78) --
	( 57.37, 71.31) --
	( 64.89, 82.96) --
	( 72.42, 89.74) --
	( 79.94, 90.13) --
	( 87.47, 88.65) --
	( 95.00, 85.10) --
	(102.52, 82.42) --
	(110.05, 74.05) --
	(117.57, 65.55) --
	(117.57, 61.35) --
	(110.05, 69.88) --
	(102.52, 78.63) --
	( 95.00, 81.51) --
	( 87.47, 85.37) --
	( 79.94, 87.00) --
	( 72.42, 86.59) --
	( 64.89, 79.20) --
	( 57.37, 67.11) --
	( 49.84, 39.87) --
	( 42.32, 33.97) --
	cycle;
\definecolor{drawColor}{RGB}{217,95,2}

\path[draw=drawColor,line width= 0.3pt,line join=round,line cap=round] ( 42.32, 35.70) --
	( 49.84, 42.78) --
	( 57.37, 71.31) --
	( 64.89, 82.96) --
	( 72.42, 89.74) --
	( 79.94, 90.13) --
	( 87.47, 88.65) --
	( 95.00, 85.10) --
	(102.52, 82.42) --
	(110.05, 74.05) --
	(117.57, 65.55);

\path[draw=drawColor,line width= 0.3pt,line join=round,line cap=round] (117.57, 61.35) --
	(110.05, 69.88) --
	(102.52, 78.63) --
	( 95.00, 81.51) --
	( 87.47, 85.37) --
	( 79.94, 87.00) --
	( 72.42, 86.59) --
	( 64.89, 79.20) --
	( 57.37, 67.11) --
	( 49.84, 39.87) --
	( 42.32, 33.97);

\path[fill=fillColor,fill opacity=0.10] ( 42.32, 35.70) --
	( 49.84, 50.78) --
	( 57.37, 72.36) --
	( 64.89, 89.04) --
	( 72.42, 97.63) --
	( 79.94, 99.17) --
	( 87.47, 99.45) --
	( 95.00, 99.62) --
	(102.52, 99.62) --
	(110.05, 99.62) --
	(117.57, 99.45) --
	(117.57, 98.75) --
	(110.05, 99.37) --
	(102.52, 99.37) --
	( 95.00, 99.37) --
	( 87.47, 98.75) --
	( 79.94, 98.21) --
	( 72.42, 95.94) --
	( 64.89, 85.81) --
	( 57.37, 68.19) --
	( 49.84, 47.12) --
	( 42.32, 33.97) --
	cycle;
\definecolor{drawColor}{RGB}{117,112,179}

\path[draw=drawColor,line width= 0.3pt,line join=round,line cap=round] ( 42.32, 35.70) --
	( 49.84, 50.78) --
	( 57.37, 72.36) --
	( 64.89, 89.04) --
	( 72.42, 97.63) --
	( 79.94, 99.17) --
	( 87.47, 99.45) --
	( 95.00, 99.62) --
	(102.52, 99.62) --
	(110.05, 99.62) --
	(117.57, 99.45);

\path[draw=drawColor,line width= 0.3pt,line join=round,line cap=round] (117.57, 98.75) --
	(110.05, 99.37) --
	(102.52, 99.37) --
	( 95.00, 99.37) --
	( 87.47, 98.75) --
	( 79.94, 98.21) --
	( 72.42, 95.94) --
	( 64.89, 85.81) --
	( 57.37, 68.19) --
	( 49.84, 47.12) --
	( 42.32, 33.97);
\definecolor{drawColor}{RGB}{0,0,0}

\path[draw=drawColor,line width= 0.6pt,line join=round] ( 38.56, 35.27) -- (121.33, 35.27);
\definecolor{drawColor}{gray}{0.20}

\path[draw=drawColor,line width= 0.6pt,line join=round,line cap=round] ( 38.56, 30.69) rectangle (121.33,102.91);
\end{scope}
\begin{scope}
\definecolor{drawColor}{gray}{0.30}

\node[text=drawColor,anchor=base east,inner sep=0pt, outer sep=0pt, scale=  0.88] at ( 33.61, 28.85) {0.00};

\node[text=drawColor,anchor=base east,inner sep=0pt, outer sep=0pt, scale=  0.88] at ( 33.61, 45.79) {0.25};

\node[text=drawColor,anchor=base east,inner sep=0pt, outer sep=0pt, scale=  0.88] at ( 33.61, 62.72) {0.50};

\node[text=drawColor,anchor=base east,inner sep=0pt, outer sep=0pt, scale=  0.88] at ( 33.61, 79.66) {0.75};

\node[text=drawColor,anchor=base east,inner sep=0pt, outer sep=0pt, scale=  0.88] at ( 33.61, 96.59) {1.00};
\end{scope}
\begin{scope}
\definecolor{drawColor}{gray}{0.20}

\path[draw=drawColor,line width= 0.6pt,line join=round] ( 35.81, 31.88) --
	( 38.56, 31.88);

\path[draw=drawColor,line width= 0.6pt,line join=round] ( 35.81, 48.82) --
	( 38.56, 48.82);

\path[draw=drawColor,line width= 0.6pt,line join=round] ( 35.81, 65.75) --
	( 38.56, 65.75);

\path[draw=drawColor,line width= 0.6pt,line join=round] ( 35.81, 82.69) --
	( 38.56, 82.69);

\path[draw=drawColor,line width= 0.6pt,line join=round] ( 35.81, 99.62) --
	( 38.56, 99.62);
\end{scope}
\begin{scope}
\definecolor{drawColor}{gray}{0.20}

\path[draw=drawColor,line width= 0.6pt,line join=round] ( 42.32, 27.94) --
	( 42.32, 30.69);

\path[draw=drawColor,line width= 0.6pt,line join=round] ( 57.37, 27.94) --
	( 57.37, 30.69);

\path[draw=drawColor,line width= 0.6pt,line join=round] ( 72.42, 27.94) --
	( 72.42, 30.69);

\path[draw=drawColor,line width= 0.6pt,line join=round] ( 87.47, 27.94) --
	( 87.47, 30.69);

\path[draw=drawColor,line width= 0.6pt,line join=round] (102.52, 27.94) --
	(102.52, 30.69);

\path[draw=drawColor,line width= 0.6pt,line join=round] (117.57, 27.94) --
	(117.57, 30.69);
\end{scope}
\begin{scope}
\definecolor{drawColor}{gray}{0.30}

\node[text=drawColor,anchor=base,inner sep=0pt, outer sep=0pt, scale=  0.88] at ( 42.32, 19.68) {1};

\node[text=drawColor,anchor=base,inner sep=0pt, outer sep=0pt, scale=  0.88] at ( 57.37, 19.68) {2};

\node[text=drawColor,anchor=base,inner sep=0pt, outer sep=0pt, scale=  0.88] at ( 72.42, 19.68) {3};

\node[text=drawColor,anchor=base,inner sep=0pt, outer sep=0pt, scale=  0.88] at ( 87.47, 19.68) {4};

\node[text=drawColor,anchor=base,inner sep=0pt, outer sep=0pt, scale=  0.88] at (102.52, 19.68) {5};

\node[text=drawColor,anchor=base,inner sep=0pt, outer sep=0pt, scale=  0.88] at (117.57, 19.68) {6};
\end{scope}
\begin{scope}
\definecolor{drawColor}{RGB}{0,0,0}

\node[text=drawColor,anchor=base,inner sep=0pt, outer sep=0pt, scale=  1.10] at ( 79.94,  7.64) {Target mean, $\mu$};
\end{scope}
\begin{scope}
\definecolor{drawColor}{RGB}{0,0,0}

\node[text=drawColor,rotate= 90.00,anchor=base,inner sep=0pt, outer sep=0pt, scale=  1.10] at ( 13.08, 66.80) {Rejection rate};
\end{scope}
\begin{scope}
\definecolor{fillColor}{RGB}{255,255,255}

\path[fill=fillColor] (132.33, 32.01) rectangle (211.31,101.58);
\end{scope}
\begin{scope}
\definecolor{drawColor}{RGB}{0,0,0}

\node[text=drawColor,anchor=base west,inner sep=0pt, outer sep=0pt, scale=  1.10] at (137.83, 87.44) {Test};
\end{scope}
\begin{scope}
\definecolor{fillColor}{RGB}{255,255,255}

\path[fill=fillColor] (137.83, 66.42) rectangle (152.29, 80.87);
\end{scope}
\begin{scope}
\definecolor{drawColor}{RGB}{27,158,119}

\path[draw=drawColor,line width= 0.6pt,line join=round] (139.28, 73.64) -- (150.84, 73.64);
\end{scope}
\begin{scope}
\definecolor{drawColor}{RGB}{27,158,119}
\definecolor{fillColor}{RGB}{179,179,179}

\path[draw=drawColor,line width= 0.3pt,line cap=rect,fill=fillColor,fill opacity=0.10] (138.26, 66.84) rectangle (151.86, 80.44);
\end{scope}
\begin{scope}
\definecolor{fillColor}{RGB}{255,255,255}

\path[fill=fillColor] (137.83, 51.96) rectangle (152.29, 66.42);
\end{scope}
\begin{scope}
\definecolor{drawColor}{RGB}{217,95,2}

\path[draw=drawColor,line width= 0.6pt,line join=round] (139.28, 59.19) -- (150.84, 59.19);
\end{scope}
\begin{scope}
\definecolor{drawColor}{RGB}{217,95,2}
\definecolor{fillColor}{RGB}{179,179,179}

\path[draw=drawColor,line width= 0.3pt,line cap=rect,fill=fillColor,fill opacity=0.10] (138.26, 52.39) rectangle (151.86, 65.99);
\end{scope}
\begin{scope}
\definecolor{fillColor}{RGB}{255,255,255}

\path[fill=fillColor] (137.83, 37.51) rectangle (152.29, 51.96);
\end{scope}
\begin{scope}
\definecolor{drawColor}{RGB}{117,112,179}

\path[draw=drawColor,line width= 0.6pt,line join=round] (139.28, 44.73) -- (150.84, 44.73);
\end{scope}
\begin{scope}
\definecolor{drawColor}{RGB}{117,112,179}
\definecolor{fillColor}{RGB}{179,179,179}

\path[draw=drawColor,line width= 0.3pt,line cap=rect,fill=fillColor,fill opacity=0.10] (138.26, 37.93) rectangle (151.86, 51.53);
\end{scope}
\begin{scope}
\definecolor{drawColor}{RGB}{0,0,0}

\node[text=drawColor,anchor=base west,inner sep=0pt, outer sep=0pt, scale=  0.88] at (157.79, 70.61) {Resampling};
\end{scope}
\begin{scope}
\definecolor{drawColor}{RGB}{0,0,0}

\node[text=drawColor,anchor=base west,inner sep=0pt, outer sep=0pt, scale=  0.88] at (157.79, 56.16) {IPW};
\end{scope}
\begin{scope}
\definecolor{drawColor}{RGB}{0,0,0}

\node[text=drawColor,anchor=base west,inner sep=0pt, outer sep=0pt, scale=  0.88] at (157.79, 41.70) {IPWClipped};
\end{scope}
\end{tikzpicture}
     \caption{Rejection rates  
    for testing whether the mean of $X^3$ equals zero, which is the case
    if and only if $\mu = 1$, see Section~\ref{sec:ipw}. As $\mu$ increases,
    the weights become more and more ill-behaved and IPW loses some of its power.
    Neither the clipped version of  IPW nor the resampling framework  suffer from this problem.
    }
    \label{fig:ipw-clt}
\end{figure}

\subsection{Resampling for heterogeneity to 
identify causal predictors} \label{sec:icpplu}
In \cref{icpplus}, we propose to use our resampling approach to create heterogeneous data
and apply ICP \citep{Peters2016jrssb} to estimate (a subset of) the causal predictors $X^{\PA_Y}$ of a response variable $Y$, even when no environments are given.  We now illustrate that this approach can infer casual relationships that would not be detectable using conditional independence statements.
We therefore generate $n=1'000$ i.i.d.\ observations of $Y, X^1, X^2,X^3,X^4$ according to a linear Gaussian SCM with the graphical representation given in \cref{fig:icp_graph_res} (left). Furthermore, we assume that the conditional distribution $q^*(x^2|x^1, x^3)$ is known (instead, one could also assume that $\PA_2=\{1, 3\}$ is known and estimate the conditional).
As described in \cref{icpplus}, we now generate two environments by considering the observational distribution and a modified distribution based on a distributional shift. Specifically, we take the entire sample to form environment $K=1$ and then, to form environment $K=2$, we resample from the same data $m=30$ (approximately $\sqrt{n}$) observations under the distributional shift generated by replacing the conditional $q^*(x^2|x^1, x^3)$ with the target distribution $p^*(x^2|x^1, x^3)$ (which flips the sign of the dependence on $x^3$). 
The precise data generating process is described in \cref{sec:sim-details}.
This results in a data set with $n+m$ observations from two environments. We then apply ICP to the joint data from both environments and output the following estimate 
of the causal predictors:
\begin{align*}
    \hat{S} \coloneqq \bigcap_{S: H_{0,S} \text{ accepted}} S,
\end{align*}
Here,  $H_{0, S}$ is the hypothesis defined in \cref{icpplus}. We use the \texttt{InvariantCausalPrediction} R-package for this experiment, which tests $H_{0,S}$ using a Chow test \citep{chow1960}.
We repeat the experiment $500$ times and report in \cref{fig:icp_graph_res} (right) how many times each set $S$ is output 
As an oracle benchmark, we also report the corresponding frequencies when we sample the target distribution directly, instead of resampling it (in particular we use the same total sample size $m+n$). 
Our method frequently returns the invariant set $\{2, 3\}$ and holds the predicted coverage guarantee: in only $4.2\%$ of the cases, the estimated set is not a subset of $\{2,3\}$.

The output of the method is guaranteed to be  (with large probability) a subset of the set of  true causal predictors, but depending on the type of heterogeneity, the method may output the empty set. E.g., if the true (unknown) underlying graph equals $X^4\rightarrow Y\rightarrow X^1\rightarrow X^2\leftarrow X^3$, then (for the same experiment), both ICP based on the resampled data and ICP based on the true target distribution always output the empty set.

The difference between the oracle method and the resampling method (see~\cref{fig:icp_graph_res}) indicates 
that the resampled distribution does not equal the target distribution. 
Indeed, in some regions where the  target density has substantial mass, there are no data points that can be sampled. 
This, however, does not show any
effect on the level of the overall procedure. 
Thus, in the 
resampled data the conditional distribution of
$X^2$, given $X^1$ and $X^3$ 
differs from $q^*(x^2|x^1,x^3)$
(even though it does not equal the target conditional
$q^*(x^2|x^1,x^3)$).
We hypothesize that the result is therefore similar to choosing a different target distribution in the first place. 
Indeed, when changing
$p^*(x^1, x^2, x^3)$ to match the data support, 
the set frequencies of the oracle version closely match the resampled version.

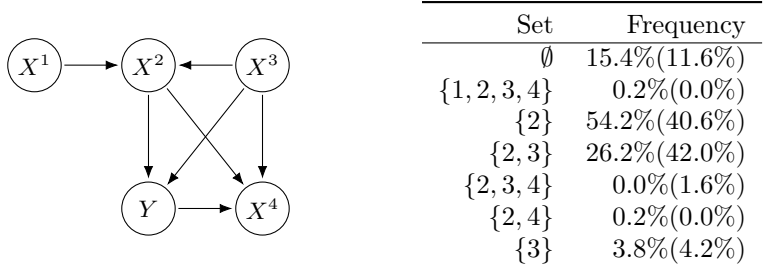
\begin{figure}
\begin{minipage}{0.5\textwidth}
\centering
\begin{tikzpicture}[xscale=1.5, yscale=1.9, shorten >=1pt, shorten <=1pt]
\small
\draw (0,3) node(x1) [observedsmall] {$X^1$};
\draw (1,3) node(x2) [observedsmall] {$X^2$};
\draw (2,3) node(x3) [observedsmall] {$X^3$};
\draw (1,2) node(y) [observedsmall] {$Y$};
\draw (2,2) node(x4) [observedsmall] {$X^4$};
\draw[-Latex] (x1) -- (x2);
\draw[-Latex] (x3) -- (x2);
\draw[-Latex] (x2) -- (x4);
\draw[-Latex] (x3) -- (x4);
\draw[-Latex] (y) -- (x4);
\draw[-Latex] (x2) -- (y);
\draw[-Latex] (x3) -- (y);
\end{tikzpicture}
\end{minipage}%
\begin{minipage}{0.5\textwidth}
\begin{tabular}{rrr}
\toprule
Set & Frequency \\ \hline 
$\emptyset$ &$15.4\% (11.6\%)$\\ 
$\{1,2,3,4\}$ &$0.2\% (0.0\%)$\\ 
$\{2\}$ &$54.2\% (40.6\%)$\\ 
$\{2,3\}$ &$26.2\% (42.0\%)$\\ 
$\{2,3,4\}$ &$0.0\% (1.6\%)$\\ 
$\{2,4\}$ &$0.2\% (0.0\%)$\\ 
$\{3\}$ &$3.8\% (4.2\%)$\\ 
 \bottomrule
\end{tabular}
\end{minipage}
\caption{
(Left) Graphical representation of the SCM used in the simulation of Section~\ref{sec:icpplu}. We assume that the conditional $q^*(x^2|x^1, x^3)$ is known and use this to generate resamples that mimic a heterogeneity in the data that we can then exploit for causal discovery.
(Right) Frequencies of the estimated sets.
The theory guarantees that in at most $5\%$ of the cases, the estimated set is not a subset of $\{2,3\}$, the correct set of causal predictors.
The numbers in parenthesis are oracle benchmarks that sample the data directly from the actual target distribution instead of using resampling.
}
\label{fig:icp_graph_res}
\end{figure}

\section{Conclusion and future work}
\label{sec:summary_future_work}
We formally introduce statistical testing under distributional shifts and illustrate that it can be applied in a diverse set of areas such as reinforcement learning, conditional independence testing and causal inference. We provide a general testing procedure based on weighted resampling and prove pointwise asymptotic level guarantees under mild assumptions. 
Our simulation experiments underline the usefulness of our method:
It is able to test complicated hypotheses, such as dormant independences -- for which to-date no test with provable level guarantees exists -- and 
can be applied to test complex hypotheses in off-policy testing or covariate shift. 
The framework is competitive even in some of the problems, where more specialized solutions exist. 
Its key strength is that it is very easy to apply and can be combined with any existing test making it an attractive go-to method for complicated testing problems.

We believe that several directions would be worthwhile to investigate further. 
In many of the empirical experiments, the requirement that $m = o(\sqrt{n})$ seems too strict and can be relaxed, see also \cref{subsec:test-resampling-worked}.
We hypothesize that under further restrictions on the weights or the test statistics, the assumption for the theoretical results can be relaxed to $m=o(n)$. \citet{bickel2012resampling} consider the `$\binom{n}{m}$' bootstrap, which resamples distinct sequences without weights, and show that under mild assumptions, bootstrap estimates converge if $m=o(n)$. Further work is required to extend this to the case of weighted samples. 

We show in \cref{subsec:uniform-level} that the main convergence result, \cref{thm:asymptotic-level-SIR}, can be extended to uniform level, if we make uniform assumptions on the target test, $\varphi$, and that the weights are uniformly bounded over $\tau^{-1}(H_0)$. In many model classes, the latter assumption may be too strict, and a better understanding of necessary conditions would help. 

\new{
While \cref{thm:asymptotic-level-unknown-weights} provides guarantees when $r(x^A)$ is unknown, the theorem requires guarantees on the relative error $\hat{r}(x^A) / r(x^A)$. 
A more natural guarantee would be on the absolute error $|\hat{r}(x^A) - r(x^A)|$, and we hope further work can shed light on the appropriate conditions (such as model classes $p$ and $q$ or properties of the estimator $\hat{r}$) to achieve such a guarantee.}

Resampling distinct sequences is less prone to weight degeneracy than IPW or resampling with replacement, but in setups with well-behaved weights this may come at a cost of power when resampling only $m \ll n$ points. Resampling non-distinct sequences share many similarities with IPW (for fixed $n$, expectations of $\Psi_{\texttt{REPL}}^{r,m}$ converge to the IPW estimate when $m \to \infty$), but additionally benefits from the ability to test hypotheses where the test statistic cannot be written as an average over the data points (see \cref{sec:ipw}).
Further investigation of the differences between the sampling schemes and benefits and disadvantages in comparison to IPW is needed.

Our methodology considers the setting where we only observe data $\bX_n$ from the distribution $Q^*$. If additionally a sample $\bZ_{n'}$ from the target distribution $P^*$ is already available, one can combine the two data sets, to get a larger approximate sample from $P^*$\new{, a problem known as `domain adaptation' in the literature \citep{finn2017model}.}
In particular, if $\bZ_{n'}$ is also available, one could perform the testing on the combined data set $(\Psi(\bX_n), \bZ_{n'})$.
We believe that similar theoretical guarantees can be proved.

When testing for a hypothesis in the observed domain, we often have the freedom to choose a target distribution which could help us improve the performance of our test (as discussed in \cref{sec:testingobserved}). In the experiments, e.g., \cref{sec:exp-off-policy} and \cref{exp:cond-ind-test}, we choose the target distribution that matches certain marginals, which often helps to minimize the variance of the weights. Another possibility is to choose a target distribution such that the alternative becomes easier to detect which can be achieved by minimizing the $p$-value of the test with respect to the choice of the target distribution.

\begin{ack}
We thank Peter Rasmussen for valuable ideas about the combinatorics, Mathias Drton for helpful
discussions on hidden variable models and Tom Berrett for insightful comments during a discussion of an earlier version of this paper. 
NT, SS, and JP were supported by a research grant (18968) from VILLUM FONDEN and JP was, in addition, supported by the Carlsberg Foundation. NP was supported by a research grant (0069071) from Novo Nordisk Fonden.
\end{ack}

\medskip
{
\small
\bibliography{references.bib}
}

\newpage
\appendix
\section{Further comments on the framework}

\subsection{Forward and backward shifts, \texorpdfstring{$\tau$ and $\eta$}{}}\label{sec:forward-backward-map}

In this paper, as mentioned in \cref{subsec:framework}, we take the starting point that $Q^*$ is observed, and view $P^* = \tau(Q^*)$ as a shifted version of $Q^*$. 
One could instead suppose that we started with a distribution of interest $P^*$, from which no sample is available, and then construct a map $\eta$ such that $Q^* = \eta(P^*)$ is a distribution which can be sampled from in practice. 
If $\tau$ and $\eta$ are invertible, the two views are mathematically equivalent, but if not, there is a subtle difference; the corresponding level guarantees take a supremum either over $Q \in \{Q' \mid \eta^{-1}(Q')\cap H_0 \neq \emptyset \}$ ($\eta$ view) or over $Q \in \{Q'|\tau(Q') \in H_0\}$ ($\tau$ view). To see this, 
we first start with the (natural) level guarantee from the $\eta$ view: 
$\sup_{P \in H_0} \P_{\eta(P)}(\psi_n(\bX_n, U) = 1) \leq \alpha$. We then have
\begin{align*}
&\sup_{P \in H_0} \P_{\eta(P)}(\psi_n(\bX_n, U) = 1) \leq \alpha \\
    \Leftrightarrow \quad 
&\sup_{Q \in \eta(H_0)} \P_Q(\psi_n(\bX_n, U) = 1) \leq \alpha\\
    \Leftrightarrow \quad 
&\sup_{Q \in \{Q' \,|\, \eta^{-1}(Q') \cap H_0 \neq \emptyset \}} \P_Q(\psi_n(\bX_n, U) = 1) \leq \alpha.
\end{align*}
If, alternatively, we start with the level guarantee from the $\tau$ view, we find 
\begin{align*}
&\sup_{P \in H_0}
\sup_{Q \in \tau^{-1}(P)}
\P_{Q}(\psi^r_n(\bX_n, U) = 1) 
\leq \alpha 
 \\ 
    \Leftrightarrow \quad&\sup_{Q \in \tau^{-1}(H_0)} \P_Q(\psi_n(\bX_n, U) = 1) \leq \alpha\\
     \Leftrightarrow \quad
&\sup_{Q \in \{Q' \,|\, \tau(Q') \in H_0 \}} \P_Q(\psi_n(\bX_n, U) = 1) \leq \alpha. \end{align*}
Comparing the last two lines yields the claim. 

\subsection{Example: Interventions in causal models} \label{sec:scm}
One example of a distributional shift $\tau$ is the case where $\tau$ represents an intervention in a structural causal model (SCM) over $X^1, \ldots, X^d$ \citep{pearl2009causality}.
 An SCM $\scm$ over $X^1, \ldots, X^d$ is a collection of structural assignments $f^1, \ldots, f^d$ and noise distributions $Q_{N^1}, \ldots, Q_{N^1}$ such that for each $j = 1, \ldots, d$, we have $X^j \coloneqq f^j(\PA^{j}, N^j)$.
Here, the noise variables $N^j$ are distributed according to $N^j \sim Q_{N^j}$ and are assumed to be jointly independent.
The sets $\PA^{j}\subseteq\{X^1, \ldots, X^d\}\backslash\{X^j\}$\footnote{For notational convenience we sometimes refer to the parent sets by their indices, i.e., $\PA^{j}\subseteq\{1,\ldots,d\}\setminus\{j\}$.}  denote the causal parents of $X^j$.
The induced graph over $X^1, \ldots, X^d$ is the graph obtained by drawing directed edges from each variable on the right-hand side of each assignment to the variables on the left-hand side; 
see \citet{bongers2021foundations} for a more formal introduction to SCMs.

Let us assume that $\scm$ induces a unique observational distribution $Q$ over $X^1, \ldots, X^d$ (which is the case if the graph is acyclic, for example), and assume that $Q$ 
admits a joint density $q$ with respect to a product measure. 
Then $q$ satisfies the factorization property 
(see \citet{lauritzen1990independence} or Theorem 1.4.1 in \citet{pearl2009causality}): $q(x^1, \ldots, x^d) = \prod_{j=1}^d q_{X^j\mid \PA^j}(x^j | x^{\PA^j})$. 
In an SCM, an intervention on a variable $X^k$ replaces the tuple $(f^k, \PA^k, Q_{N^k})$ with $(\bar{f}^k, \widebar{\PA}^k, \bar{Q}_{N^j})$ in the structural assignment for $X^k$, and we denote the replacement by $\DO(X^k \coloneqq \bar{f}^k(\widebar{\PA}^k, \bar{N}^k))$ \citep{pearl2009causality}.
This new mechanism determines a 
conditional that we denote by 
$p^*(x^k| x^{\widebar{\PA}^k})$.
The interventional distribution is the induced distribution with the new structural assignment, and we denote this by 
$P \coloneqq Q^{\DO(X^k := \bar{f}(\widebar{\PA^k},\bar{N}^k))}$. 
If $P$ admits the density $p$, only the conditional density of $X^k$ changes \citep[e.g.,][]{Haavelmo1944,Aldrich1989,pearl2009causality, peters2017elements}, 
that is, for $j \neq k$, we have $p(x^j | x^{\PA^j}) = q(x^j| x^{\PA^j})$, for all $x^j$ and $x^{\PA^j}$. 
Assume that for the true but unknown distribution $Q^*$ we know the conditional 
$q^*(x^k | x^{\PA^k})$ (e.g., because this was part of the design when generating the data).
Due to the factorization property, the intervention $\DO(X^k \coloneqq \bar{f}^k(\widebar{\PA}^k, \bar{N}^k))$ can then be represented as a map $\tau$ that acts on the density $q$:
\begin{align*}
    \tau(q)(x^1, \ldots, x^d) \coloneqq \frac{p^*(x^k | x^{\widebar{\PA}^k})}{q^*(x^k | x^{\PA^k})}\cdot q(x^1, \ldots, x^d).
\end{align*}
Defining $r(x^{\{k\}\cup \PA^k\cup\widebar{\PA}^k}) := {p^*(x^k | x^{\widebar{\PA}^k})}/{q^*(x^k | x^{\PA^k})}$, this takes the form of \cref{eq:tauform}. As the conditional $p^*(x^k | x^{\widebar{\PA}^k})$ is fully specified by the intervention, we therefore know the function $r$.
Our proposed framework allows us to test statements about the distribution $Q^{\DO(X^k := \bar{f}(\widebar{\PA^k},\bar{N}^k))}$.
We obtain similar expressions when intervening on several variables at the same time.

Similar distributional shifts can be obtained, of course, if the factorization is non-causal (see also Section~\ref{sec:conditionaltesting}), so while our framework contains
intervention distributions as a special case, 
it equally well applies to 
non-causal models.

\new{
\section{Efficient computation of \texorpdfstring{$V(n,m)$ in \cref{thm:finite-level-SIR}}{}}
\label{app:evaluation-of-variance}
In this section, 
we show that for $n,m \in \N$ and $K \geq 1$
\begin{equation*}
    V(n,m) = \binom{n}{m}^{-1} \sum_{\ell=1}^m \binom{m}{l}\binom{n-m}{m-\ell}(K^\ell - 1)
\end{equation*}
can be evaluated efficiently. If $m  n/2$, such that for some $\ell$ one has $m-\ell \geq n-m$, we use the convention that if $a > b$ then $\binom{b}{a} = 0$.
If one evaluated the term $\binom{n}{m}^{-1}$ separately, this could potentially cause numerical underflow, and similarly terms in the sum could get very 
large, such as the summand including $K^m - 1$.

Denote the summands by $s_\ell,$ that is
\begin{equation*}
    s_\ell = \frac{\binom{m}{l}\binom{n-m}{m-\ell}(K^\ell - 1)}{\binom{n}{m}}.
\end{equation*}
We can compute $s_1$ by:
\begin{align*}
    s_1 &= \frac{\binom{m}{1}\binom{n-m}{m-1}}{\binom{n}{m}}(K-1) \\
    &= m^2 (K-1) \frac{(n-m)!(n-m)!}{n!(n-2m+1)!} \\
    &=(K-1) \frac{m^2}{n-m + 1} \prod_{j=0}^{m-2} \frac{n-m-j}{n-j}.
\end{align*}
This can be evaluated in $O(m)$ time.
Further, if $s_\ell\neq 0$, the ratio of two consecutive summands is
\begin{align*}
    \tfrac{s_{\ell+1}}{s_\ell} &= \frac{\binom{m}{\ell+1}\binom{n-m}{m-\ell-1}(K^{\ell+1} - 1)}{\binom{m}{\ell}\binom{n-m}{m-\ell}(K^{\ell} - 1)} \\
    &= \frac{(m-\ell)^2}{(\ell+1) (n-2m+\ell+1)} \frac{K^{\ell+1} -1}{K^\ell -1},
\end{align*}
which for a given $\ell$, can be evaluated in $O(1)$ time.
Hence, we can compute $\sum_{\ell=1}^m s_\ell$, by first computing $s_1$, and for each $\ell$, compute $s_{\ell+1} = \tfrac{s_{\ell+1}}{s_\ell}s_\ell$, as long as $s_\ell \neq 0$ (after which the remaining terms are $0$). The overall computational cost of computing $V(n,m) = \sum_{\ell=1}^m s_\ell$ is thus $O(m)$.

}

\section{Algorithm for hypothesis testing with unknown distributional shift}\label{sec:sir-details}
This section contains \cref{alg:resampling-and-testing-unknown-shift}, which describes our method for testing under distributional shifts for the case where the shift factor $r_q$ is unknown,
but can be estimated by an estimator $\hat{r}$. \Cref{alg:resampling-and-testing-unknown-shift} is similar to \cref{alg:resampling-and-testing} but one additionally splits the sample $\bX_n$ into two disjoint samples $\bX_{n_1}$ and $\bX_{n_2}$ and uses $\bX_{n_1}$ for estimating the weights $\hat{r}_{n_1}$, which are then, together with $\bX_{n_2}$, used as an input to \cref{alg:resampling-and-testing}.

We view the sample splitting as a theoretical device. In practice, we are using the full sample both for estimating the weights and for applying the test.
\begin{algorithm}[H]
\caption{Testing a target hypothesis with unknown distributional shift and resampling}
\begin{algorithmic}[1]
\Statex \textbf{Input:} Data $\bX_{n}$, target sample size $m$, hypothesis test $\varphi_m$, estimator $\hat{r}$ for $r_q$, and $a$.
\State Let $n_1, n_2$ be s.t.\ $n_1 + n_2 = n$ and $n_1^a = \sqrt{n_2}$
\State $\bX_{n_1} \gets X_{1}, \ldots, X_{n_1}$
\State $\bX_{n_2} \gets X_{n_1 + 1}, \ldots, X_{n_1 + n_2}$
\State $\hat{r}_{n_1} \gets$ estimate of $r_q$ based on on $\bX_{n_1}$
\State $(i_1, \ldots, i_m) \gets$ sample from $\{1, \ldots, n_2\}^m$ with weights~\eqref{eq:SIR-weights} based on $\hat{r}_{n_1}$.
\State 
$\Psi_{\texttt{DRPL}}^{\hat{r}_{n_1},m}(\bX_{n_2}, U) \gets (X_{n_1 + i_1}, \ldots, X_{n_1 + i_m})$
\Statex \Return $\psi_n^{\hat{r}}(\bX_{n}, U) \coloneqq \varphi_m(\Psi_{\texttt{DRPL}}^{\hat{r}_{n_1},m}(\bX_{n_2}, U))$
\end{algorithmic}
\label{alg:resampling-and-testing-unknown-shift}
\end{algorithm}

\section{Sampling from \texorpdfstring{$\Psi_{\texttt{DRPL}}$}{}}\label{sec:sampling-DRPL}
This section provides details on sampling from 
$\Psi_{\texttt{DRPL}}^{r,m}$, as defined by \cref{eq:SIR-weights}.
We have defined $\Psi_{\texttt{REPL}}^{r,m}$ and $\Psi_{\texttt{NO-REPL}}^{r,m}$ as weighted resampling with and without replacement, respectively.
$\Psi_{\texttt{NO-REPL}}^{r,m}$ can be implemented as a sequential procedure that first draws $i_1$ with weights $r(X_{i})/\sum_{j=1}^n r(X_j)$, and then draws $i_2$ with weights $r(X_i)/\sum_{j=1,j\neq i_1}^n r(X_j)$, and so forth. Although both 
$\Psi_{\texttt{NO-REPL}}^{r,m}$
and 
$\Psi_{\texttt{DRPL}}^{r,m}$
sample distinct sequences $(i_1, \ldots, i_m)$, they are, in general, not equivalent, as can be seen from the form of the weights $w_{(i_1, \ldots, i_m)}^{\texttt{DRPL}}$ and $w_{(i_1, \ldots, i_m)}^{\texttt{NO-REPL}}$ below.
When there is no ambiguity, we omit superscripts and write $\Psi_{\texttt{DRPL}}$, for example. We also interchangeably consider a sample from $\Psi_{\texttt{DRPL}}$ to be a sequence $(i_1, \ldots, i_m)$ and a subsample $(X_{i_1}, \ldots, X_{i_m})$ of $\bX_n$.

The procedures $\Psi_{\texttt{DRPL}}, \Psi_{\texttt{REPL}}$ and $\Psi_{\texttt{NO-REPL}}$ sample a sequence $(i_1, \ldots, i_m)$ with weights $w_{(i_1,\ldots, i_m)}$ that are, respectively, given by:
\begin{align*}
    w_{(i_1,\ldots, i_m)}^{\texttt{DRPL}} &= \frac{\prod_{\ell=1}^m r(X_{i_\ell})}{\sum\limits_{\substack{(j_1, \ldots, j_m) \\ \text{distinct}}}\prod_{\ell=1}^m r(X_{j_\ell})} \quad \text{for distinct}\,(i_1, \ldots, i_m)& \\
    w_{(i_1,\ldots, i_m)}^\texttt{REPL} &= \frac{\prod_{\ell=1}^m r(X_{i_\ell})}{\sum\limits_{(j_1, \ldots, j_m)}\prod_{\ell=1}^m r(X_{j_\ell})} \quad \text{for all}\,(i_1, \ldots, i_m) \\
    w_{(i_1,\ldots, i_m)}^\texttt{NO-REPL} &= \frac{\prod_{\ell=1}^m r(X_{i_\ell})}
    {\sum\limits_{j_1 = 1}^n r(X_{j_1}) \sum\limits_{\substack{j_2 = 1 \\j_2 \neq i_1}}^nr(X_{j_2}) \cdots \sum\limits_{\substack{j_m=1\\j_{m} \notin \{i_1, \ldots, i_{m-1}\}}}^n r(X_{j_m})}
    \quad \text{for distinct}\,(i_1, \ldots, i_m) 
\end{align*}
Here, the comment `$\text{for distinct}\,(i_1, \ldots, i_m)$' implies that the weights are zero otherwise.
Most statistical software have standard implementations for sampling from $\Psi_{\texttt{REPL}}$ and $\Psi_{\texttt{NO-REPL}}$ (known simply as sampling with or without replacement). We now detail a number of ways to sample a sequence $(i_1, \ldots, i_m)$ from $\Psi_{\texttt{DRPL}}$. 
The first two sampling methods are exact, the third sampling method is approximate.

\subsection{Acceptance-rejection sampling with \texorpdfstring{$\Psi_{\texttt{REPL}}$}{} as proposal}\label{subsec:acceptance-rejection-repl}
Given a sample $\bX_n$, one can sample from $\Psi_{\texttt{DRPL}}$ by acceptance-rejection sampling from $\Psi_{\texttt{REPL}}$, by drawing sequences $(i_1, \ldots, i_m)$ from $\Psi_{\texttt{REPL}}$ until one gets a draw that is distinct, which is then used as the draw from $\Psi_{\texttt{DRPL}}$. This is a valid sampling method for $\Psi_{\texttt{DRPL}}$, because for any distinct sequence $(i_1, \ldots, i_m)$ 
we have
\begin{align*}
    \P_Q(\Psi_\texttt{REPL} = (i_1, \ldots, i_m) \mid \Psi_\texttt{REPL} \,\text{distinct}, \bX_n) 
    &= \frac{\P_Q(\Psi_\texttt{REPL} = (i_1, \ldots, i_m), \Psi_\texttt{REPL} \,\text{distinct}\mid \bX_n)}{\P_Q(\Psi_\texttt{REPL} \,\text{distinct} \mid \bX_n)} \\
    &= \frac{\P_Q(\Psi_\texttt{REPL} = (i_1, \ldots, i_m) \mid \bX_n)}{\P_Q(\Psi_\texttt{REPL} \,\text{distinct} \mid \bX_n)} \\
    &= \frac{w^{\texttt{REPL}}_{(i_1, \ldots, i_m)}}{\sum\limits_{\substack{(j_1, \ldots, j_m) \\ \text{distinct}}} w^{\texttt{REPL}}_{(j_1, \ldots, j_m)}} \\
    &= \frac{\prod_{\ell=1}^m r(X_{i_\ell})}{\sum\limits_{(j_1, \ldots, j_m)}\prod_{\ell=1}^m r(X_{j_\ell})}
    \frac{\sum\limits_{(j_1, \ldots, j_m)}\prod_{\ell=1}^m r(X_{j_\ell})}{\sum\limits_{\substack{(j_1, \ldots, j_m) \\ \text{distinct}}}\prod_{\ell=1}^m r(X_{j_\ell})} \\
    &= w_{(i_1,\ldots, i_m)}^{\texttt{DRPL}}\\
    & = \P_Q(\Psi_\texttt{DRPL} = (i_1, \ldots, i_m)\mid \bX_n).
\end{align*}
By integrating over $\bX_n$, this implies that $\P_Q(\Psi_\texttt{REPL} = (i_1, \ldots, i_m) \mid \Psi_\texttt{REPL} \,\text{distinct}) =  \P_Q(\Psi_\texttt{DRPL} = (i_1, \ldots, i_m))$.
\Cref{prop:repl-becomes-dist} shows that under certain assumptions, the probability of sampling a distinct sample from $\Psi_{\texttt{REPL}}$ converges to $1$.

\subsection{Acceptance-rejection sampling with \texorpdfstring{$\Psi_{\texttt{NO-REPL}}$}{} as proposal}\label{sec:no-repl-sampling}
If $m$ is large compared to $n$, it may be that most of the samples drawn from $\Psi_{\texttt{REPL}}$ are not distinct, and so the acceptance rejection scheme in \cref{subsec:acceptance-rejection-repl} may take too many attempts to produce a distinct sample. 
As an alternative, one can use $\Psi_{\texttt{NO-REPL}}$ as a proposal distribution for an acceptance-rejection sampler, which is typically faster, since $\Psi_{\texttt{NO-REPL}}$ has the same support as $\Psi_{\texttt{DRPL}}$.
Given a sample $\bX_n$, we thus need to identify an $M$ such that $$\forall\,\, \text{distinct } (i_1, \ldots, i_m): \frac{\P_Q(\Psi_{\texttt{DRPL}} = (i_1, \ldots, i_m)\mid \bX_n)}{\P_Q(\Psi_{\texttt{NO-REPL}} = (i_1, \ldots, i_m)\mid \bX_n)} \leq M.$$ 
We have 
$$
    \frac{\P_Q(\Psi_{\texttt{DRPL}} = (i_1, \ldots, i_m)\mid \bX_n)}{\P_Q(\Psi_{\texttt{NO-REPL}} = (i_1, \ldots, i_m) \mid \bX_n)} = \frac{w^{\texttt{DRPL}}_{(i_1, \ldots, i_m)}}{w^{\texttt{NO-REPL}}_{(i_1, \ldots, i_m)}} = \frac{\sum\limits_{j_1} r(X_{j_1}) \sum\limits_{j_2 \neq i_1} r(X_{j_2}) \cdots \sum\limits_{j_m \neq i_1, \ldots, i_{m-1}} r(X_{j_m})}{\sum\limits_{\substack{(j_1, \ldots, j_m) \\ \text{distinct}}}\prod_{\ell=1}^m r(X_{j_\ell})}.
$$
The denominator does not depend on $i_1, \ldots, i_m$, and the numerator can be upper bounded by
\begin{align*}
    (1-0)(1-p_{(1)})(1 - p_{(1)} - p_{(2)}) \cdots (1 - p_{(1)} - \ldots - p_{(m-1)}),
\end{align*}
where $p_{(1)} = \min \{r(X_1), \ldots, r(X_n)\}$ is the smallest of the weights, $p_{(2)}$ is the second smallest, etc. Thus, we can choose
\begin{align*}
    M \coloneqq \frac{(1-0)(1-p_{(1)})(1 - p_{(1)} - p_{(2)}) \cdots (1 - p_{(1)} - \ldots - p_{(m-1)})}{\sum\limits_{\substack{(j_1, \ldots, j_m) \\ \text{distinct}}}\prod_{\ell=1}^m r(X_{j_\ell})}.
\end{align*}
We now proceed with an ordinary acceptance-rejection sampling scheme: We sample a (distinct) sequence $(i_1, \ldots, i_m)$ from $\Psi_{\texttt{NO-REPL}}$ and an independent, uniform variable $V$ on the interval $(0,1)$. We accept $(i_1, \ldots, i_m)$ if 
\begin{align*}
    V &\leq \frac{\P_Q(\Psi_{\texttt{DRPL}}=(i_1, \ldots, i_m) \mid \bX_n)}{M\cdot\P_Q(\Psi_{\texttt{NO-REPL}}=(i_1, \ldots, i_m) \mid \bX_n)} \\
    &= \frac{(1-0)(1-p_{i_1})(1 - p_{i_1} - p_{i_2}) \cdots (1 - p_{i_1} - \ldots - p_{i_{m-1}})}{(1-0)(1-p_{(1)})(1 - p_{(1)} - p_{(2)}) \cdots (1 - p_{(1)} - \ldots - p_{(m-1)})}.
\end{align*}
Here, we have used that the denominator of $M$ cancels with the normalization constant of $\P_Q(\Psi_{\texttt{DRPL}}=(i_1, \ldots, i_m) \mid \bX_n)$.
If the sample is not accepted, we
draw another
sample from $\Psi_{\texttt{NO-REPL}}$
until one sample is accepted.

\subsection{Approximate Gibbs sampling starting from \texorpdfstring{$\Psi_{\texttt{NO-REPL}}$}{}}\label{sec:no-repl-gibbs}
There are cases, where the sampling schemes presented in \cref{subsec:acceptance-rejection-repl,sec:no-repl-sampling} 
do not yield an accepted sample in a reasonable amount of time
(this is typically due to $m$ being too large compared to $n$, see \cref{assump:m-rate-n}). 
In such cases, one can get an approximate sample of $\Psi_{\texttt{DRPL}}$ by sampling $\Psi_{\texttt{NO-REPL}}$ and shifting it towards $\Psi_{\texttt{DRPL}}$ using a Gibbs sampler \citep{geman1984stochastic}. 

Let therefore $(i_1, \ldots, i_m)$ be an initial (distinct) sample from $\Psi_{\texttt{NO-REPL}}$, and define $i_{-\ell}$ to be the sequence without the $\ell$'th entry. The Gibbs sampler sequentially samples $i_\ell$ from the conditional distribution $j \mid i_{-\ell}$ in $\Psi_{\texttt{DRPL}}$. 
To compute this conditional probability let $\Psi_{\texttt{DRPL}}^\ell$ be the $\ell$'th index of a sample. Then
\begin{align*}
    \P_Q(\Psi_{\texttt{DRPL}}^\ell = j | \Psi_{\texttt{DRPL}}^{-\ell} &= i_{-\ell}) \\
    &= \frac{\P_Q(\Psi_{\texttt{DRPL}} = (i_1, \ldots, j, \ldots,  i_m))}{\P_Q(\Psi_{\texttt{DRPL}}^{-\ell} = i_{-\ell})} \\
    &= \frac{r(X_{i_1})\cdots r(X_j) \cdots r(X_{i_m})}{\sum_{v \notin i_{-\ell}} r(X_{i_1})\cdots r(X_v)\cdots r(X_{i_m})} = \frac{r(X_{j})}{\sum_{v\notin i_{-\ell}} r(X_v)},
\end{align*}
i.e., the conditional distribution of one index $i_\ell$ given $i_{-\ell}$ is just a weighted draw among $\{i_1, \ldots, i_m\}\backslash i_{-\ell}$. This is simple to sample from and the Gibbs sampler now iterates through the indices $(i_1, \ldots, i_m)$, at each 
iteration replacing the index $i_\ell$ by a sample from the conditional given $i_{-\ell}$. Iterating this a large number of times produces an approximate sample from $\Psi_{\texttt{DRPL}}$.

\subsection{Sampling with replacement}\label{subsec:testing-with-REPL}
Instead of the sampling scheme $\Psi_{\texttt{DRPL}}$ presented above, we can also use weighted sampling with replacement, which we denote $\Psi_{\texttt{REPL}}$ (see \cref{sec:sampling-DRPL} for details). 
Sampling from $\Psi_{\texttt{REPL}}$ is simpler than from $\Psi_{\texttt{DRPL}}$, and while sampling from $\Psi_{\texttt{REPL}}$ is in some cases disadvantageous for testing (e.g., if we test whether the target distribution has a point mass), if the test is not prone to duplicate data points, testing based on $\Psi_{\texttt{REPL}}$ may be advantageous over $\Psi_{\texttt{DRPL}}$ (further examination is needed to clarify this relationship). When sampling without weights, \citet{bickel2012resampling} present regularity conditions on the test statistic that guarantee consistency even with $m=o(n)$.

Here we show that under additional assumptions, the probability of a non-distinct sample from $\Psi_{\texttt{REPL}}$ converges to $0$. Consider the following strengthening of \cref{assump:finite-second-moment}.
\begin{assumpenum}[start=3]
    \item \label{assump:bounded-weights} 
    There exists $L \in \R$ such that for all $v \geq 1$, $\E_Q[r(X_i)^{v+1}] \leq L^v$. 
\end{assumpenum}
This is for instance trivially satisfied if $r(X_i)$ is $Q$-a.s. bounded by a constant $L$.
The following proposition shows that under \cref{assump:m-rate-n,assump:bounded-weights} the probability of drawing a distinct sample from $\Psi_{\texttt{REPL}}$ converges to $1$.
\begin{proposition}[Asymptotic equivalence of REPL and DREPL for bounded weights]\label{prop:repl-becomes-dist}
    Let $\tau: \cQ \rightarrow \cP$ be a distributional shift for which a known map $r:\mathcal{X}\rightarrow[0,\infty)$ exists, satisfying $\tau(q)(x)\propto r(x)q(x)$, see \cref{eq:tauform}.
    Consider an arbitrary $Q \in \cQ$ and $P = \tau(Q)$.
    Let $m=m(n)$ be a resampling size and let $\Psi_{\texttt{REPL}}$ 
    be the weighted resampling with replacement defined in \cref{subsec:method}. Then, if $m$ and $Q$ satisfy \cref{assump:m-rate-n,assump:bounded-weights}, it holds that
    \begin{align*}
        \lim_{n\rightarrow\infty}\P_Q(\Psi^{r,m}_{\texttt{REPL}}(\bX_n, U) \,\text{distinct})=1.
    \end{align*}
\end{proposition}
As a corollary to \cref{thm:asymptotic-level-SIR}, we also have pointwise asymptotic level of a test when $\Psi_{\texttt{REPL}}$ is used instead of $\Psi_{\texttt{DRPL}}$.
\begin{corollary}[Pointwise asymptotics -- REPL]\label{cor:asymptotic-level-SIR-REPL}
    Assume the same setup and assumptions as in \cref{thm:asymptotic-level-SIR} and additionally assume \cref{assump:bounded-weights}.
    Let $\Psi_{\texttt{REPL}}$ 
    be the weighted resampling with replacement defined in \cref{subsec:method} and let $\psi^r_n$ be the REPL-based resampling test defined by $\psi^r_n(\bX_n, U) \coloneqq \varphi_m(\Psi_{\texttt{REPL}}^{r,m}(\bX_n, U))$.
    Then, it holds that
    \begin{equation*}
        \limsup_{n\rightarrow\infty} \P_{Q}(\psi^r_{n}(\mathbf{X}_n, U)=1)=\alpha_\varphi.
    \end{equation*}
    The same statement holds when replacing both $\limsup$'s (including the one in $\alpha_\varphi$) with $\liminf$'s.
\end{corollary}

\section{Target heuristic for choosing \texorpdfstring{$m$}{}} \label{sec:targetheuris}
The target heuristic is a data driven procedure to choose $m$ in finite sample settings, see Section~\ref{subsec:test-resampling-worked}. It is summarized in  
Algorithm~\ref{alg:tuning-m}.
\begin{algorithm}[ht]
\caption{Target heuristic: Choosing $m$ by testing resampling validity}
\begin{algorithmic}[1]
\Statex \textbf{Input:} Data $\bX_{n}$, shift factor $r(x^A)$, threshold $\alpha_c$, initial target size $m_0$, increment size $\Delta$, repetitions $K$, conditional goodness-of-fit test $\kappa$.
\State $\texttt{qt} \gets $ $\alpha_c$-quantile of $\texttt{mean}(U_1, \ldots, U_K)$, where $U_i \sim Unif(0, 1)$
\State $m \gets m_0$
\State $\texttt{m\_valid} \gets \texttt{true}$ 
\While {\texttt{m\_valid}}
\For{$k=1, \ldots K$}
    \State $\texttt{res}_k \gets \kappa(\Psi^{r,m}(\bX_n))$
\EndFor
\If {$\texttt{mean}(\texttt{res}_1, \ldots, \texttt{res}_K) > \texttt{qt}$}
    \State $m \gets m + \Delta$
\Else
    \State $\texttt{m\_valid} \gets \texttt{false}$
    \State $m \gets m - \Delta$
\EndIf
\EndWhile
\State \Return m
\end{algorithmic}
\label{alg:tuning-m}
\end{algorithm}

\section{Additional experiments}
\subsection{Assumption~\texorpdfstring{\cref{assump:m-rate-n}}{} when sample size increases}
As indicated by the results in \cref{sec:exp-assumptions-a2-a3}, \cref{assump:m-rate-n} may sometimes be too strict of an assumption, in the sense that level can be attained also when \cref{assump:m-rate-n} is violated.
We explored this for a continuous data example in \cref{sec:exp-assumptions-a2-a3}, but here, we investigate
the effect of violating \cref{assump:m-rate-n} as the sample size $n$ increases for a binary setting.

We simulate data $\bX_n$ from a binary distribution $\P_Q(X = 0) = 0.9$, $\P_Q(X = 1) = 0.1$, and consider the target distribution $P$ where the probabilities are flipped
, i.e. $\P_P(X = 0) = 0.1$, $\P_P(X = 1) = 0.9$. 
We consider the hypothesis $H_0: \P_P(X = 1) > 0.8$, which is clearly satisfied for $P$, but may not be detected by the resample if $m = n^a$ for $a>0.5$ chosen too large. We repeat the experiment $200$ times and compute the resulting rejection rates for various sample sizes and rates $m=n^a$.
The results are shown in \cref{fig:rate}.
\begin{figure}[t]
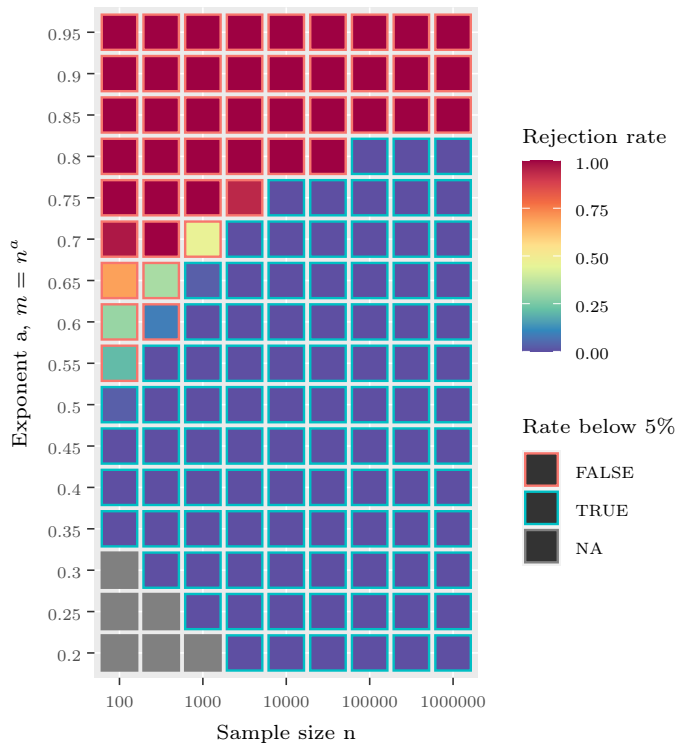

    \centering
    \scriptsize

     \caption{Rejection rates when $m = n^a$ points are resampled for various exponents $a$ and sample sizes $n$. Each tile represents a combination of sample size $n$ and exponent $a$, and the color indicates the rate at which the hypothesis $\P_P(X=1)>0.8$ is rejected in the resampled data, when resampling $m=n^a$ points; ideally this should be low (blue) since we consider a target distribution where $\P_P(X=1) = 0.9$. Even though this is not provided by our theoretical results, it seems that level is possible even for rates larger than $0.5$.}
    \label{fig:rate}
\end{figure}
Under \cref{assump:m-rate-n}, \cref{thm:asymptotic-level-SIR} guarantees asymptotic level. 
Indeed, at all sample sizes $n$, the test has the correct level, 
when \cref{assump:m-rate-n} is satisfied (that is, when $a < 0.5)$. Yet, though this is not guaranteed by theory, we have indications of asymptotic level also for $a > 0.5$.

\subsection{Model selection under covariate shift}
\label{sec:model_selection_experiment}
In this section, we apply our testing method to the problem of model selection under covariate shift as discussed in Section~\ref{sec:modselundercovsh}. We generate a data set $D \coloneqq \{(X_i, Y_i)\}_{i=1}^n$ of size $n = 3'000$. Each $(X_i, Y_i)$ is drawn i.i.d.\ according to the following data generating process:
\begin{align*}
    (X^1, X^2) \sim \texttt{GaussianMixture}(
    \begin{bmatrix} 3 & 3 \end{bmatrix}, 
    \begin{bmatrix} -3 & -3 \end{bmatrix}) \qquad
    Y \coloneqq \begin{cases} 
     \sin(X_2 + \epsilon_{Y}),
     &\mbox{if } X_1 \geq 0 \\
     \sin(3 + X_1 + \epsilon_{Y}),
     & \mbox{if } X_1 < 0,
    \end{cases}
\end{align*}
where $\texttt{GaussianMixture}(
    \begin{bmatrix} 3 & 3 \end{bmatrix}, 
    \begin{bmatrix} -3 & -3 \end{bmatrix})$ is an even mixture (i.e., $p = 0.5$) of two 2-dimensional Gaussian distributions with means $\mu_1= (3,3)^\top$, $\mu_2=(-3,-3)^\top$ 
    and unit covariance matrix, and $\epsilon_Y$ is a standard Gaussian $\mathcal{N}(0, 1)$-variable. \cref{fig:model_selection_data} illustrates a sample of size $1'000$ from this data generating process.
    
We randomly split the data $D$ into a training set $D_{train}$ of size $2'000$ and a test set $D_{test}$ of size $1'000$. Using $D_{train}$, we train two candidate classifiers, namely logistic regression (\texttt{LR}) and random forest (\texttt{RF}) to predict $Y$ from $X$. Both models are trained using the \texttt{Scikit-Learn} Python package \citep{scikit-learn} with default parameters. We consider the area under the curve (AUC) as the scoring function, where we denote the AUC scores for the models \texttt{LR} and \texttt{RF} by $\text{AUC}(\texttt{LR})$ and $\text{AUC}(\texttt{RF})$, respectively. Then, we apply our resampling approach on $D_{test}$ to test whether \texttt{LR} outperforms \texttt{RF} when the distribution of
$(X^1, X^2)$ is changed to a single 2-dimensional Gaussian distribution with mean $\mu = (3,3)^\top$ and unit covariance matrix. In this experiment, we choose the resampling size $m$ by the target heuristic (see Algorithm~\ref{alg:tuning-m}) and assume that the shift factor $r(x) \coloneqq p^*(x)/q^*(x)$ is known, where $q^*(x)$ is a pdf of the mixture of Gaussian $\texttt{GaussianMixture}(
\begin{bmatrix} 3 & 3 \end{bmatrix}, 
\begin{bmatrix} -3 & -3 \end{bmatrix})$ and $p^*(x)$ is a pdf of the Gaussian $\mathcal{N}((3, 3)^\top, \boldsymbol{I}_2)$. We employ DeLong's test \citep{delong1988comparing} to test the hypothesis $\text{AUC}(\texttt{LR}) \leq \text{AUC}(\texttt{RF})$
using the R-package \texttt{pROC} \citep{pROC}.

We repeat the experiment $500$ times and report in \cref{fig:model_selection_exp} how many times our resampling test (resampling) rejects and returns that AUC(\texttt{LR}) $>$ AUC(\texttt{RF}). As a benchmark, we also report the corresponding rejection rate when we perform the test directly on a sample from the target distribution (oracle) in which the covariate distribution is changed to $p^*(x)$. We also report the rejection rate when we perform the test on the observed test set directly $D_{test}$ (observed), without resampling it first.

As shown in \cref{fig:model_selection_exp}, under the observed distribution we have $\text{AUC}(\texttt{LR}) \leq \text{AUC}(\texttt{RF})$ (the rejection rate is $0$ in the observed sample). However, under the target distribution  $\text{AUC}(\texttt{LR})$ is higher than $\text{AUC}(\texttt{RF})$ (the rejection rate is $1$ in the target sample). 
Our resampling approach yields 
high power against the alternative hypothesis, even without having access to the oracle information but using resampling instead.

\begin{figure}[t]
    \centering
    \includegraphics[width=0.6\textwidth]{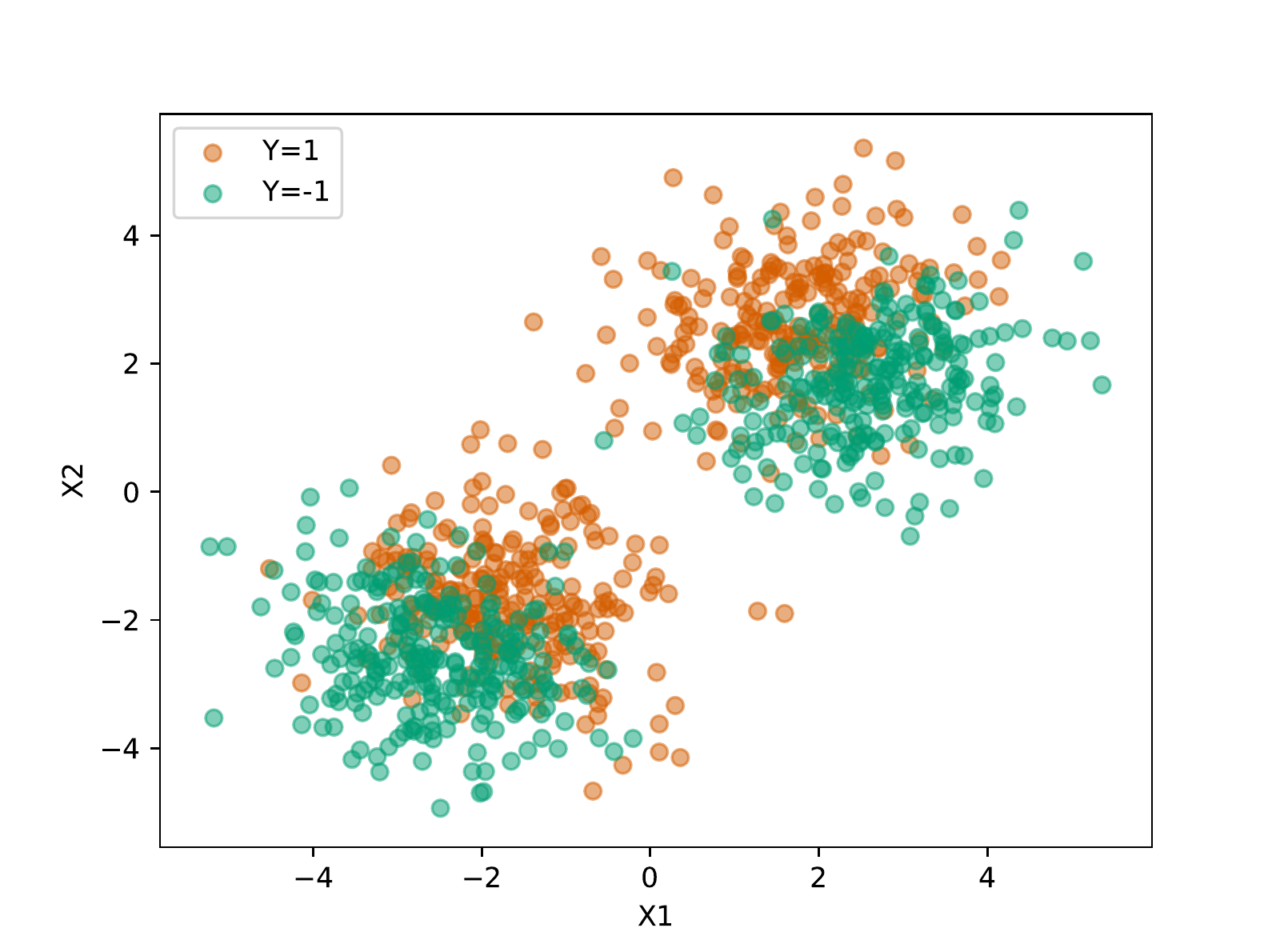}
    \caption{Sample from the data generating process for the experiment described in \cref{sec:model_selection_experiment}.}
    \label{fig:model_selection_data}
\end{figure}

\begin{table}[t]
    \centering
	\begin{tabular}{ccc}
    \toprule
resampling & oracle & observed \\ \hline 
$0.874 \pm 0.029$ & 1.0 & 0.0 \\      \bottomrule
    \end{tabular}
    \caption{Rejection rates of the test with alternative AUC(\texttt{LR}) $>$ AUC(\texttt{RF}) over 500 repetitions. The resampling test shows high power to detect the alternative.}
    \label{fig:model_selection_exp}
\end{table}

\subsection{Additional simulation details}\label{sec:sim-details}
In this section, we state the models used for generating the data in \cref{sec:expverma,sec:icpplu}.

\subsubsection*{\texorpdfstring{\cref{sec:expverma}}{}}
The binary data in \cref{sec:expverma} was generated by first sampling hyper-parameters:
\begin{align*}
    &p_H \sim \textrm{Dirichlet}^{1\times 4}(3, 3, 3, 3) \quad\quad p_1 \sim B^{1\times 1}(1, 1) \quad\quad p_2 \sim B^{2 \times 4}(1, 1) p_3 \sim B^{1\times 2}(1, 1) \\ 
    &\mathcal{G}: p_4 \sim B^{2, 1}(1,1) \quad\quad \mathcal{H}: p_4 \sim B^{2, 2}(1,1),
\end{align*}
where $B(a,b)$ is a Beta distribution and the superscript indicates the dimension of the sampled parameter matrix. 
The matrices $p_1, \ldots, p_4$ correspond to conditional probability tables given parent variables in the graphs $\mathcal{G}$ or $\mathcal{H}$. The distributions are the same when sampling from $\mathcal{G}$ and $\mathcal{H}$, except for $X^4$, which has an additional parent in $\mathcal{H}$. 
In each repetition, given hyper-parameters, a data set is sampled from the structural equation model
\begin{align*}
    &H \sim \textrm{choice}(\{1, \ldots, 4\}, \textrm{weights} = p_H) \quad \quad X^1 \sim \textrm{Bernoulli}(p_1) \\ 
    &X^2 \sim  \textrm{Bernoulli}(p_{2_{X^1, H}}) \quad\quad X^3 \sim \textrm{Bernoulli}(p_{3_{X^2}}) \\
    &\mathcal{G}: X^4 \sim \textrm{Bernoulli}(p_{4_{X^3}}) \quad\quad \mathcal{H}: X^4 \sim \textrm{Bernoulli}(p_{4_{X^3, X^1}})
\end{align*}
where the subscript $p_{2_{X^1, U}}$ indicate that for an outcome of $(X^1, U)$, the Bernoulli distribution uses the probability in the corresponding entry of $p_2$ (and similar for $p_3$ and $p_4$).

The Gaussian data in \cref{sec:expverma} was generated by the structural equation model
\begin{align*}
    H \coloneqq \epsilon_H
    \quad\quad X^1 \coloneqq \epsilon_{X^1}
    \quad\quad X^2 \coloneqq X^1 + H + \epsilon_{X^2} \\
    X^3 \coloneqq X^2 + 2\epsilon_{X^3}
    \quad\quad X^4 \coloneqq \theta \cdot X^1 + X^3 + H + \epsilon_{X^4}
\end{align*}
where $\epsilon_H, \epsilon_{X^j} \sim \mathcal{N}(0, 1)$, and $\theta \in\{0, 0.3\}$ indicates the absence or presence of the edge $X^1\rightarrow X^4$.

The non-Gaussian data in \cref{sec:expverma} was generated by the structural equation model
\begin{align*}
    H \coloneqq \tfrac{1}{2}\cdot\epsilon_H\cdot \epsilon_H
    \quad\quad X^1 \coloneqq \gamma_{X^1}
    \quad\quad X^2 \coloneqq X^1\cdot H + \epsilon_{X^2} \\
    X^3 \coloneqq X^2\cdot X^2 + \tfrac{3}{2}\epsilon_{X^3}
    \quad\quad X^4 \coloneqq \theta \cdot X^1 + X^3 + H + \epsilon_{X^4},
\end{align*}
where $\epsilon_H, \epsilon_{X^j} \sim \mathcal{N}(0,1)$ and $\gamma_{X^1}$ follows a $\Gamma(2)$-distribution, and $\theta \in\{0, 0.3\}$ indicates the absence or presence of the edge $X^1\rightarrow X^4$.

\subsubsection*{\texorpdfstring{\cref{sec:icpplu}}{}}
The data in \cref{sec:icpplu} was generated by the structural equation model
\begin{align*}
        X^1 \coloneqq \epsilon_{X^1}
        \quad\quad X^3 \coloneqq \epsilon_{X^3}
        \quad\quad X^2 \coloneqq X^3 + X^1 + 2 \epsilon_{X^2} \\
        Y  \coloneqq X^2 + X^3 + 0.3 \epsilon_Y
        \quad\quad X^4 \coloneqq -Y + X^2 + X^3 + 0.7 \epsilon_{X^4}
\end{align*}

\section{Proofs}\label{sec:proofs}

\subsection{Proof of \texorpdfstring{\Cref{thm:asymptotic-level-SIR}}{}}
\begin{proof}[Proof of \cref{thm:asymptotic-level-SIR}]
We show the statement only for $\limsup$, the corresponding statement for $\liminf$ follows by replacing $\limsup$ with $\liminf$ everywhere. 

Let $p$ and $q$ denote the respective densities of $P$ and $Q$ with respect to the dominating measure $\mu$.
By assumption $p = \tau(q)$, so $p(x) \propto r(x)q(x)$. 
Let $\bar{r}$ be the normalized version of $r$ satisfying $p(x) = \bar{r}(x) q(x)$. 
Recall that we call a sequence $(i_1, \ldots, i_m)$ distinct if for all $\ell \neq \ell'$ we have $i_\ell \neq i_{\ell'}$.
The resampling scheme $\Psi_{\texttt{DRPL}}$, defined by \cref{eq:SIR-weights}, samples from the space of distinct sequences $(i_1, \ldots, i_m)$, where every sequence has probability 
$w_{(i_1, \ldots, i_m)} \propto \prod_{\ell=1}^m r(X_{i_\ell})$. The normalization constant here is the sum over the weights in the entire space of distinct sequences, that is,
    \begin{align*}
        w_{(i_1, \ldots, i_m)} = \frac{\prod_{\ell=1}^m r(X_{i_\ell})}{\sum_{\substack{(j_1, \ldots, j_m)\\ \text{distinct}}} \prod_{\ell=1}^m r(X_{j_\ell})} = \frac{\prod_{\ell=1}^m \bar{r}(X_{i_\ell})}{\sum_{\substack{(j_1, \ldots, j_m)\\ \text{distinct}}} \prod_{\ell=1}^m \bar{r}(X_{j_\ell})}.
    \end{align*}
    Thus, taking an expectation involving $\varphi_m(\Psi_{\texttt{DRPL}}^{r, m}(\bX_n, U))$, amounts to evaluating $\varphi_m$ in all distinct sequences $X_{i_1}, \ldots, X_{i_m}$ and weighting with the probabilities $w_{(i_1, \ldots, i_m)}$.
    \begin{align}
    \P_{Q}(\varphi_m(\Psi_{\texttt{DRPL}}^{r, m}(\bX_n, U)) = 1)
        = \E_{Q}\left[\frac{\frac{1}{\tfrac{n!}{(n-m)!}}\sum\limits_{\substack{(i_1, \ldots, i_m) \\ \text{distinct}}}
        \left(\prod\limits_{\ell=1}^m \bar{r}(X_{i_\ell})\right)\mathds{1}_{\{\varphi_m(X_{i_1}, \ldots, X_{i_m}) = 1\}}}
        {\frac{1}{\tfrac{n!}{(n-m)!}}\sum\limits_{\substack{(j_1, \ldots, j_m) \\ \text{distinct}}}\prod\limits_{\ell=1}^m \bar{r}(X_{j_\ell})}\right], \label{eq:proof-ratio-of-averages}
    \end{align}
where we divide by the number of distinct sequences $\tfrac{n!}{(n-m)!}$ in both numerator and denominator.

Let $c(n,m)$ and $d(n,m)$ be the numerator and denominator terms of \cref{eq:proof-ratio-of-averages}, i.e.,
\begin{align*}
    c(n,m) &\coloneqq \frac{1}{\tfrac{n!}{(n-m)!}}\sum\limits_{\substack{(i_1, \ldots, i_m) \\ \text{distinct}}}
        \left(\prod\limits_{\ell=1}^m \bar{r}(X_{i_\ell})\right)\mathds{1}_{\{\varphi_m(X_{i_1}, \ldots, X_{i_m}) = 1\}}, \\
        d(n,m) &\coloneqq \frac{1}{\tfrac{n!}{(n-m)!}}\sum\limits_{\substack{(j_1, \ldots, j_m) \\ \text{distinct}}}\prod\limits_{\ell=1}^m \bar{r}(X_{j_\ell}).
\end{align*}
We want to show that $\limsup_{n\rightarrow\infty}\E_Q\left[\frac{c(n,m)}{d(n,m)}\right]=\alpha_{\varphi}$. To see this, define for all $\delta>0$ the set $A_{\delta}\coloneqq\{|d(n,m)-1|\leq\delta\}$. It holds for all $\delta\in(0,1)$ that
\begin{align*}
    \E_Q\left[\frac{c(n,m)}{d(n,m)}\right]
    &=\E_Q\left[\frac{c(n,m)}{d(n,m)} \mathds{1}_{A_{\delta}}\right]+\E_Q\left[\frac{c(n,m)}{d(n,m)} \mathds{1}_{A_{\delta}^c}\right]\\
    &\leq\E_Q\left[\frac{c(n,m)}{1-\delta} \mathds{1}_{A_{\delta}}\right]+\E_Q\left[\frac{c(n,m)}{d(n,m)} \mathds{1}_{A_{\delta}^c}\right]\\
    &\leq\E_Q\left[\frac{c(n,m)}{1-\delta}\right]+\P_Q\left(A_{\delta}^c\right)\\
    &=\frac{1}{1-\delta}\P_P(\varphi_m(X_1, \ldots, X_m)=1)+\P_Q\left(A_{\delta}^c\right),
\end{align*}
where we used that $\frac{c(n,m)}{d(n,m)}\leq 1$ and \cref{lemma:sum-distinct-weights} (a). Further combining Chebyshev's inequality with \cref{lemma:sum-distinct-weights} (b) and (d), it follows that
\begin{equation*}
    \lim_{n\rightarrow\infty}\P_Q\left(A_{\delta}^c\right)\leq\lim_{n\rightarrow\infty}\frac{\E_Q[(d(n,m)-1)^2]}{\delta^2}=\lim_{n\rightarrow\infty}\frac{\VAR_Q(d(n,m))}{\delta^2}=0.
\end{equation*}
Hence, using that $\limsup_{k\rightarrow\infty}\P_{P}(\varphi_k(X_1,\ldots,X_k)=1)=\alpha_{\varphi}$ we have shown for all $\delta\in(0,1)$ that
\begin{equation}
\label{eq:upperbound_limit}
    \limsup_{n\rightarrow\infty}\E_Q\left[\frac{c(n,m)}{d(n,m)}\right]\leq \frac{1}{1-\delta}\alpha_{\varphi}.
\end{equation}
Similarly, we also get for all $\delta\in (0,1)$ the following lower bound
\begin{align*}
    \E_Q\left[\frac{c(n,m)}{d(n,m)}\right]
    &=\E_Q\left[\frac{c(n,m)}{d(n,m)} \mathds{1}_{A_{\delta}}\right]+\E_Q\left[\frac{c(n,m)}{d(n,m)}\mathds{1}_{A_{\delta}^c}\right]\\
    &\geq\E_Q\left[\frac{c(n,m)}{1+\delta} \mathds{1}_{A_{\delta}}\right]\\
    &\geq\frac{1}{1+\delta}\E_Q\left[c(n,m)\right],
\end{align*}
where in the last inequality we used that $c(n,m)\geq 0$. Again using \cref{lemma:sum-distinct-weights} (a) and that $\limsup_{k\rightarrow\infty}\P_{P}(\varphi_k(X_1,\ldots,X_k)=1)=\alpha_{\varphi}$, we get that  for all $\delta\in(0,1)$ that
\begin{equation}
\label{eq:lowerbound_limit}
    \limsup_{n\rightarrow\infty}\E_Q\left[\frac{c(n,m)}{d(n,m)}\right]\geq \frac{1}{1+\delta}\alpha_{\varphi}.
\end{equation}
Combining both limits \eqref{eq:upperbound_limit} and \eqref{eq:lowerbound_limit} and 
using that $\delta\in(0,1)$ is arbitrary, this proves that
\begin{equation*}
    \limsup_{n\rightarrow\infty}\E_Q\left[\frac{c(n,m)}{d(n,m)}\right]=\alpha_{\varphi},
\end{equation*}
which completes the proof of \cref{thm:asymptotic-level-SIR}.

\end{proof}

\begin{lemma}[Distinct draws]\label{lemma:sum-distinct-weights}

    Let $P\in \cP$ and $Q\in\mathcal{Q}$ have densities $p$ and $q$ with respect to a dominating measure $\mu$. Let $\bar{r}:\mathcal{X}\rightarrow [0,\infty)$ satisfy for all $x\in\mathcal{X}$ that $p(x)=\bar{r}(x)q(x)$. Let $c(n,m)$ and $d(n,m)$ be defined by
    \begin{align}
        c(n,m) &\coloneqq \frac{1}{\tfrac{n!}{(n-m)!}}\sum\limits_{\substack{(i_1, \ldots, i_m) \\ \text{distinct}}}
        \left(\prod\limits_{\ell=1}^m \bar{r}(X_{i_\ell})\right)\mathds{1}_{\{\varphi_m(X_{i_1}, \ldots, X_{i_m}) = 1\}}, \label{eq:conv-distinct-weights-num} \\
        d(n,m) &\coloneqq \frac{1}{\tfrac{n!}{(n-m)!}}\sum\limits_{\substack{(j_1, \ldots, j_m) \\ \text{distinct}}}\prod\limits_{\ell=1}^m \bar{r}(X_{j_\ell}). \label{eq:conv-distinct-weights-denom}
    \end{align}
    Then, if $m$ and $Q$ satisfy \cref{assump:m-rate-n} and \cref{assump:finite-second-moment} it holds that
    \begin{enumerate}
        \item[(a)] $\E_Q[c(n,m)]=\P_{P}(\varphi_m(X_1,\ldots,X_m)=1)$,
        \item[(b)] $\E_Q[d(n,m)]=1$,
        \item[(c)] $\lim_{n\rightarrow\infty}\VAR_Q[c(n,m)]=0$,
        \item[(d)] $\lim_{n\rightarrow\infty}\VAR_Q[d(n,m)]=0$.
    \end{enumerate}
\end{lemma}

\begin{proof}
(A) we first prove the statements for the means, i.e., (a) and (b), and (B) we then prove the statements for the variances, i.e., (c) and (d).

\hypertarget{linkto:parta}{\textbf{Part A (means):}}
Define $\delta_m := \mathds{1}_{\{\varphi_m(X_{i_1}, \ldots, X_{i_m}) =1\}}$ (for the case \cref{eq:conv-distinct-weights-num}) or $\delta_m :=1$ (for the case \cref{eq:conv-distinct-weights-denom}). Then, in both cases it holds that
\begin{align*}
        &\E_Q\left[\frac{1}{\tfrac{n!}{(n-m)!}}\sum\limits_{\substack{(i_1, \ldots, i_m) \\ \text{distinct}}} \left(\prod\limits_{\ell=1}^m \bar{r}(X_{i_\ell})\right)\delta_m\right] \\
        &= \frac{1}{\tfrac{n!}{(n-m)!}}\sum\limits_{\substack{(i_1, \ldots, i_m) \\ \text{distinct}}} \E_Q\left[\left(\prod\limits_{\ell=1}^m \bar{r}(X_{i_\ell})\right)\delta_m\right] \\
        &= \frac{1}{\tfrac{n!}{(n-m)!}}\sum\limits_{\substack{(i_1, \ldots, i_m) \\ \text{distinct}}} \int \left(\prod_{\ell=1}^m \bar{r}(x_{i_\ell})q(x_{i_\ell})\right)\delta_m \mathrm{d}\mu^m(x_{i_1}, \ldots, x_{i_m})\\
        &= \frac{1}{\tfrac{n!}{(n-m)!}} \sum\limits_{\substack{(i_1, \ldots, i_m) \\ \text{distinct}}} \int\left(\prod_{\ell=1}^m p(x_{i_\ell})\right)\delta_m \mathrm{d}\mu^m(x_{i_1}, \ldots, x_{i_m})\\
        &= \frac{1}{\tfrac{n!}{(n-m)!}}\sum\limits_{\substack{(i_1, \ldots, i_m) \\ \text{distinct}}} \E_P[\delta_m] \\
        &= \E_P[\delta_m]
    \end{align*}
    In the second and fourth equality, we use that $i_1,\ldots,i_{m}$ are all distinct, and in the last equality, we use that the number of distinct sequences $(i_1, \ldots, i_m)$ is $\tfrac{n!}{(n-m)!}$.
    Consequently the term in \cref{eq:conv-distinct-weights-num} has mean $\E_P[\mathds{1}_{\{\varphi_m(X_{i_1}, \ldots, X_{i_m}) =1\}}] = \P_P(\varphi_m(X_1, \ldots, X_m) = 1)$ and
    the term in \cref{eq:conv-distinct-weights-denom} has mean $1$. 
    
    \hypertarget{linkto:partb}{\textbf{Part B (variances):}} We begin by expressing $\frac{1}{\frac{n!}{(n-m)!}}\sum\limits_{\substack{(i_1, \ldots, i_m) \\ \text{distinct}}} \left(\prod\limits_{\ell=1}^m \bar{r}(X_{i_\ell})\right)\delta_m$
    as a U-statistic \citep{serfling1980approximation}. A U-statistic has the form
    \begin{equation}
    \label{eq:ustat-def}
        \frac{1}{\frac{n!}{(n-m)!}}\sum\limits_{\substack{(i_1, \ldots, i_m) \\ \text{distinct}}} h_m(Z_{i_1},\ldots,Z_{i_m})
    \end{equation}
    for some symmetric function
    $h_m(z_1,\ldots,z_m)$ (called a kernel function). In our case, the kernel function is $h_m(X_{i_1}, \ldots, X_{i_m}) := \prod_{\ell=1}^m \bar{r}(X_{i_\ell})\delta_m$. The variance of the corresponding U-statistic \citep[see][Section 5.2]{serfling1980approximation} is given by
    \begin{align}
        \VAR_Q\left(\frac{1}{\frac{n!}{(n-m)!}}\sum\limits_{\substack{(i_1, \ldots, i_m) \\ \text{distinct}}} \left(\prod\limits_{\ell=1}^m \bar{r}(X_{i_\ell})\right)\delta_m \right) = \binom{n}{m}^{-1}\sum_{v=1}^m \binom{m}{v}\binom{n-m}{m-v}\zeta_v \label{eq:var-u-statistic}
    \end{align}
    where for all $v \in \{1, \ldots, m\}$
    \begin{align*}
        \zeta_v := \VAR_Q\left(\E_Q[h_m(X_{i_1}, \ldots, X_{i_m}) \mid X_{i_1}, \ldots, X_{i_v}]\right).
    \end{align*}
    We now bound $\zeta_v$ from above by the second moment as follows
    \begin{equation*}
        \zeta_v\leq \E_Q\left[\E_Q[h_m(X_{i_1}, \ldots, X_{i_m}) \mid X_{i_1}, \ldots, X_{i_v}]^2\right].
    \end{equation*}
    Moreover, using that $\delta_m$ is upper bounded by $1$, we get for both cases \cref{eq:conv-distinct-weights-denom,eq:conv-distinct-weights-num} that
    \begin{align}
       \zeta_v &\leq\E_Q\left[\E_Q[h_m(X_{i_1}, \ldots, X_{i_m}) \mid X_{i_1}, \ldots, X_{i_v}]^2\right] \nonumber \\
       &\leq \E_Q\left[\E_Q\left[\prod_{\ell=1}^m \bar{r}(X_{i_\ell}) \mid X_{i_1}, \ldots, X_{i_v}\right]^2\right]. \label{eq:U-stat-zeta-j}
    \end{align}
    Next, since $(i_1, \ldots, i_m)$ are distinct, the variables $X_{i_1}, \ldots, X_{i_m}$ are independent. Hence we have that
        \begin{align}
        &\E_Q\left[\prod_{\ell=1}^m \bar{r}(X_{i_\ell}) \mid X_{i_1}, \ldots, X_{i_v}\right] \nonumber\\
        &= \left(\prod_{\ell=1}^v \bar{r}(X_{i_\ell})\right)\prod_{\ell=v+1}^m \E_Q \left[ \bar{r}(X_{i_\ell})\right] \nonumber\\
        &= \left(\prod_{\ell=1}^v \bar{r}(X_{i_\ell})\right), \label{eq:zeta-product-of-weights}
    \end{align}
    where the last equality follows because 
    \begin{align*}
        \E_Q[\bar{r}(X_{i_\ell})] = \int \bar{r}(x)q(x)\mathrm{d}\mu(x)
        = \int p(x) \mathrm{d}\mu(x)
        = 1.
    \end{align*}
    Next, combining \cref{eq:zeta-product-of-weights,eq:U-stat-zeta-j} we get that
    \begin{align*}
        \zeta_v &\leq \E_Q\left[\right(\prod_{\ell=1}^v \bar{r}(X_{i_\ell})\left)^2\right] \\
        &= \prod_{\ell=1}^v \E_Q\left[\bar{r}(X_{i_\ell})^2\right] \\
        &= \E_Q\left[\bar{r}(X_{i_1})^2\right]^v.
    \end{align*}
    Here, we again use the independence of the distinct terms. Plugging this into \cref{eq:var-u-statistic}, we get
    \begin{align*}
        &\VAR_Q\left(\frac{1}{\frac{n!}{(n-m)!}}\sum\limits_{\substack{(i_1, \ldots, i_m) \\ \text{distinct}}} \left(\prod\limits_{\ell=1}^m \bar{r}(X_{i_\ell})\right)\delta_m \right)\\
        &\leq \binom{n}{m}^{-1}\sum_{\ell=1}^m \binom{m}{\ell}\binom{n-m}{m-\ell}\E_Q\left[\bar{r}(X_{i_1})^2\right]^\ell.
    \end{align*}
    
    By \cref{assump:finite-second-moment}, $\E_Q\left[\bar{r}(X_{i_1})^2\right] < \infty$, so \cref{lemma:sum-of-binomials} implies that this converges to $0$ for $n \rightarrow \infty$.
    This shows that the variance converges to zero in both cases \cref{eq:conv-distinct-weights-denom,eq:conv-distinct-weights-num}, 
    which completes the proof of \cref{lemma:sum-distinct-weights}. 
\end{proof}

\begin{lemma}\label{lemma:sum-of-binomials}
    Let $m = o(\sqrt{n})$ as $n$ goes to infinity. Then for any $K \geq 0$, it holds that
\begin{align}\label{eq:U-statistic}
    \lim_{n\rightarrow\infty}\frac{1}{\binom{n}{m}}\sum_{\ell=1}^m \binom{m}{\ell}\binom{n-m}{m-\ell}K^\ell = 0.
\end{align}
\end{lemma}
\begin{remark}
    The Chu-Vandermonde identity states that $\frac{1}{\binom{n}{m}}\sum_{\ell=0}^m \binom{m}{\ell}\binom{n-m}{m-\ell} = 1$. In light of this identity, one may be surprised that when including the exponentially growing term, $K^\ell$, the sum vanishes. The reason is that the summation in \cref{eq:U-statistic} starts at $\ell = 1$, not $\ell = 0$, and since $n$ grows at least quadratically in $m$, $\binom{n-m}{m-\ell}$ for $\ell = 0$ dominates all the other summands as $n$ (and thereby also $m$) approaches $\infty$.
\end{remark}
\begin{proof}
Denote by $s_\ell$ the $\ell$'th summand, i.e., 
\begin{align*}
    s_\ell := \binom{m}{\ell}\binom{n-m}{m-\ell}K^\ell.
\end{align*}
It then holds for all $\ell\in\{1,\ldots,m-1\}$ that
\begin{align*}
    \frac{s_{\ell+1}}{s_\ell} &= 
    \frac{\binom{m}{\ell+1}\binom{n-m}{m-\ell-1}K^{\ell+1}}{\binom{m}{\ell}\binom{n-m}{m-\ell}K^{\ell}} \\
    &=\frac{\frac{m!}{(\ell+1)!(m-\ell-1)!}\frac{(n-m)!}{(m-\ell-1)!(n-2m+\ell+1)!}}{\frac{m!}{\ell!(m-\ell)!}\frac{(n-m)!}{(m-\ell)!(n-2m+\ell)!}}K \\
    &= \frac{(m-\ell)^2}{(\ell+1)(n-2m+\ell+1)}K \\
   &\leq \frac{m^2}{2(n-2m+2)}K.
\end{align*}
Since, by assumption, $m=o(\sqrt{n})$, this converges to $0$ as $n$ goes to infinity. In particular, there exists a constant $c\in(0,1)$ such that for $n$ sufficiently large it holds for all $\ell\in\{1,\ldots,m-1\}$ that $\frac{s_{\ell+1}}{s_\ell} \leq c$. This implies that $s_\ell \leq s_1 c^{\ell-1}$, and hence also
\begin{align*}
    \sum_{\ell=1}^m s_\ell \leq s_1 \sum_{\ell=1}^m c^{\ell-1} \leq s_1 \frac{1}{1-c},
\end{align*}
where for the last inequality we used the explicit solution of a geometric sum.
We now conclude the proof by explicitly bounding \cref{eq:U-statistic} as follows
\begin{align}
    \frac{1}{\binom{n}{m}}\sum_{\ell=1}^m \binom{m}{\ell}\binom{n-m}{m-\ell}K^\ell &= \frac{1}{\binom{n}{m}}\sum_{\ell=1}^m s_\ell \nonumber\\
    &< \frac{1}{1-c} \frac{s_1}{\binom{n}{m}} \nonumber\\
    &= \frac{K}{1-c}\frac{m \binom{n-m}{m-1}}{\binom{n}{m}} \nonumber\\
    &= \frac{K}{1-c} \frac{m^2}{n}\frac{\binom{n-m}{m-1}}{\binom{n-1}{m-1}},\label{eq:upperboundlemma3}
\end{align}
where in the last equation we use the relation $\binom{n}{m} = \frac{n}{m}\binom{n-1}{m-1}$. Using that by assumption $\lim_{n\rightarrow\infty}\frac{m^2}{n}=0$ and that $\frac{\binom{n-m}{m-1}}{\binom{n-1}{m-1}} \leq 1$, it immediately follows that \eqref{eq:upperboundlemma3} converges to zero. This completes the proof of \cref{lemma:sum-of-binomials}.
\end{proof}

\subsection{Proof of \texorpdfstring{\cref{thm:necessity-sqrt-n}}{}}
\begin{proof}
We explicitly construct an example hypothesis test for which the worst case
rate is achieved.
We construct a target and observation density on
$[0,\infty)$. First, for fixed $\alpha\in (0,1)$
and all $v\in\mathbb{N}\setminus \{0\}$ define
\begin{equation*}
  c_{v}\coloneqq (1-\alpha)^{\frac{1}{v}}
  \quad\text{and}\quad
  p_{v}\coloneqq 1-(v+1)^{-\epsilon},
\end{equation*}
and $c_0:=0$ and
$p_0:=0$, with $\epsilon\in (0, \frac{\ell}{\ell-1})$ 
to be chosen below. Then, for all $v\in\mathbb{N}$ define
\begin{equation*}
  f_{v}\coloneqq c_{v+1}-c_{v}
  \quad\text{and}\quad
  g_{v}\coloneqq p_{v+1}-p_{v}.
\end{equation*}
Using these sequences, we 
define
the following two densities:
\begin{itemize}
\item[(1)] Target density (cdf is denoted by $F$): For all $x \in \mathbb{R}$, we define
  \begin{equation*}
    f(x):=\sum_{v=0}^{\infty}\mathds{1}_{\{v\leq x< v+1\}}f_{v},
  \end{equation*}
\item[(2)] Observation density (cdf is denoted by $G$): For all $x \in \mathbb{R}$, we define
  \begin{equation*}
    g(x):=\sum_{v=0}^{\infty}\mathds{1}_{\{v\leq x < v+1\}}g_{v}.
  \end{equation*}
\end{itemize}
As $\lim_{v \rightarrow \infty} c_{v} = \lim_{v \rightarrow \infty} p_{v} =1$,
these functions are indeed
densities.
Moreover, one can verify that
for all $m \in \mathbb{N}$, we have
\begin{equation*}
  F(m)=c_{m}
  \quad\text{and}\quad
  G(m)=p_{m}.
\end{equation*}
A visualization of the densities is given in \cref{fig:densities}.

\begin{figure}[t]
    \centering
    \includegraphics[width=0.5\textwidth]{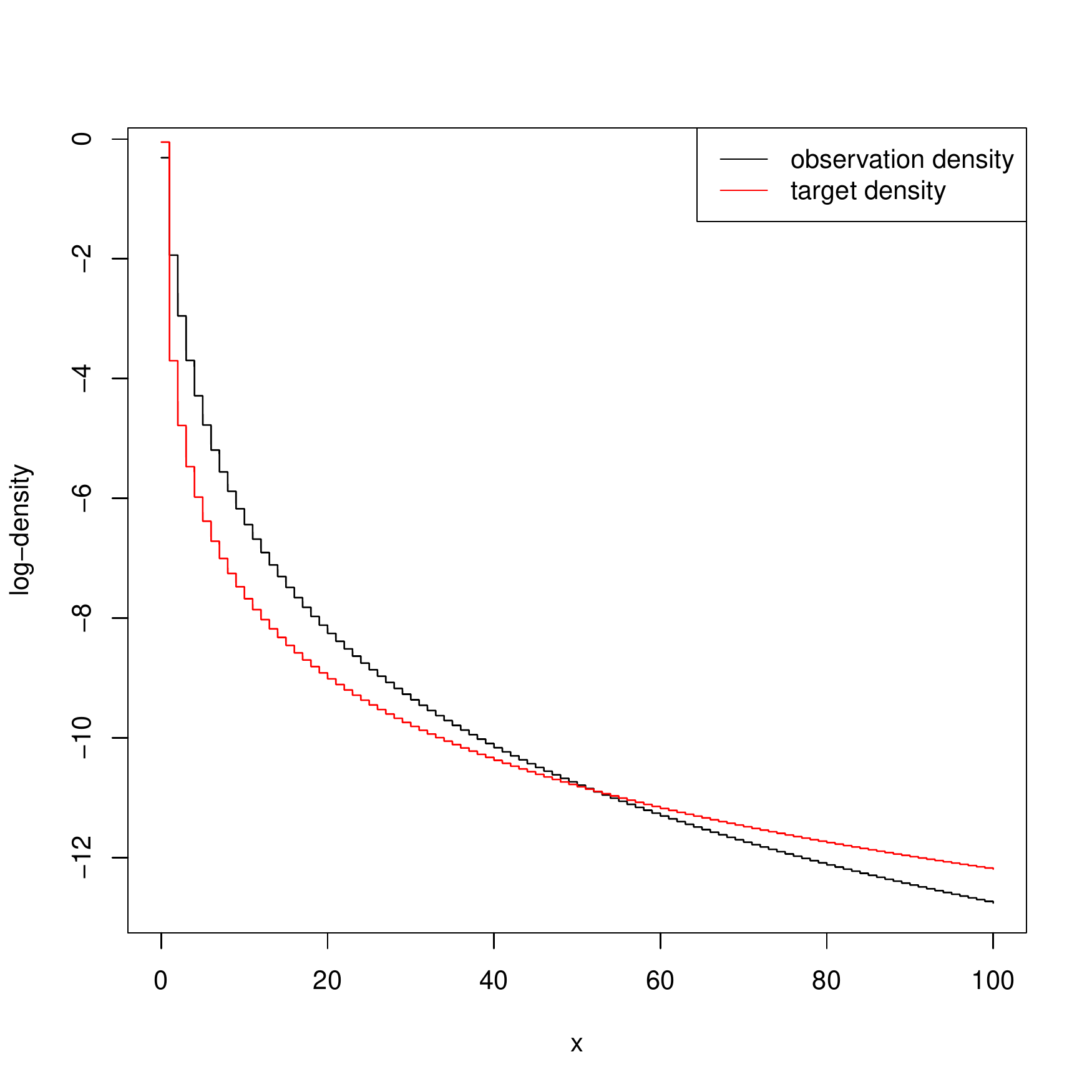}
    \caption{Visualization of densities in the proof of \cref{thm:necessity-sqrt-n} with $\epsilon=1.9$. The tail of the target density eventually becomes larger than that of the observation density.}
    \label{fig:densities}
\end{figure}
Finally, if we define, for all $x \geq 0$, $r(x):=\frac{f(x)}{g(x)}$ then we get
\begin{align*}
  \mathbb{E}_{g}\left[r(X)^\ell\right]
  &=\int_0^{\infty}r(x)^2g(x)dx\\
  &=\sum_{v=1}^{\infty}\left(\frac{f_{v}}{g_{v}}\right)^\ell g_{v}\\
  &=\sum_{v=1}^{\infty}\frac{(c_{v+1}-c_{v})^\ell}{(p_{v+1}-p_{v})^{\ell-1}}.
\end{align*}
This series converges for all possible parameter choices $\epsilon\in(0,\frac{\ell}{\ell-1})$ because
\begin{itemize}
\item[(a)] $(c_{v+1}-c_{v})\sim -\log(1-\alpha)v^{-2}$ as $v \rightarrow \infty$ and
\item[(b)] $(p_{v+1}-p_{v})\sim \epsilon v^{-(\epsilon+1)}$ as $v \rightarrow \infty$.
\end{itemize}
\begin{quote}
    (Indeed, both results follow from the mean value theorem as follows: First, for (a) applying
the mean value theorem to $x \mapsto (1-\alpha)^{1/x}$ implies that for all $v \in \mathbb{N}$ there exists $\xi_v \in [v, v+1]$ such that
\begin{align*}
    \frac{c_{v+1} - c_v}{v+1 - v} = -\log(1-\alpha)\frac{(1-\alpha)^{1/\xi_v}}{\xi_v^2}.
\end{align*}
We therefore get 
$$
\lim_{v \rightarrow \infty}
({c_{v+1} - c_v}){v^2} 
= 
-\log(1-\alpha)
\lim_{v \rightarrow \infty}
\frac{v^2(1-\alpha)^{1/\xi_v}}{\xi_v^2} =
-\log(1-\alpha).
$$
Similarly, for (b), we apply the mean value theorem to $x \mapsto 1-(x+1)^{-\epsilon}$ to get that for all $v \in \mathbb{N}$ there is a $\xi_v \in [v, v+1]$ such that
\begin{align*}
    \frac{p_{v+1}-p_v}{v+1-v} = \epsilon(\xi_v + 1)^{-(\epsilon+1)}.
\end{align*}
Again taking the limits we get
$$
\lim_{v \rightarrow \infty}
({p_{v+1} - p_v})v^{(\epsilon+1)} 
= \epsilon\lim_{v \rightarrow \infty}
(\xi_v + 1)^{-(\epsilon+1)}v^{(\epsilon+1)} = \epsilon.
$$
This completes the proofs of (a) and (b).)
\end{quote}

Now, consider the null hypothesis
\begin{equation*}
  H_0\coloneqq \{P_{f}\}.
\end{equation*}
For all $m\in\mathbb{N}$, we define the test
$\phi_m:[0,\infty)^m\rightarrow \{0,1\}$ for all $x_1,\ldots,x_n\in [0,\infty)$ by
\begin{equation*}
  \phi_m(x_1,\ldots,x_m)\coloneqq \mathds{1}_{\{\max(x_1,\ldots,x_m)\leq m\}}.
\end{equation*}
Then, it holds that
\begin{align*}
  \P_{f}(\phi_m(X_1,\ldots,X_m)=0)
  &=\P_{f}(\max(X_1,\ldots,X_m)> m)\\
  &=1-\P_{f}(X_i\leq m)^m\\
  &=1-F(m)^m\\
  &=1 - c_m^m \\
  &= 1 - [(1-\alpha)^{\tfrac{1}{m}}]^m \\
  &=\alpha.
\end{align*}
Hence, $\phi_m$ achieves valid level in the target distribution $f$. Our goal is now to show that any resampling procedure for testing under distributional shifts
with $m=n^q$ and $q> \frac{\ell-1}{\ell}$ cannot achieve asymptotic
level. Let $\Psi^m$ be the resampling scheme from the theorem that outputs a (not necessarily distinct) sample of size $m=n^q$. Then, it holds that
\begin{align*}
  \mathbb{P}_{g}(\phi_{m}(\Psi^m(X_1,\ldots,X_n))=1)
  &=\mathbb{P}_{g}(\max(\Psi^m(X_1,\ldots,X_n))\leq m)\\
  &\geq\mathbb{P}_{g}(\max(X_{1},\ldots,X_{n})\leq m)\\
  &=\mathbb{P}_{g}(X_{i}\leq m)^n\\
  &=G(m)^n\\
  &=p_m^n \\
  &=(1-(m+1)^{-\epsilon})^n\\
  &=\exp(n\log(1-(m+1)^{-\epsilon})).
\end{align*}
Taylor expanding $x \mapsto \log(x)$ in $x_0 = 1$ yields
\begin{align*}
	\log(x) = \log(x_0) + \frac{1}{x_0}(x-x_0) + \frac{1}{2}\frac{-1}{\xi_x^2}(x- x_0)^2
\end{align*}
for an $\xi_x \in (x, 1)$. Plugging in $x = 1 - (m+1)^{-\epsilon}$, we get
\begin{align} \label{eq:logoneminus}
\log(1-(m+1)^{-\epsilon}) = 	-(m+1)^{-\epsilon} -\frac{1}{2\xi_m^2}(m+1)^{-2\epsilon},
\end{align}
where $\xi_x$ is lower bounded by $(1-(m+1)^{-\epsilon})$, and so $\xi_m\to 1$ 
and thus $\xi_m^2\to 1$
for $m\to\infty$. 
Next, observe that
\begin{equation*}
  n(m+1)^{-\epsilon}\leq nm^{-\epsilon}=n^{1-q\epsilon}.
\end{equation*}
Now if we select $\frac{\ell-1}{\ell}>\epsilon>1/q$ (this is always possible because
$q>1/2$), we have that $q\epsilon > 1$, and it holds that
$\lim_{n\rightarrow\infty}n^{1-q\epsilon}=0$. 
For the same reason, we have $n (m+1)^{-2\epsilon} \leq nm^{-\epsilon} \to 0$ as $n \rightarrow \infty$. 
Combining this with the
\eqref{eq:logoneminus}, we get
$$
\lim_{n \rightarrow \infty} n\log(1-(m+1)^{-\epsilon}) = 0
$$
and thus 
\begin{equation*}
  \lim_{n\rightarrow\infty}\mathbb{P}_{g}(\phi_{m}(\Psi^m(X_{1},\ldots,X_{n}))=1)\geq 1.
\end{equation*}
This completes the proof of \cref{thm:necessity-sqrt-n}.
\end{proof}

\subsection{Proof of \texorpdfstring{\cref{thm:asymptotic-level-unknown-weights}}{}}
\begin{proof}
The proof is similar to the the proof of \cref{thm:asymptotic-level-SIR} 
but we will need to adjust for the estimation of the distributional shift factor. In particular, we will reprove the results in \cref{lemma:sum-distinct-weights} when using the estimator $\hat{r}_{n_1}$.

Fix any $P \in H_0$ and let $Q \in \tau^{-1}(\{P\})$. Denote by $p$ and $q$ their respective densities with respect to the dominating measure $\mu$.
We begin by recalling the details for the sample splitting procedure described in \cref{alg:resampling-and-testing-unknown-shift}: $\bX_n$ is split into two disjoint data sets $\bX_{n_1}$ and $\bX_{n_2}$ of sizes $n_1, n_2$, where $n_1 + n_2 = n$. The assumptions $n_1^a = \sqrt{n_2}$ and $m = o(\min(n^a, \sqrt{n}))$, ensure that $m = o(n_1^a)$ and $m = o(\sqrt{n_2})$.\footnote{When $n_1^a = \sqrt{n_2}$ and $n_1 + n_2 = n$, we have $n = n_1^{2a} + n_1$ and $n = n_2 + n_2^{1/(2a)}$. If $a > 1/2$, we have $m = o(\sqrt{n}) = o(\sqrt{n_1^{2a}}) = o(n_1^a)$, and $m = o(\sqrt{n}) = o(\sqrt{n_2})$. Similar arguments apply if $a < 1/2$.} We use $\bX_{n_1}$ to fit an estimator $\hat{r}_{n_1}$ of $r_q$ and then use $\bX_{n_2}$ for the resampling. When taking expectations over $\bX_{n_1}$ we use the notation $\E_{Q_1}$. Similarly, $\E_{Q_2}$ denotes an expectation over $\bX_{n_2}$. We write $\E_Q$ when taking expectations with respect to the entire sample $\bX_n$.
Let $I_1 \coloneqq \{1, \ldots, n_1\}$ and $I_2\coloneqq \{n_1 + 1, \ldots, n_1 + n_2\}$ be the indices of $\bX_{n_1}$ and $\bX_{n_2}$ respectively. 

Using the same argument as we used to derive \cref{eq:proof-ratio-of-averages} in the proof of \cref{thm:asymptotic-level-SIR}, we get that
\begin{align}
    \P_Q(\varphi_m(\Psi^{\hat{r}_{n_1}}(\bX_{n_2}, U) = 1))
    = \E_Q\left[\frac{\frac{1}{\tfrac{n_2!}{(n_2-m)!}}\sum\limits_{\substack{(i_1, \ldots, i_m)\\ \text{distinct from }I_2}}
    \left(\prod\limits_{\ell=1}^m \hat{r}_{n_1}(X_{i_v})\right)\mathds{1}_{\{\varphi_m(X_{i_1}, \ldots, X_{i_m}) = 1\}}}
    {\frac{1}{\tfrac{n_2!}{(n_2-m)!}}\sum\limits_{\substack{(i_1, \ldots, i_m) \\ \text{distinct from }I_2}}\prod\limits_{\ell=1}^m \hat{r}_{n_1}(X_{i_\ell})}\right],\label{eq:proof-ratio-of-averates-unkown-weights}
\end{align}
where $X_{i_\ell}$ are observations from $\bX_{n_2}$. As in the proof of \cref{lemma:sum-distinct-weights}, we prove the convergence in probability of the numerator and denominator in \cref{eq:proof-ratio-of-averates-unkown-weights} separately. Again, we do this in two steps: (A) We show that the means converge to the desired quantity and (B) we show that the variances converge to zero. First we show the following intermediate result. 

\textbf{Intermediate result:} 
Let $\epsilon(n_1) \coloneqq \sup_{x \in \cX} \E_{Q_1} \left|\left(\frac{\hat{r}_{n_1}(x)}{r_q(x)}\right)^{n_1^a} - 1 \right|$ and consider a sequence  $i_1, \ldots, i_m$ from the indices of $\mathbf{X}_2$. Then, for $n_1$ sufficiently large and using Jensen's inequality, it holds $Q_{n_2}$-a.s. that
\begin{align}
    \left|\E_{Q_1}\left[\prod\limits_{\ell=1}^m \frac{\hat{r}_{n_1}(X_{i_\ell})}{r_q(X_{i_\ell})}\right] -1 \right|\nonumber
    &\leq \E_{Q_1}\left[\left|\prod\limits_{\ell=1}^m \frac{\hat{r}_{n_1}(X_{i_\ell})}{r_q(X_{i_\ell})} - 1\right|\right]\nonumber \\
    &\leq \sup_{x_{i_1}, \ldots, x_{i_m} \in\cX} \E_{Q_1}\left[\left|\prod\limits_{\ell=1}^m \frac{\hat{r}_{n_1}(x_{i_\ell})}{r_q(x_{i_\ell})} - 1\right|\right] \nonumber \\
    &\leq \sup_{x\in\cX}\E_{Q_1}\left[\left|\left(\frac{\hat{r}_{n_1}(x)}{r_q(x)}\right)^m - 1\right|\right] \nonumber\\
    &\leq \sup_{x\in\cX}\E_{Q_1}\left[\left|\left(\frac{\hat{r}_{n_1}(x)}{r_q(x)}\right)^{n_1^a} - 1\right|\right] \nonumber\\
    &=\epsilon(n_1).\label{eq:proof-intermediate-bound1}
\end{align}
The last inequality holds because by assumption, $m = o(n_1^a)$ when $n_1 \to \infty$, so for $n_1$ sufficiently large, $n_1^a > m$. For any $m, k\in\mathbb{N}$ we have $c^{m+k}\geq c^m \geq 1$ if $c>1$ and $c^{m+k}\leq c^m \leq 1$ if $0\leq c \leq 1$, and in either case 
it holds that $|c^{m+k}-1|\geq |c^m-1|$.
Similarly, for $i_1, \ldots, i_m$ and $i'_1, \ldots, i'_m$ from $I_2$, we get $Q_{n_2}$-a.s. that
\begin{align}
    &\left|\E_{Q_1}\left[\left(\prod\limits_{\ell=1}^m \frac{\hat{r}_{n_1}(X_{i_\ell})}{r_q(X_{i_\ell})}\right)
    \left(\prod\limits_{\ell=1}^m \frac{\hat{r}_{n_1}(X_{i_\ell'})}{r_q(X_{i_\ell'})}\right)\right] -1 \right|\nonumber\\
    &\quad\leq \E_{Q_1}\left[\left|\left(\prod\limits_{\ell=1}^m \frac{\hat{r}_{n_1}(X_{i_\ell})}{r_q(X_{i_\ell})}\right)
    \left(\prod\limits_{\ell=1}^{m} \frac{\hat{r}_{n_1}(X_{i_\ell'})}{r_q(X_{i_\ell'})}\right)- 1\right|\right]\nonumber \\
    &\quad\leq \sup_{x\in\cX}\E_{Q_1}\left[\left|\left(\frac{\hat{r}_{n_1}(x)}{r_q(x)}\right)^{2m} - 1\right|\right] \nonumber\\
    &\quad\leq \sup_{x\in\cX}\E_{Q}\left[\left|\left(\frac{\hat{r}_{n_1}(x)}{r_q(x)}\right)^{n_1^a} - 1\right|\right] \nonumber\\
    &\quad=\epsilon(n_1),\label{eq:proof-intermediate-bound2}
\end{align}
using that for $n_1$ sufficiently large, $n_1^a > 2m$. This concludes the intermediate result.

Before showing parts A and B, we introduce some notation. Depending on whether we consider the numerator or denominator case, we define either $\delta_m \coloneqq \mathds{1}_{\{\varphi_m(X_{i_1}, \ldots, X_{i_m}) = 1\}}$ or $\delta_m\coloneqq 1$. Furthermore, we introduce for any function $r:\mathcal{X}\rightarrow (0,\infty)$ the following random variable
\begin{equation*}
    M(r)\coloneqq \frac{1}{\tfrac{n_2!}{(n_2-m)!}}\sum\limits_{\substack{(i_1, \ldots, i_m) \\ \text{distinct from }I_2}} \left(\prod\limits_{\ell=1}^m r(X_{i_\ell})\right) \delta_m.
\end{equation*}

\textbf{Part A (means):} Since $\hat{r}_{n_1}(x) = r_q(x) \tfrac{\hat{r}_{n_1}(x)}{r_q(x)}$, using the independence between $\bX_{n_1}$ and $\bX_{n_2}$ we get that
\begin{align}
    \E_Q\left[M(\hat{r}_{n_1})\right]
    &= \frac{1}{\tfrac{n_2!}{(n_2-m)!}}\sum\limits_{\substack{(i_1, \ldots, i_m) \\ \text{distinct from }I_2}}\E_Q\left[ \left(\prod\limits_{\ell=1}^m r_q(X_{i_\ell})\right)\left(\prod\limits_{\ell=1}^m \frac{\hat{r}_{n_1}(X_{i_\ell})}{r_q(X_{i_\ell})}\right) \delta_m\right] \nonumber\\
    &= \frac{1}{\tfrac{n_2!}{(n_2-m)!}}\sum\limits_{\substack{(i_1, \ldots, i_m) \\ \text{distinct from }I_2}}\E_{Q_2}\left[\left(\prod\limits_{\ell=1}^m r_q(X_{i_\ell})\right)\E_{Q_1}\left[\prod\limits_{\ell=1}^m \frac{\hat{r}_{n_1}(X_{i_\ell})}{r_q(X_{i_\ell})}\right] \delta_m\right]. \label{eq:numerator-with-estimated-weights}
\end{align}
We emphasize that the expectation $\E_{Q_1}$ only averages over the randomness in estimating $\hat{r}$, and does take expectations over $X_{i_\ell}$. 
which is drawn from $Q_{n_2}$.
Furthermore, using the intermediate result \cref{eq:proof-intermediate-bound1},
we get the following  upper bound
\begin{equation*}
    \E_Q\left[M(\hat{r}_{n_1})\right]
    \leq \E_{Q_2}\left[M(r_q)\right](1+\epsilon(n_1))
\end{equation*}
and lower bound
\begin{equation*}
    \E_Q\left[M(\hat{r}_{n_1})\right]
    \geq \E_{Q_2}\left[M(r_q)\right](1-\epsilon(n_1)).
\end{equation*}
Since $m=o(\sqrt{n_2})$, we can apply \cref{lemma:sum-distinct-weights} (a) and (b) to get that the means $\E_{Q}[M(\hat{r}_{n_1})]$ of the denominator and numerator converge to the desired values.

\textbf{Part B (variances):} Next, we show that both for the numerator and denominator the variance converges to zero. To this end, we expand the second moment as follows
\begin{align*}
    \E_Q\left[M(\hat{r}_{n_1})^2\right]
    &= \E_Q\left[\frac{1}{\tfrac{n_2!}{(n_2-m)!}\tfrac{n_2!}{(n_2-m)!}}\sum\limits_{\substack{(i_1, \ldots, i_m) \\ \text{distinct from }I_2}}\sum\limits_{\substack{i_1', \ldots, i_m' \\ \text{distinct from }I_2}} \left(\prod\limits_{\ell=1}^m \hat{r}_{n_1}(X_{i_\ell})\right)\left(\prod\limits_{\ell=1}^m \hat{r}_{n_1}(X_{i_\ell'})\right) \delta_m\delta'_m\right] \\
    &= \E_{Q_2}\Bigg[\frac{1}{\tfrac{n_2!}{(n_2-m)!}\tfrac{n_2!}{(n_2-m)!}}\sum\limits_{\substack{(i_1, \ldots, i_m) \\ \text{distinct from }I_2}}\sum\limits_{\substack{i_1', \ldots, i_m' \\ \text{distinct from }I_2}} 
    \left(\prod\limits_{\ell=1}^m r(X_{i_\ell})\right)
    \left(\prod\limits_{\ell=1}^m r(X_{i_\ell'})\right)\\
    &\quad\qquad\qquad
    \E_{Q_1}\left[\left(\prod\limits_{\ell=1}^m \frac{\hat{r}_{n_1}(X_{i_\ell})}{r_q(X_{i_\ell})}\right)
    \left(\prod\limits_{\ell=1}^m \frac{\hat{r}_{n_1}(X_{i_\ell'})}{r_q(X_{i_\ell'})}\right)\right]
    \delta_m\delta'_m\Bigg].
\end{align*}
Here $\delta'_m \coloneqq \mathds{1}_{\{\varphi_m(X_{i'_1}, \ldots, X_{i'_m}) = 1\}}$.  Using the intermediate result \cref{eq:proof-intermediate-bound2} we get the following upper bound 
\begin{equation*}
    \E_Q\left[M(\hat{r}_{n_1})^2\right]
    \leq \E_{Q_2}\left[M(r_q)^2\right](1+\epsilon(n_1))
\end{equation*}
and lower bound
\begin{equation*}
    \E_Q\left[M(\hat{r}_{n_1})^2\right]
    \geq \E_{Q_2}\left[M(r_q)^2\right](1-\epsilon(n_1)).
\end{equation*}
In \cref{lemma:sum-distinct-weights} (c) and (d) we have shown that $\lim_{n\rightarrow\infty}\VAR_{Q}(M(r_q))=0$.
Hence, combining these bounds on the second moment with the above bounds on the first moment shows that also $\lim_{n\rightarrow\infty}\VAR_{Q}(M(\hat{r}_{n_1}))=0$.
This completes the proof of \cref{thm:asymptotic-level-unknown-weights}.
\end{proof}

\subsection{Proof of \texorpdfstring{\Cref{thm:finite-level-SIR}}{}}
\begin{proof}[Proof of \cref{thm:finite-level-SIR}]
The first part of the proof follows that of \cref{thm:asymptotic-level-SIR}.
Let $p$ and $q$ denote the respective densities of $P$ and $Q$ with respect to the dominating measure $\mu$.
Recall that we call a sequence $(i_1, \ldots, i_m)$ distinct if for all $\ell \neq \ell'$ we have $i_\ell \neq i_{\ell'}$.
The resampling scheme $\Psi_{\texttt{DRPL}}$ samples from the space of distinct sequences $(i_1, \ldots, i_m)$, where every sequence has probability 
$w_{(i_1, \ldots, i_m)} \propto \prod_{\ell=1}^m r(X_{i_\ell})$. The normalization constant is the sum over the weights in the entire space of distinct sequences, that is,
    \begin{align*}
        w_{(i_1, \ldots, i_m)} = \frac{\prod_{\ell=1}^m r(X_{i_\ell})}{\sum_{\substack{(j_1, \ldots, j_m)\\ \text{distinct}}} \prod_{\ell=1}^m r(X_{j_\ell})}.
    \end{align*}
    Thus, taking an expectation involving $\varphi_m(\Psi_{\texttt{DRPL}}^{r, m}(\bX_n, U))$, amounts to evaluating $\varphi_m$ in all distinct sequences $X_{i_1}, \ldots, X_{i_m}$ and weighting with the probabilities $w_{(i_1, \ldots, i_m)}$.
    \begin{align}
    \P_{Q}(\varphi_m(\Psi_{\texttt{DRPL}}^{r, m}(\bX_n, U)) = 1)
        = \E_{Q}\left[\frac{\frac{1}{\tfrac{n!}{(n-m)!}}\sum\limits_{\substack{(i_1, \ldots, i_m) \\ \text{distinct}}}
        \left(\prod\limits_{\ell=1}^m r(X_{i_\ell})\right)\mathds{1}_{\{\varphi_m(X_{i_1}, \ldots, X_{i_m}) = 1\}}}
        {\frac{1}{\tfrac{n!}{(n-m)!}}\sum\limits_{\substack{(j_1, \ldots, j_m) \\ \text{distinct}}}\prod\limits_{\ell=1}^m r(X_{j_\ell})}\right], \label{eq:proof-ratio-of-averages-finite-sample}
    \end{align}
where we divide by the number of distinct sequences $\tfrac{n!}{(n-m)!}$ in both numerator and denominator.

Let $c(n,m)$ and $d(n,m)$ be the numerator and denominator terms of \cref{eq:proof-ratio-of-averages-finite-sample}, i.e.,
\begin{align*}
    c(n,m) &\coloneqq \frac{1}{\tfrac{n!}{(n-m)!}}\sum\limits_{\substack{(i_1, \ldots, i_m) \\ \text{distinct}}}
        \left(\prod\limits_{\ell=1}^m r(X_{i_\ell})\right)\mathds{1}_{\{\varphi_m(X_{i_1}, \ldots, X_{i_m}) = 1\}}, \\
        d(n,m) &\coloneqq \frac{1}{\tfrac{n!}{(n-m)!}}\sum\limits_{\substack{(j_1, \ldots, j_m) \\ \text{distinct}}}\prod\limits_{\ell=1}^m r(X_{j_\ell}).
\end{align*}
Define for all $\delta>0$ the set $A_{\delta}\coloneqq\{d(n,m)\geq 1-\delta\}$. It holds for all $\delta\in(0,1)$ that
\begin{align*}
    \E_Q\left[\frac{c(n,m)}{d(n,m)}\right]
    &=\E_Q\left[\frac{c(n,m)}{d(n,m)} \mathds{1}_{A_{\delta}}\right]+\E_Q\left[\frac{c(n,m)}{d(n,m)} \mathds{1}_{A_{\delta}^c}\right]\\
    &\leq\E_Q\left[\frac{c(n,m)}{1-\delta} \mathds{1}_{A_{\delta}}\right]+\E_Q\left[\frac{c(n,m)}{d(n,m)} \mathds{1}_{A_{\delta}^c}\right]\\
    &\leq\E_Q\left[\frac{c(n,m)}{1-\delta}\right]+\P_Q\left(A_{\delta}^c\right)\\
    &=\frac{1}{1-\delta}\P_P(\varphi_m(X_1, \ldots, X_m)=1)+\P_Q\left(A_{\delta}^c\right),
\end{align*}
where we used that $\frac{c(n,m)}{d(n,m)}\leq 1$ and \cref{lemma:sum-distinct-weights} (a). 
Further, by applying Cantelli's inequality to $\P_Q(A_\delta^c)$, it follows that
\begin{equation*}
    \P_Q\left(A_{\delta}^c\right)\leq
   \frac{\operatorname{VAR}_Q(d(n,m))}{\operatorname{VAR}_Q(d(n,m)) + \delta^2}.
\end{equation*}
Finally, we can apply \cref{eq:var-u-statistic},
\begin{align*}
        V(n,m) \coloneqq \operatorname{VAR}_Q(d(n,m)) = &\operatorname{VAR}\left(\frac{1}{\frac{n!}{(n-m)!}}\sum\limits_{\substack{(i_1, \ldots, i_m) \\ \text{distinct}}} \left(\prod\limits_{\ell=1}^m r(X_{i_\ell})\right)\right)\\
        &= \binom{n}{m}^{-1}\sum_{\ell=1}^m \binom{m}{\ell}\binom{n-m}{m-\ell}(\E_Q\left[r(X_{i_1})^2\right]^\ell - 1),
\end{align*}
where we use that $\zeta_v$ (used in \cref{eq:var-u-statistic}) is given by

\begin{align*}
    \zeta_v &= \VAR_Q\left(\E_Q\left[\prod\limits_{\ell=1}^m r(X_{i_\ell}) \mid X_{i_1}, \ldots, X_{i_v}\right]\right) \\
    &= \VAR_Q\left(\left(\prod\limits_{\ell=1}^v r(X_{i_\ell})\right)\E_Q\left[\prod\limits_{\ell={v+1}}^m r(X_{i_\ell})\right]\right) \\
    &= \VAR_Q\left(\prod\limits_{\ell=1}^v r(X_{i_\ell}) \right) \\
    &= \E_Q[r(X_{i_1})^2]^v - 1.
\end{align*}
Plugging in this upper bound for $\P_Q(A_\delta^c)$ yields
\begin{equation*}
    \P_{Q}(\varphi_m(\Psi_{\texttt{DRPL}}^{r, m}(\bX_n, U)) = 1) = \frac{1}{1-\delta}\P_P(\varphi_m(X_1, \ldots, X_m)=1) + \frac{V(n,m)}{V(n,m) + \delta^2}.
\end{equation*}
Since $\delta \in (0,1)$ was arbitrary, the theorem statement follows.
\end{proof}

\subsection{Proof of \texorpdfstring{\cref{thm:uniform-level}}{}}
\begin{proof}
We adjust part of the proof of \cref{thm:asymptotic-level-SIR} to the uniform case. 
Again, let $c(n,m)$ and $d(n,m)$ be the numerator and denominator terms of \cref{eq:proof-ratio-of-averages}, i.e.,
\begin{align*}
    c(n,m) &\coloneqq \frac{1}{\tfrac{n!}{(n-m)!}}\sum\limits_{\substack{(i_1, \ldots, i_m) \\ \text{distinct}}}
        \left(\prod\limits_{\ell=1}^m \bar{r}(X_{i_\ell})\right)\mathds{1}_{\{\varphi_m(X_{i_1}, \ldots, X_{i_m}) = 1\}}, \\
        d(n,m) &\coloneqq \frac{1}{\tfrac{n!}{(n-m)!}}\sum\limits_{\substack{(j_1, \ldots, j_m) \\ \text{distinct}}}\prod\limits_{\ell=1}^m \bar{r}(X_{j_\ell}).
\end{align*}
We want to show that $\limsup_{n\rightarrow\infty}\sup_{Q\in\tau^{-1}(H_0)}\E_Q\left[\frac{c(n,m)}{d(n,m)}\right]\leq\alpha_{\varphi}$. To see this, define for all $\delta>0$ the set $A_{\delta}\coloneqq\{|d(n,m)-1|\leq\delta\}$, and take any $P \in H_0$ and $Q \in \tau^{-1}(\{P\})$. It holds for all $\delta\in (0,1)$ that
\begin{align*}
    \E_Q\left[\frac{c(n,m)}{d(n,m)}\right]
    &=\E_Q\left[\frac{c(n,m)}{d(n,m)} \mathds{1}_{A_{\delta}}\right]+\E_Q\left[\frac{c(n,m)}{d(n,m)} \mathds{1}_{A_{\delta}^c}\right]\\
    &\leq\E_Q\left[\frac{c(n,m)}{1-\delta} \mathds{1}_{A_{\delta}}\right]+\E_Q\left[\frac{c(n,m)}{d(n,m)} \mathds{1}_{A_{\delta}^c}\right]\\
    &\leq\E_Q\left[\frac{c(n,m)}{1-\delta}\right]+\P_Q\left(A_{\delta}^c\right)\\
    &=\frac{1}{1-\delta}\P_P(\varphi_m(X_1, \ldots, X_m)=1)+\P_Q\left(A_{\delta}^c\right),
\end{align*}
where we used that $\frac{c(n,m)}{d(n,m)}\leq 1$ and that $\E_Q\left[c(n,m)\right] = \P_P(\varphi_m(X_1, \ldots, X_m)=1)$, as shown in \cref{lemma:sum-distinct-weights} (a). Further, given the uniform bound on the weights, combining Chebyshev's inequality with \cref{lemma:sum-distinct-weights} (b) and (d) leads to $\lim_{n\rightarrow\infty}\sup_{Q\in\tau^{-1}(H_0)}\P_Q\left(A_{\delta}^c\right)=0$. Hence, using that $\varphi$ has uniform asymptotic level $\alpha_{\varphi}$ we have shown for all $\delta\in (0,1)$ that
\begin{equation*}
    \limsup_{n\rightarrow\infty}\sup_{Q\in\tau^{-1}(H_0)}\E_Q\left[\frac{c(n,m)}{d(n,m)}\right]\leq \frac{1}{1-\delta}\alpha_{\varphi}.
\end{equation*}
Using that $\delta\in(0,1)$ is arbitrary, completes the proof of \cref{thm:uniform-level}.
\end{proof}

\subsection{Proof of \texorpdfstring{\cref{cor:asymptotic-level-SIR-REPL}}{}}
\begin{proof}
    We have
\begin{align*}
    &\P_Q(\varphi_m(\Psi_{\texttt{REPL}}^{r, m}(\bX_n, U)) = 1)\nonumber\\
    &\quad= \P_Q(\varphi_m(\Psi_{\texttt{REPL}}^{r, m}(\bX_n, U))= 1\mid \Psi_{\texttt{REPL}}^{r, m}(\bX_n, U) \,\text{distinct})\P_Q(\Psi_{\texttt{REPL}}^{r, m}(\bX_n, U) \,\text{distinct})\\
    &\qquad + \P_Q(\varphi_m(\Psi_{\texttt{REPL}}^{r, m}(\bX_n, U)) = 1\mid \Psi_{\texttt{REPL}}^{r, m}(\bX_n, U) \,\text{not distinct})\P_Q(\Psi_{\texttt{REPL}}^{r, m}(\bX_n, U) \,\text{not distinct}). 
\end{align*}
This converges to the same limit as $\P_Q(\varphi_m(\Psi_{\texttt{REPL}}^{r, m}(\bX_n, U))= 1\mid \Psi_{\texttt{REPL}}^{r, m}(\bX_n, U) \,\text{distinct})$ 
(because $\P_Q(\Psi_{\texttt{REPL}}^{r, m}(\bX_n, U) \,\text{distinct})$ converges to $1$, see \cref{prop:repl-becomes-dist}), which, as we argue in \cref{sec:sampling-DRPL}, equals $\P_Q(\varphi_m(\Psi_{\texttt{DRPL}}^{r, m}(\bX_n, U))= 1)$. The result then follows from \cref{thm:asymptotic-level-SIR}.
\end{proof}

\subsection{Proof of \texorpdfstring{\cref{prop:repl-becomes-dist}}{}} \label{sec:proofProp1}
\begin{proof}
When sampling with replacement, $\Psi_{\texttt{REPL}}^{r, m}(\bX_n, U)$ contains non-distinct draws with positive probability (assuming, wlog, that $m$ is not $1$).
Yet, we show that $\P_Q(\Psi_{\texttt{REPL}}^{r, m}(\bX_n, U) \,\,\text{distinct})$ approaches $1$ as $m \rightarrow \infty$.
By assumption $p = \tau(q)$, so $p(x) \propto r(x)q(x)$. 
Let $\bar{r}$ be the normalized version of $r$ satisfying, for all $x$, $p(x) = \bar{r}(x) q(x)$. 
The probability $w_{(i_1, \ldots, i_m)}$ of drawing a sequence $X_{i_1}, \ldots, X_{i_m}$ is defined by \cref{eq:SIR-weights} as the product of weights $r$:
\begin{align*}
    w_{(i_1, \ldots, i_m)} = \frac{\prod_{\ell=1}^m r(X_{i_\ell})}{\sum_{(j_1, \ldots, j_m)}\prod_{\ell=1}^m r(X_{j_\ell})} = \frac{\prod_{\ell=1}^m \bar{r}(X_{i_\ell})}{\sum_{(j_1, \ldots, j_m)}\prod_{\ell=1}^m \bar{r}(X_{j_\ell})},
\end{align*}
where the sum over $(j_1, \ldots, j_m)$ in the denominator is over all sequences of length $m$ (including distinct and non-distinct sequences).
The probability of drawing a non-distinct sequence equals the sum of the weights corresponding to all non-distinct sequences $w_{(i_1, \ldots, i_m)}$. 
Therefore,
    \begin{align}
        &\P_Q(\Psi_{\texttt{REPL}}^{r, m}(\bX_n, U) \,\text{not distinct}) \nonumber \\
        &= \E_Q\left[\sum\limits_{\substack{(i_1, \ldots, i_m) \\ \text{not distinct}}} w_{(i_1, \ldots, i_m)}\right]
        \nonumber \\
        &= \E_Q\left[\sum\limits_{\substack{(i_1, \ldots, i_m) \\ \text{not distinct}}} \frac{\prod_{\ell=1}^m \bar{r}(X_{i_\ell})}{\sum_{(j_1, \ldots, j_m)}\prod_{\ell=1}^m \bar{r}(X_{j_\ell})}\right]
        \nonumber \\
        &= \E_Q\left[ \frac{\sum\limits_{\substack{(i_1, \ldots, i_m) \\ \text{not distinct}}} \prod\limits_{\ell=1}^m \bar{r}(X_{i_\ell})}{\sum\limits_{(i_1, \ldots, i_m)} \prod\limits_{\ell=1}^m \bar{r}(X_{i_\ell})} \right] \nonumber \\
        &= \E_Q\left[ \frac{\frac{1}{n^m}\sum\limits_{\substack{(i_1, \ldots, i_m) \\ \text{not distinct}}} \prod\limits_{\ell=1}^m \bar{r}(X_{i_\ell})}{\frac{1}{n^m}\sum\limits_{\substack{(i_1, \ldots, i_m) \\ \text{not distinct}}} \prod\limits_{\ell=1}^m \bar{r}(X_{i_\ell}) + \frac{1}{n^m}\sum\limits_{\substack{(i_1, \ldots, i_m) \\ \text{distinct}}}\prod\limits_{\ell=1}^m \bar{r}(X_{i_\ell})} \right]
        \label{eq:P-distinct}.
    \end{align}
    Observe that this expectation is taken both over %
    $\bX_n$ and $U$.
    By \cref{lemma:sum-nondistinct-weights}, the numerator of \cref{eq:P-distinct} (which equals the first term in the denominator) converges to~$0$ in ${L}^1$. The second term in the denominator converges in probability to $1$ by \cref{lemma:sum-distinct-weights} (this requires \cref{assump:finite-second-moment}, which is implied by 
    \cref{assump:bounded-weights}); 
    thus, the entire denominator converges to $1$ in probability. By Slutsky's lemma, the entire fraction (inside the mean) converges to $0$ in probability. 
    Since the fraction is lower bounded by $0$ and upper bounded by $1$, convergence in probability implies convergence of the mean (see the proof of \cref{thm:asymptotic-level-SIR} for an argument for this), and it follows that $\P_Q(\Psi_{\texttt{REPL}}^{r, m}(\bX_n, U) \,\text{not distinct})\rightarrow 0$.
\end{proof}

\begin{lemma}[Non-distinct draws]\label{lemma:sum-nondistinct-weights}
Let $P\in \cP$ and $Q\in\cQ$ be distributions with densities $p$ and $q$ with respect to a dominating measure $\mu$. Let $\bar{r}:\mathcal{X}\rightarrow (0,\infty)$ satisfy for all $x\in\mathcal{X}$ that $p(x)=\bar{r}(x)q(x)$. Then, under \cref{assump:m-rate-n,assump:bounded-weights} it holds that
    \begin{align*}
        \lim_{n\rightarrow\infty}\E_Q\left[\frac{1}{n^m}\sum\limits_{\substack{(i_1, \ldots, i_m) \\ \text{not distinct}}}\prod_{\ell=1}^m \bar{r}(X_{i_\ell})\right]=0.
    \end{align*}
    In particular, since the integrand is non-negative, this implies that $$\frac{1}{n^m}\sum\limits_{\substack{(i_1, \ldots, i_m) \\ \text{not distinct}}}\prod_{\ell=1}^m \bar{r}(X_{i_\ell}) \stackrel{{L}^1}{\longrightarrow} 0 \quad \text{as }n\rightarrow\infty.$$
\end{lemma}
\begin{proof}
    We first rewrite the sum using the number $k$ of distinct draws, i.e., we consider cases, in which there are $k$ distinct elements among $i_1, \ldots, i_m$. The number $k$ is at least $1$ and, since not all draws are distinct, at most $m-1$. For fixed $k$, we then further sum over the numbers $r_1, \ldots, r_k$ of occurrences of each index, i.e., $j_\ell$ appears $r_\ell$ times. 
    \begin{align*}
        \sum_{\substack{(i_1, \ldots, i_m) \\ \text{not distinct}}}\prod_{\ell=1}^m \bar{r}(X_{i_\ell})
        = \sum_{k=1}^{m-1} \sum_{\substack{(i_1, \ldots, i_m) \text{ with } k \text{ distinct entries} \\ (\text{i.e. }j_\ell \in (i_1, \ldots, i_m) \text{ appears } r_\ell > 0 \text{ times,}\\ r_1+\cdots+r_k=m)}} 
        \prod_{\ell=1}^k \bar{r}(X_{j_\ell})^{r_{\ell}}.
    \end{align*}
    Using the independence across distinct observations, this implies that 
    \begin{equation}
        \E_Q\left[\sum\limits_{\substack{(i_1, \ldots, i_m) \\ \text{not distinct}}}\prod_{\ell=1}^m \bar{r}(X_{i_\ell})\right]
        = \sum_{k=1}^{m-1} \sum_{\substack{(i_1, \ldots, i_m) \text{ with } k \text{ distinct entries} \\ (\text{i.e. }j_\ell \in (i_1, \ldots, i_m) \text{ appears } r_\ell > 0 \text{ times,}\\ r_1+\cdots+r_k=m)}} \prod_{\ell=1}^k \E_Q\left[\bar{r}(X_{j_\ell})^{r_{\ell}}\right].\label{eq:non-distinct-partition-exp}
    \end{equation}
    We now use the uniform bound on the weights given in \cref{assump:bounded-weights} and the fact that $\E_Q[\bar{r}(X_i)^t]=\int \bar{r}(x_i) q(x_i) \bar{r}(x_i)^{t-1} \mathrm{d}\mu(x_i)=\int p(x_i) \bar{r}(x_i)^{t-1} \mathrm{d}\mu(x_i) = \E_P[\bar{r}(X_i)^{t-1}]$ to get for all $i\in\{1,\ldots,n\}$ and all $t\in\{1,\ldots, m-1\}$ that
    \begin{equation*}
        \E_Q[\bar{r}(X_i)^t] = \E_P[\bar{r}(X_i)^{t-1}] \leq L^{t-1}.
    \end{equation*}
    Together with \eqref{eq:non-distinct-partition-exp} this results in
    \begin{align}
        \E\left[\sum\limits_{\substack{(i_1, \ldots, i_m) \\ \text{not distinct}}}\prod_{\ell=1}^m \bar{r}(X_{i_\ell})\right]
        &\leq \sum_{k=1}^{m-1} \sum_{\substack{(i_1, \ldots, i_m) \text{ with } k \text{ distinct entries} \\ (\text{i.e. }j_\ell \in (i_1, \ldots, i_m) \text{ appears } r_\ell > 0 \text{ times,}\\ r_1+\cdots+r_k=m)}} L^{m-k} \nonumber \\
        &= \sum_{k=1}^{m-1} \binom{n}{k} \pi(m, k) L^{m-k}\nonumber\\ 
        &\leq \sum_{k=1}^{m-1} \binom{n}{k} \tilde{\pi}(m, k) L^{m-k}, \label{eq:non-distinct-partition-exp2}
    \end{align}
    where $\pi(m, k)$ is the number of words of length $m$ using $k$ letters such that each letter is used at least once and 
    \begin{equation*}
        \tilde{\pi}(m,k)\coloneqq k^{m-k}\frac{m!}{(m-k)!}.
    \end{equation*}
    The last inequality holds because we have $\pi(m,k) \leq \tilde{\pi}(m,k)$: Consider constructing a word of length $m$ by first distributing one of each of the $k$ letters (ensuring that each letter is used at least once) among the $m$ positions, which can be done in $m!/(m-k)!$ ways. For the remaining $m-k$ positions pick any combination of letters, which can be done in $k^{m-k}$ ways. In total, this two-step procedure has $\tilde{\pi}(m,k)$ possible outcomes. This enumeration contains all words of length $m$ using $k$ letters such that each is used at least once, so $\pi(m,k) \leq \tilde{\pi}(m,k)$. We do not have equality, because $\tilde{\pi}$ counts some words several times, but with different intermediate steps. For example if $k = 2$ and $m=3$, $\tilde{\pi}$ counts $(a, \rule{0.2cm}{0.15mm}, b) + (\rule{0.2cm}{0.15mm}, a, \rule{0.2cm}{0.15mm})$ and $(\rule{0.2cm}{0.15mm}, a, b) + (a, \rule{0.2cm}{0.15mm}, \rule{0.2cm}{0.15mm})$ as two distinct words, although they both yield $(a, a, b)$; indeed $\pi(3,2) = 6$ and $\tilde{\pi}(3,2) = 12$.
    
    Then, with $s_k\coloneqq \binom{n}{k}\tilde{\pi}(m,k)L^{m-k}$ it holds that
    \begin{align*}
        \frac{s_{k+1}}{s_k} &= 
        \frac{\binom{n}{k+1}\tilde{\pi}(m, k+1) L^{m - k -1}}{\binom{n}{k}\tilde{\pi}(m, k)L^{m - k}} \\
        &= \frac{1}{L}\frac{n-k}{k+1}\frac{m-k}{1}\frac{(k+1)^{m-k-1}}{k^{m-k}} \\
        &= \frac{1}{L}\frac{n-k}{k+1}\frac{m-k}{k+1}\left(\frac{k+1}{k}\right)^{m-k} \\
        &\geq \frac{1}{L}\frac{n-m+1}{m^2}\\
        &=: c,
    \end{align*}
    where the inequality follows by using $k \leq m-1$ and $(k+1)/k \geq 1$.
    By \cref{assump:m-rate-n} (i.e., $m=o(\sqrt{n})$) it holds for $n$ sufficiently large that $c>1$.
    Iterating this inequality, we get (again for $n$ sufficiently large) that $s_k\leq c^{-(m-1-k)}s_{m-1}$, which we can plug into \eqref{eq:non-distinct-partition-exp2} to get
    \begin{equation}
    \label{eq:last_estimate}
        \E\left[\sum\limits_{\substack{(i_1, \ldots, i_m) \\ \text{not distinct}}}\prod_{\ell=1}^m \bar{r}(X_{i_\ell})\right]
        \leq \sum_{k=1}^{m-1}s_k
        \leq s_{m-1}\sum_{k=1}^{m-1}c^{-(m-1-k)}
        = s_{m-1}\sum_{k=0}^{m-2}c^{-k}
        \leq s_{m-1}\frac{1}{1-\tfrac{1}{c}}. 
    \end{equation}
    In the last inequality, we use the trivial bound $\sum_{k=0}^{m-2} c^{-k} < \sum_{k=0}^{\infty} c^{-k}$ 
    and $0<c^{-1}<1$.
    Finally, observe that 
    \begin{align*}
        n^{-m}s_{m-1} &=n^{-m} \binom{n}{m-1}\tilde{\pi}(m, m-1)L \\
        &= n^{-m}\frac{n!}{(n-(m-1))!(m-1)!}(m-1)^{m-(m-1)}\frac{m!}{(m-(m-1))!} L \\
        &= n^{-m} \frac{n!}{(n-m)!}\frac{m(m-1)}{(n-m+1)} L \\
        &= \underbrace{g(n,m)}_{\coloneqq n^{-m}\frac{n!}{(n-m)!}} \frac{m(m-1)}{(n-m+1)} L,
    \end{align*}
   which by \cref{lemma:convergence-m-n} converges to zero (by the assumption $m = o(\sqrt{n})$). 
    Therefore, we have that
    \begin{align*}
        \E\left[n^{-m}\sum\limits_{\substack{(i_1, \ldots, i_m) \\ \text{not distinct}}}\prod_{\ell=1}^m \bar{r}(X_{i_\ell})\right] \leq n^{-m} s_{m-1} \frac{1}{1 - \tfrac{1}{c}} \rightarrow 0,
    \end{align*}
    which completes the proof of \cref{lemma:sum-nondistinct-weights}.
\end{proof}

\begin{lemma}
\label{lemma:convergence-m-n}
    Define for all $n,m\in\mathbb{N}$ the function
    \begin{equation*}
        g(n,m)\coloneqq \frac{n!}{(n-m)!}n^{-m}.
    \end{equation*}
    Then, it holds that
    \begin{align*}
        \lim_{n\rightarrow\infty}g(n,n^q)=
    \begin{cases}
    0\quad&\text{if }q\in(\tfrac{1}{2}, 1)\\
    \exp(-\tfrac{1}{2})\quad&\text{if }q=\tfrac{1}{2}\\
    1\quad&\text{if }q\in[0, \tfrac{1}{2}).
    \end{cases}
    \end{align*}
\end{lemma}
\begin{proof}
First, apply the Stirling approximation to get for $n$
sufficiently large that
\begin{align*}
    g(n,m)
    &\sim n^{n+\frac{1}{2}}\cdot e^{-n}\cdot (n-m)^{m-n-\frac{1}{2}}\cdot e^{n-m}\cdot n^{-m}\\
    &=n^{n-m+\frac{1}{2}}\cdot (n-m)^{m-n-\frac{1}{2}}\cdot e^{-m}\\
    &=\exp\{(n-m+\tfrac{1}{2})\log(n)+(m-n-\tfrac{1}{2})\log(n-m)-m\}.
\end{align*}
Next, we look at cases where $m=n^q$ for some $q\in [0,1)$. The above expression can then be simplified further as\begin{align*}
    g(n,n^q)
    &\sim\exp\{(n-n^q+\tfrac{1}{2})\log(n)+(n^q-n-\tfrac{1}{2})\log(n-n^q)-n^q\}\\
    &=\exp\{(n-n^q+\tfrac{1}{2})\log(n)+(n^q-n-\tfrac{1}{2})[\log(n)+\log(1-n^{q-1})]-n^q\}\\
    &=\exp\{(n^q-n-\tfrac{1}{2})\log(1-n^{q-1})-n^q\}.
\end{align*}
Finally, since $n^{q-1}\rightarrow 0$ as $n$ goes to infinity we can use the following Taylor expansion
\begin{equation*}
    \log(1-n^{q-1})=-n^{q-1}-\tfrac{1}{2}n^{2(q-1)}+O(n^{3(q-1)}),
\end{equation*}
which results in
\begin{align*}
    g(n,n^q)
    &\sim\exp\{(n^q-n-\tfrac{1}{2})\log(1-n^{q-1})-n^q\}\\
    &=\exp\{(n^q-n-\tfrac{1}{2})(-n^{q-1}-\tfrac{1}{2}n^{2(q-1)}+O(n^{3(q-1)}))-n^q\}\\
    &=\exp\{-n^{2q-1}-\tfrac{1}{2}n^{3q-2}+n^q+\tfrac{1}{2}n^{3q-2}+\tfrac{1}{2}n^{q-1}+\tfrac{1}{4}n^{2q-2}+O(n^{2q-1})-n^q\}\\
    &=\exp\{-\tfrac{1}{2}n^{2q-1}+O(n^{3q-2})\}.
\end{align*}
From this we see that
\begin{equation*}
\lim_{n\rightarrow\infty}g(n,n^q)=
    \begin{cases}
    0\quad&\text{if }q\in(\tfrac{1}{2}, 1)\\
    \exp(-\tfrac{1}{2})\quad&\text{if }q=\tfrac{1}{2}\\
    1\quad&\text{if }q\in[0, \tfrac{1}{2}).
    \end{cases}
\end{equation*}
This completes the proof of \cref{lemma:convergence-m-n}.
\end{proof}

\subsection{Proof of \texorpdfstring{\cref{cor:hypothesis-obs-space-short}}{}}\label{sec:hypothesis-obs-space}
As discussed in \cref{sec:test-target-to-observable-domain}, the proposed procedure in \cref{subsec:method} can also be used to construct a test for a hypothesis $H_0^\mathcal{Q}$ in the observed domain, satisfying the same theoretical guarantees.
\begin{corollary}[Pointwise level in the observed domain - detailed version] \label{cor:hypothesis-obs-space}
Consider hypotheses $H_0^{\mathcal{Q}} \subseteq \mathcal{Q}$ and $H_0^{\mathcal{P}} \subseteq \mathcal{P}$ in the observational and in the target domain, respectively.
Let $\tau:\cQ \rightarrow \cP$ be a distributional shift for which there exist a known map $r:\mathcal{X}\rightarrow(0,\infty)$ 
and a set $A$ 
satisfying
for all $q\in\cQ$ and all $x\in\mathcal{Z}$  that 
$\tau(q)(x)\propto r(x^A)q(x)$, see~\eqref{eq:tauform}. Assume 
$\tau(H_0^{\mathcal{Q}}) \subseteq H_0^{\mathcal{P}}$.
Let $\varphi_k$ be a sequence of tests for $H_0^{\mathcal{P}}$ with pointwise asymptotic level $\alpha_\varphi$. Let $m=m(n)$ be a resampling size and let $\psi^r_n$ be the DRPL-based resampling test defined by $\psi^r_n(\bX_n, U) \coloneqq \varphi_m(\Psi_{\texttt{DRPL}}^{r,m}(\bX_n, U))$, see \cref{alg:resampling-and-testing}. Then, if $m$ satisfies \cref{assump:m-rate-n} and all $Q\in H_0^{\mathcal{Q}}$ satisfy \cref{assump:finite-second-moment}, it holds that
$$
\sup_{Q \in H_0^{\mathcal{Q}}} \limsup_{n \rightarrow \infty} \mathbb{P}_Q(\psi^r_n(\mathbf{X}_n, U) = 1) \leq \alpha_\varphi,
$$
i.e., $\psi^r_n$ satisfies pointwise asymptotic level $\alpha$ for the hypothesis $H_0^{\mathcal{Q}}$.
\end{corollary}
Clearly, the condition 
$H_0^{\cQ} \subseteq \tau(H_0^{\cP})$
is satisfied when $H_0^{\mathcal{P}} = \tau(H_0^{\mathcal{Q}})$. This is the case for the conditional independence test described in Section~\ref{sec:conditionaltesting}, for example.
\begin{proof}
We have
\begin{align*}
    \tau(H_0^{\mathcal{Q}}) \subseteq H_0^{\cP} 
\;     \Rightarrow \;  H_0^{\cQ} \subseteq \tau^{-1}(H_0^{\cP})
\end{align*}
and therefore
\begin{align*}
\sup_{Q \in H_0^{\mathcal{Q}}} \limsup_{n \rightarrow \infty} \mathbb{P}_Q(\psi^r_n = 1)
\leq \sup_{Q \in \tau^{-1}(H_0^{\mathcal{P}})} \limsup_{n \rightarrow \infty} \mathbb{P}_Q(\psi^r_n = 1).
\end{align*}
Since \cref{assump:finite-second-moment} is satisfied for all $Q\in H_0^{\mathcal{Q}}$, the statement follows from Theorem~\ref{thm:asymptotic-level-SIR}. This completes the proof of \cref{cor:hypothesis-obs-space-short}.
\end{proof}

\subsection{Proof of \texorpdfstring{\cref{prop:rejection-sampler}}{}}
\begin{proof}
    We analyze the output of \cref{alg:rejection-sampler}. 
    For each $i \in \{1, \ldots, n\}$, we discard $X_i$ if $U_i > \tfrac{r(X_i)}{M}$, where $U_i$ is uniform on $(0,1)$. The probability of the event $E_i$ that $X_i$ is not discarded equals
     \begin{equation*}
        q(E_i) = \int q(E_i|x_i)q(x_i) \textrm{d}x_i = \int \tfrac{r(x_i)}{M}q(x_i) \textrm{d}x_i
        = \int \tfrac{c}{M}\tfrac{p(x_i)}{q(x_i)}q(x_i) \textrm{d}x_i
        = \tfrac{c}{M},
    \end{equation*}
    where $c$ is a constant such that $r(x)q(x)=cp(x)$.
    and the conditional density of $X_i$ given that the sample is not discarded, $q(x_i|E_i)$, is given by
    \begin{align*}
        q(x_i|E_i) = \frac{q(x_i)}{q(E_i)}q(E_i|x_i) 
        = q(x_i)\frac{M}{c}\frac{c}{M}\frac{p(x_i)}{q(x_i)} 
        = p(x_i)
    \end{align*}
    If $X_{i_1}, \ldots, X_{i_m}$ are the points that are not discarded, this means that $(X_{i_1}, \ldots, X_{i_m})$ is distributed as if it was $m$ i.i.d.\ draws from $P^*$ (where $m$ is random). 
    In particular, by the assumption that the probability that for all $k\in\Z$: $\P_P(\varphi_k(\bZ_k) = 1) = \alpha_\varphi$ for $k$ i.i.d. samples $\bZ_k$ from $\P_P$, it follows that $\P_Q(\psi_n^r(\bX_n, U) = 1) = \P_Q(\varphi(X_{i_1}, \ldots, X_{i_m})) = \alpha_\varphi$.
\end{proof}

\section{Analyzing \texorpdfstring{\cref{assump:finite-second-moment}}{} in a linear Gaussian model}
\label{sec:assumption-a3-gaussian}
In this section, we show conditions for assumption \cref{assump:finite-second-moment} to be satisfied when we consider the shift that changes a Gaussian conditional into a marginal, independent Gaussian target distribution. 
\begin{proposition}
Consider a linear Gaussian setting where $Y = X + \epsilon$ with $\epsilon \sim \mathcal{N}(0, \sigma_\epsilon^2)$ and $X\sim\mathcal{N}(0,\sigma_X^2)$, with $\sigma_\epsilon, \sigma_X$ known. Assume that we are interested in the distributional shift that replaces the conditional $q(y|x)$ (an $\mathcal{N}(x, \sigma_\epsilon^2)$-density, evaluated at $y$) with an independent $\mathcal{N}(0,\sigma^2)$ target distribution $p(y)$. Formally, define the shift factor $r$ for all $x,y\in\R$ as
\begin{equation*}
    r(x,y)=\frac{p(y)}{q(y|x)}.
\end{equation*}
Then \cref{assump:finite-second-moment} is satisfied for $Q$ if and only if
\begin{equation*}
    \sigma^2<2(\sigma_{\epsilon}^2-\sigma_X^2).
\end{equation*}
\end{proposition}

\begin{proof}
We begin by directly expanding the second moment of the factor $r$ under the observational distribution $Q$ as follows,
 \begin{align}
    \E_Q\left[r(X,Y)^2\right] 
    &=\E_Q\left[\left(\frac{p(Y)}{q(Y| X)}\right)^2\right]\nonumber\\
    &= \frac{\sigma_\epsilon^2}{\sigma^2} \E_Q\left[\exp\left(\left(\frac{Y - X}{\sigma_\epsilon}\right)^2-\left(\frac{Y}{\sigma}\right)^2\right)\right] \nonumber\\
    &= \frac{\sigma_\epsilon^2}{\sigma^2} \E_Q\left[\exp\left(\left(\frac{\epsilon}{\sigma_\epsilon}\right)^2-\left(\frac{X + \epsilon}{\sigma}\right)^2\right)\right] \nonumber\\
    &= \frac{\sigma_\epsilon^2}{\sigma^2} \E_Q\left[\exp\left(\epsilon^2\left(\frac{1}{\sigma_\epsilon^2} - \frac{1}{\sigma^2}\right)-\frac{X^2}{\sigma^2} - \frac{2X \epsilon}{\sigma^2}\right)\right]\nonumber \\
    &= \frac{\sigma_\epsilon^2}{\sigma^2}\E_Q\left[\exp\left(-\frac{X^2}{\sigma^2}- \frac{2\epsilon }{\sigma^2}X-\frac{\epsilon^2}{\sigma^2}+\frac{\epsilon^2}{\sigma^2}\right)\exp\left(\frac{\epsilon^2}{\sigma_\epsilon^2} - \frac{\epsilon^2}{\sigma^2}\right)\right] \nonumber\\
    &= \frac{\sigma_\epsilon^2}{\sigma^2}\E_Q\left[\exp\left(-\frac{\sigma_X^2}{\sigma^2}W+\frac{\epsilon^2}{\sigma^2}\right)\exp\left(\frac{\epsilon^2}{\sigma_\epsilon^2} - \frac{\epsilon^2}{\sigma^2}\right)\right] \nonumber\\
    &=\frac{\sigma_\epsilon^2}{\sigma^2} \E_Q\left[\exp\left(-\frac{\sigma_X^2}{\sigma^2}W\right)\exp\left(\frac{\epsilon^2}{\sigma_\epsilon^2}\right)\right],\label{eq:part1A3}
\end{align}
where $W\coloneqq(X/\sigma_X+\epsilon/\sigma_X)^2$. Next, observe that, conditioned on $\epsilon$, $W$ has a non-central $\chi^2_{(1)}$-distribution with non-centrality parameter $\epsilon^2/\sigma_X^2$.
The moment generating function of $W$ is given by $M_{W}(t)=(1-2t)^{-1/2}\exp\left(\frac{\epsilon^2}{\sigma_X^2}\frac{t}{1-2t}\right)$ 
for all $t<1/2$. Hence, continuing the computation in \eqref{eq:part1A3} and by conditioning on $\epsilon$, we get 
\begin{align*}
    \E_Q\left[r(X,Y)^2\right]
        &=\frac{\sigma_\epsilon^2}{\sigma^2}\E_Q\left[\E_Q\left[\exp\left(-\frac{\sigma_X^2}{\sigma^2}W\right)\Big\vert \epsilon\right]\exp\left(\frac{\epsilon^2}{\sigma_\epsilon^2}\right)\right]\\
        &=\frac{\sigma_\epsilon^2}{\sigma^2}\E_Q\left[M_{W}\left(-\frac{\sigma_{X}^2}{\sigma^2}\right)\exp\left(\frac{\epsilon^2}{\sigma_\epsilon^2}\right)\right]\\
        &=\frac{\sigma_\epsilon^2}{\sigma^2}\left(1+2\frac{\sigma_X^2}{\sigma^2}\right)^{-\frac{1}{2}}\E_Q\left[\exp\left(\frac{\epsilon^2}{\sigma_X^2}\frac{-\sigma_X^2}{\sigma^2} \frac{1}{1+2\frac{\sigma_X^2}{\sigma^2}}\right)\exp\left(\frac{\epsilon^2}{\sigma_\epsilon^2}\right)\right]\nonumber\\
        &=    \frac{\sigma_\epsilon^2}{\sigma^2}\left(1+2\frac{\sigma_X^2}{\sigma^2}\right)^{-\frac{1}{2}}\E_Q\left[\exp\left(\frac{\epsilon^2}{\sigma_\epsilon^2}\left(1 - \frac{\sigma_\epsilon^2}{\sigma^2+2\sigma_X^2} \right) \right)\right]\nonumber\\ 
        & =\frac{\sigma_\epsilon^2}{\sigma^2}\left(1+2\frac{\sigma_X^2}{\sigma^2}\right)^{-\frac{1}{2}}M_S\left(1 - \frac{\sigma_\epsilon^2}{\sigma^2+2\sigma_X^2} \right),
\end{align*}
where $S \coloneqq (\epsilon/\sigma_\epsilon)^2$ and $M_S$ is the moment generating function of a (central) $\chi_{(1)}^2$ distribution. $M_S(t)$ is finite if and only if $t < 1/2$, corresponding to $1 - \frac{\sigma_\epsilon^2}{\sigma^2+2\sigma_X^2} < 1/2$ which is equivalent to $\sigma^2 < 2(\sigma_\epsilon^2 - \sigma_X^2)$.
\end{proof}

\end{document}